\documentclass[11pt]{article}

\usepackage[a4paper, margin=1.25in]{geometry}
\usepackage{fourier} % nicer font
\usepackage{booktabs} % nicer tables
\usepackage[T1]{fontenc} % polish letters
\usepackage{authblk} % author blocks
\usepackage{footmisc}
\usepackage{microtype} % for not going over margins
\usepackage{balance}

\usepackage[dvipsnames]{xcolor} % for more colors
\definecolor{orcidgreen}{HTML}{85A12C} % darker green than orcid
\usepackage{url}
\usepackage[hidelinks]{hyperref} % for active links
\hypersetup{
    colorlinks=true,
    allcolors=orcidgreen,
    linkcolor=Blue,
    citecolor=purple,
}

\usepackage{amssymb}
\usepackage{amsmath}
\usepackage{amsthm}
\usepackage{orcidlink}
\usepackage{academicons}
\usepackage{xspace}
\usepackage{nicefrac}

\newcommand{\reals}{\mathbb{R}}
\newcommand{\pref}{\succ}

\newcommand{\out}{{\mathrm{out}\hbox{-}\mathrm{div}}}
\newcommand{\ansd}{{\mathrm{ansd}}}
\newcommand{\swap}{{\mathrm{swap}}}

\newcommand{\kem}{{\mathrm{kem}}}
\newcommand{\pos}{{\mathrm{pos}}}
\newcommand{\pop}{{\mathrm{pop}}}
\newcommand{\npop}{{\mathrm{npop}}}
\newcommand{\calL}{{\mathcal{L}}}

\newcommand{\calT}{{T}} % {{\mathcal{T}}}
\newcommand{\kkemscore}{{{k\hbox{-}\mathrm{kem}}}}
\newcommand{\kemscore}{{{\mathrm{kem}}}}

\newcommand{\cent}{{\mathrm{center}}}
\newcommand{\far}{{\mathrm{far}}}
\newcommand{\rev}{{\mathrm{rev}}}
\newcommand{\UN}{{\mathrm{UN}}}

\newcommand{\fp}{{\mathrm{FP}}}
\newcommand{\kc}{{\mathrm{K1C}}}

\newcommand{\np}{{{\mathrm{NP}}}}

\newcommand{\SP}{{\mathrm{SP}}}
\newcommand{\SPOC}{{\mathrm{SPOC}}}
\newcommand{\SPDF}{{\mathrm{SP/DF}}}
\newcommand{\SC}{{\mathrm{SC}}}
\newcommand{\GS}{{\mathrm{GS}}}
\newcommand{\GSbal}{{\mathrm{GS/bal}}}
\newcommand{\GScat}{{\mathrm{GS/cat}}}
\newcommand{\Int}{{\mathrm{1D\hbox{-}Int.}}}
\newcommand{\Square}{{\mathrm{2D\hbox{-}Square}}}
\newcommand{\Cube}{{\mathrm{3D\hbox{-}Cube}}}

\newcommand{\LC}{{\mathrm{LC}}}

\usepackage{xspace}

\DeclareSymbolFont{bbold}{U}{bbold}{m}{n}
\DeclareSymbolFontAlphabet{\mathbbold}{bbold}

\newcommand{\probName}[1]{\textsc{#1}\xspace}

\newcommand{\kemenycenter}{\probName{Kemeny $1$-Center}}
\newcommand{\farthestpermutation}{\probName{Farthest Permutation}}

\usepackage{graphicx}
\usepackage{tikz}
\usetikzlibrary{calc}
\usepackage{subcaption}

\usepackage{algorithm}
\usepackage[noend]{algorithmic}

\usepackage[square]{natbib}
\usepackage[nameinlink]{cleveref} % load after algorithm2e hyperref, amsmath

\theoremstyle{plain}
\newtheorem{theorem}{Theorem}[section]
\newtheorem{lemma}[theorem]{Lemma}
\newtheorem{proposition}[theorem]{Proposition}
\newtheorem{corollary}[theorem]{Corollary}

\newtheorem{definition}{Definition}[section]

\newtheorem{remark}{Remark}[section]

\usepackage{thm-restate}

\title{\huge \textbf{Outer Diversity of Structured Domains}}

\author{
  Piotr Faliszewski\quad\quad
  Krzysztof Sornat\quad\quad
  Stanisław Szufa\quad\quad
  Tomasz Wąs
}

\date{}

\makeatletter
\renewcommand*{\@fnsymbol}[1]{\ifcase#1\or i\or ii\or iii\or iv\else\@ctrerr\fi}
\makeatother

\begin{document}

\maketitle

\begin{quote}
    \textbf{Abstract:}
    An ordinal preference domain is a subset of
    preference orders that the voters are allowed to cast in an election.
    We introduce and study the notion of \emph{outer diversity} of a domain and
    evaluate its value for a number of well-known structured domains, such as the
    single-peaked, single-crossing, group-separable, and Euclidean ones.
\end{quote}
\begin{quote}
    \textbf{Code:} \url{https://github.com/Project-PRAGMA/Outer-Diversity-AAMAS-2026}
\end{quote}
% Abstract: An ordinal preference domain is a subset of preference orders that the voters are allowed to cast in an election. We introduce and study the notion of outer diversity of a domain and evaluate its value for a number of well-known structured domains, such as the single-peaked, single-crossing, group-separable, and Euclidean ones.

\renewcommand{\thefootnote}{}
\footnotetext{
  \hspace{-19pt}
  Authors' Information:
  \href{https://orcid.org/0000-0002-0332-4364}{Piotr Faliszewski \orcidlink{0000-0002-0332-4364}}, faliszew@agh.edu.pl, AGH University, Poland;
  \href{https://orcid.org/0000-0001-7450-4269}{Krzysztof~Sornat~\orcidlink{0000-0001-7450-4269}}, sornat@agh.edu.pl, AGH University, Poland;
  \href{https://orcid.org/0000-0001-6301-6227}{Stanisław Szufa \orcidlink{0000-0001-6301-6227}}, s.szufa@gmail.com, University of Geneva, Geneva, Switzerland;
  \href{https://orcid.org/0000-0003-3492-6584}{Tomasz Wąs \orcidlink{0000-0003-3492-6584}}, tomasz.was@cs.ox.ac.uk, University of Oxford, United Kingdom.
}
\renewcommand{\thefootnote}{\arabic{footnote}}

%%%%%%%%%%%%%%%%%%%%%%%%% MAIN PART %%%%%%%%%%%%%%%%%%%%%%%%%%%%%

\section{Introduction}

In the standard, ordinal model of elections, each voter considers a
set of candidates and ranks them from the one that he or she likes
most to the one that he or she likes least. In principle, a voter may
order the candidates in any arbitrary way, but some of these rankings
appear more natural (or, more rational) than others.  For example, in
the political setting it would be expected that a voter would rank the
candidates with respect to their proximity to his or her political
stance, but a ranking with the most right-wing candidate and the most
left-wing one on two top positions would be surprising. Various
rationality conditions for ordinal rankings are expressed as so-called
\emph{structured domains}, i.e., sets of rankings that can be cast in
a given setting. Such domains include, e.g., the single-peaked
one~\citep{bla:b:polsci:committees-elections}, which captures
preferences based on proximity to some ideal, the single-crossing
ones, introduced in the context of
taxation~\citep{mir:j:single-crossing,rob:j:tax}, or group-separable
ones~\citep{ina:j:group-separable,ina:j:simple-majority}, where voters
derive rankings of candidates from preferences over their
features~\citep{kar:j:group-separable,fal-kar-obr:c:group-separable}.
We introduce a new measure of diversity of such domains, provide
algorithms for computing its value, and analyze diversity of a number
of structured domains.

Somewhat surprisingly, analysis of diversity for structured domains
has only recently started to receive more focused
attention~\citep{amm-pup:j:domain-diversity,kar-mar-rii-zho:t:domain-diversity,fal-sor-szu-was:c:structured-kemeny},
with a few authors also considering diversity of
elections~\citep{has-end:c:diversity-indices,fal-kac-sor-szu-was:c:div-agr-pol-map,fal-mer-nun-szu-was:c:map-dap-top-truncated}.
Two commonly used approaches are:
\begin{description}
\item[Richness Diversity.] The overarching idea is that a domain is
  diverse if it contains many different substructures in its rankings
  (these substructures are sometimes also called \emph{attributes}, as
  the approach builds on the theory of attribute diversity of
  \citet{neh-pup:j:diversity}). For example, one might consider how
  many votes appear in the domain, how many candidates are ever ranked
  on top, or---for each triple of candidates---how many ways of
  ranking these candidates appear in the domain. This approach is
  taken, e.g., by \citet{amm-pup:j:domain-diversity} and
  \citet{kar-mar-rii-zho:t:domain-diversity}
 
\item[Inner Diversity.] In this case, we say that a domain is diverse
  if its rankings do not form clear clusters. This approach was taken
  by
  \citet{fal-kac-sor-szu-was:c:div-agr-pol-map,fal-mer-nun-szu-was:c:map-dap-top-truncated,fal-sor-szu-was:c:structured-kemeny},
  who introduced the $k$-Kemeny problem to quantify the difficulty of
 clustering rankings (briefly put, one tries to optimally partition the
  rankings into a given number of groups, measuring their cohesiveness
  using the classic Kemeny rule~\citep{kem:j:no-numbers}).
\end{description}
We propose a third approach, which we refer to as \emph{outer
  diversity}:
\begin{description}
\item[Outer Diversity.] A domain is diverse if, on average, a random
  ranking from the space of all possible ones is similar to some
  ranking from the domain. In particular, we measure similarity
  between rankings using the number of swaps of adjacent candidates
  that transform one into the other.
\end{description}
Inner and outer diversity seem to capture the same basic intuition, but
the inner approach focuses on the rankings within the domain, whereas the
outer one focuses on those outside.

We believe that all the above approaches to measuring domain diversity
are meaningful and are worth studying, but outer diversity has some
advantages. First, it has a very clear interpretation: If a domain has
high diversity, then it covers the space of all possible rankings well;
if one wanted to cast a ranking from the domain but had one that did
not belong to it, then the closest member of the domain would not be
too far off from his or her original ranking.

Second, outer diversity of a given domain is a single number. On the
other hand, in case of richness diversity one has to choose from many
different substructures to count, and in case of inner diversity one
either has to choose the number of clusters to consider (for which
there is no clear solution) or somehow aggregate
obtained values for different
numbers of clusters, which is not obvious (indeed,
the works we cite with respect to inner diversity do not provide fully
satisfying recommendations).

Third, while in principle computing outer diversity may require
exponential time, we provide efficient algorithms for computing it
using sampling: Our algorithms compute the distance from a given
ranking to the closest one in a given domain of interest, such as the
single-peaked, single-crossing, and group-separable ones. Hence, we
can sample random votes, compute their distances to the domains, and
output the average of the obtained values.  On the other hand, even
the heuristics that \citet{fal-kac-sor-szu-was:c:div-agr-pol-map}
proposed for inner diversity (i.e., for $k$-Kemeny) require
exponential time if a given domain contains exponentially many
rankings (as is the case for, e.g., the single-peaked and
group-separable ones).

Our main contributions are as follows:
\begin{enumerate}
\item We introduce the notion of outer diversity and provide means of
  computing its values for a number of domains, including the
  single-peaked, single-crossing, and group-separable ones, but also
  many others (including variants of the single-peaked domain, as well
  as Euclidean domains). However, we also find that for some natural
  domains, the sampling-based approach
  requires solving an $\np$-hard problem.

\item We evaluate outer diversity across a number of domains. We find
  that ranking the domains with respect to outer diversity gives
  similar results as doing so with respect to the inner one. Further,
  while analyzing outer diversity of our domains, we note a number of
  their interesting features.

\item We compute domains of given sizes, whose outer diversity is
  (close to) the highest possible, and we analyze how close are
  various structured domains to these maximal values.
\end{enumerate}
One of the takeaway messages of our work is that the domain of
group-separable preferences based on caterpillar trees (see
\Cref{sec:prelim}) is the most diverse one among those that we study,
and has many features that other domains often lack. Consequently, and
strengthening the message of
\citet{fal-sor-szu-was:c:structured-kemeny}, we believe that this
domain should be used in numerical experiments on elections. Even if
it does not capture reality in a given setting, it is so special that
studying it may lead to the discovery of hard-to-spot phenomena.

We discuss related work throughout the paper, whenever relevant.
Omitted proofs are available in the appendix.

\section{Preliminaries}\label{sec:prelim}

For a positive integer $t$, by $[t]$ we mean the set
$\{1, 2, \ldots, t\}$.  Given an undirected graph $G$, by $V(G)$ and $E(G)$
we mean its sets of vertices and edges, respectively.  We use the
\emph{Iverson bracket} notation, i.e., for a logical formula
$\varphi$, by $[\varphi]$ we mean $1$ if $\varphi$ is true, and $0$,
otherwise. 

\subsection{Preference Orders, Domains, and Elections}
Let $C$ be a
set of $m$ \emph{candidates}.  By $\mathcal{L}(C)$ we denote the set
of all $m!$ linear orders over $C$, typically referred to as
\emph{preference orders}, \emph{votes}, or \emph{rankings}.
For each such ranking $v$
and two candidates $a,b \in C$, we write $a \pref_v b$ to indicate
that $v$ ranks $a$ ahead of $b$ (i.e., according to $v$, $a$ is
preferred to $b$).  A \emph{preference domain} (over $C$) is a subset
$D$ of $\mathcal{L}(C)$.  In particular, $\calL(C)$ is the
\emph{general domain}.
For a ranking $v$ and candidate $c$, by $\pos_v(c)$ we mean
the position of $c$ in $v$; the top candidate has position $1$, the
next one has position $2$, and so on.

An election is a pair $E = (C,V)$, where $C = \{c_1, \ldots, c_m\}$ is
a set of candidates and $V = (v_1, \ldots, v_n)$ is a collection of
voters, each of whom has a vote from $\calL(C)$.  To streamline the
discussion, we use the same symbol $v_i$ to refer both to the given
voter and to his or her vote. The exact meaning will always be clear
from the context. Given a domain $D \subseteq \calL(C)$, we say that
$E = (C,V)$ is a $D$-election if all the voters in $V$ have votes from
$D$.

Sometimes it is convenient to treat a domain $D \subseteq \calL(C)$ as
an election that contains a single voter for each of its preference
orders. In particular, we write $\UN$ to mean an election that
contains one copy of every possible order (so $\UN$ is simply
$\calL(C)$, viewed as an election). For other domains, we typically do
not introduce a second name, but $\UN$ has already been used in
preceding literature in the context of the map of
elections~\citep{szu-boe-bre-fal-nie-sko-sli-tal:j:map}.

For two rankings $u,v \in \mathcal{L}(C)$, their \emph{swap distance}%
\footnote{We focus on swap distance in this paper,
as it is standard in (computational) social choice
and it was also used in the work on inner diversity~\cite{fal-kac-sor-szu-was:c:div-agr-pol-map}.
Our initial experiments with Spearman footrule distance~\cite{elk-fal-sli:j:dr}
yield qualitatively similar outcomes.}
(also known as \emph{Kendall's $\tau$ distance}) is a number of pairs
of candidates in $C$ on whose ordering $u$ and $v$ disagree, i.e.:
\[
    \swap(u,v) = |\{ a, b \in C: a \pref_u b \land b \pref_v a\}|.
\]
The value $\swap(u,v)$ can be computed in time $O(m\sqrt{\log m})$ \cite{cha-pat:c:fast-swap-distance}.
For a domain
$D \subseteq \mathcal{L}(C)$, we let
$\swap(D,v) = \min_{u \in D}\swap(u,v)$.

\subsection{Structured Domains}
Let us fix a size-$m$ set
of candidates $C = \{c_1, \ldots, c_m\}$. Below, we describe the
preference domains over $C$ whose diversity we want to analyze.

Consider a connected, undirected graph $G$, such that $V(G) = C$ (we
refer to such graphs as $\SP$-graphs, or $\SP$-trees in case $G$ is
also acyclic).  A ranking $v \in \calL(C)$ is \emph{single-peaked}
with respect to $G$ if for every $t \in [m]$, the subgraph
induced by the $t$ top-ranked candidates from $v$ is
connected. $\SP(G)$ is the domain that consists of all rankings that
are single-peaked with respect to $G$ (see, e.g., the work of
\citet{elk-lac-pet:b:structured-domains}).  We focus on the following
variants:
\begin{enumerate}
\item $\SP$ is the classic single-peaked domain that consists of
  rankings single-peaked with respect to a path (often called an
  \emph{axis} and denoted $c_1 \rhd c_2 \rhd \cdots \rhd c_m$).  In
  politics, the axis may, e.g., indicate the progression from the
  most left-wing candidate to the most right-wing one.  SP is due to
  \citet{bla:b:polsci:committees-elections}.

\item $\SPOC$, introduced by \citet{pet-lac:j:spoc}, consists of
  rankings single-peaked with respect to a cycle.  SPOC preferences
  appear, e.g., when people located in different time zones want to
  choose a convenient time for an online meeting.
  The name $\SPOC$ stands for \emph{single-peaked on a circle.}

\item $\SPDF$ is a domain introduced by
  \citet{fal-sor-szu-was:c:structured-kemeny} and consists of votes
  single-peaked with respect to a tree that we obtain by taking a
  path and adding four vertices: two directly connected to one end of
  the path, and two directly connected to the other end. The name
  $\SPDF$ stands for \emph{single-peaked/double-forked}. Domains of
  rankings single-peaked with respect to trees were introduced by
  \citet{dem:j:sp-trees}.
\end{enumerate}
Whenever we speak of $\SP$, $\SPOC$, or $\SPDF$ the exact number of
candidates and their positions in respective graphs will be clear from
the context (or will be irrelevant). We use this convention 
for the other domains as well, omitting such details from
their names. 

A domain
is \emph{single-crossing} if it is
possible to list its members as $v_1, v_2, \ldots, v_n$, so that, as we
consider them from $v_1$ to $v_n$, the relative ordering of each
pair of candidates~$a$
and~$b$ changes at most once. Single-crossingness is due to
\citet{mir:j:single-crossing} and \citet{rob:j:tax}.
\begin{enumerate}
\setcounter{enumi}{3}
\item By $\SC$, we mean a single-crossing domain sampled from the
  space of all such domains using the algorithm of
  \citet{szu-boe-bre-fal-nie-sko-sli-tal:j:map}: We generate votes
  iteratively, starting with some arbitrary vote $v_0$. In each
  iteration, given vote $v_i$, we form $v_{i+1}$ by taking $v_i$'s
  copy and swapping a randomly selected pair of adjacent candidates
  that were not swapped in preceding iterations. Altogether, we
  generate rankings $v_0, \ldots, v_{\binom{m}{2}}$ that form our
  domain.
\end{enumerate}
Note that the algorithm
of~\citet{szu-boe-bre-fal-nie-sko-sli-tal:j:map} does not sample
single-crossing domains uniformly at random (so far, the only known
algorithm for such uniform sampling requires exponential time).

Let $\calT$ be an ordered, rooted tree, where each internal node has
at least two children and each leaf is labeled with a unique candidate
from $C$ (we refer to such trees as $\GS$-trees).  A \emph{frontier}
of $\calT$ is the ranking of the candidates, obtained by reading the
leaves of $\calT$ from left to right.  Domain $\GS(\calT)$ consists
exactly of those rankings $v \in \calL(C)$ that are either a frontier
of $\calT$ or a frontier of a tree obtained from $\calT$ by reversing
the order of some nodes' children.
A domain $D$ is \emph{group-separable} if $D = \GS(\calT)$ for some
$\calT$.  We are particularly interested in the following two such
domains:
\begin{enumerate}
  \setcounter{enumi}{4}
\item $\GSbal$ is a group-separable domain defined by balanced binary
  trees, i.e., binary trees where each internal node has exactly two children and
  for each two leaves, their distance from the root differs at most by $1$.

\item $\GScat$ is a group-separable domain defined by caterpillar
  binary trees, i.e., trees where each internal node has exactly two
  children, of which at least one is a leaf.
\end{enumerate}
Group-separable domains were introduced by
\citet{ina:j:group-separable,ina:j:simple-majority}, but the above
tree-based definition is due to \citet{kar:j:group-separable}.

Let $d$ be some positive integer, and let
$x\colon C \rightarrow \reals^d$ be a function that associates the
candidates with distinct points in $\reals^d$.  A ranking
$v \in \calL(C)$ is consistent with $x$ if there is a point
$x_v \in \reals^d$ such that for each two candidates $a,b \in C$ such
that $a \pref_v b$ it holds that the Euclidean distance between $x_v$
and $x(a)$ is smaller than that between $x_v$ and $x(b)$.  $D(x)$ is
the domain that includes exactly the rankings consistent with $x$.
Such domains are called \emph{Euclidean} and were studied, e.g., by
\citet{ene-hin:b:spatial,ene-hin:b:spatial2}. We focus on:
\begin{enumerate}
  \setcounter{enumi}{6}
\item $\Int$, $\Square$, and $\Cube$, where the position of each
  candidate is sampled uniformly at random from, respectively,
  $[-1,1]$, $[-1,1]^2$, and $[-1,1]^3$.
\end{enumerate}
It is well-known that $\Int$ is also a single-crossing domain, and all
its votes are single-peaked with respect to the axis obtained by
sorting the positions of the candidates.

$\SP$, $\SC$, all group-separable domains, and $\Int$ are examples of
so-called \emph{Condorcet domains}. That is, for every election with
odd number of votes from one of these domains, there is a ranking $v$
of the candidates such that if $v$ ranks some candidate $a$ over some
other candidate $b$, then a strict majority of voters prefers $a$ to
$b$.

\subsection{Distance Between Elections}  
\emph{Isomorphic swap distance} between two elections (with the same
numbers of candidates and the same numbers of voters) is a measure of
their structural similarity, introduced
by~\citet{fal-sko-sli-szu-tal:c:isomorphism-jcss}.
We extend it
to apply to elections with different numbers of
voters (in essence, we pretend to duplicate the votes so that the
elections appear to be equal-sized).

\begin{definition}
\label{def:isomorphic-swap-dist}
For two elections $E = (C,V)$ and $F = (B,U)$ such that $|C|=|B|$,
where $V = (v_1,\dots,v_{n})$ and $U = (u_1,\dots,u_{k})$, their
isomorphic swap distance is defined as follows (the indices of the
votes from $V$ are taken modulo $n$, and the indices of the votes from
$U$ modulo $k$):
\[ 
        d_\swap(E,F) = \frac{1}{nk}
        \min_{\pi:[nk]\rightarrow [nk]}
        \min_{\sigma: C \rightarrow B}
        \sum_{i \in [nk]} \swap(\sigma(v_i),u_{\pi(i)}),
\]
where $\pi$ and $\sigma$ are bijections, and by $\sigma(v_i)$ we mean vote
$v_i$ where each candidate $c \in C$ is replaced with candidate $\sigma(c) \in B$.      
\end{definition}

\subsection[k-Kemeny and Inner Diversity]{$\boldsymbol{k}$-Kemeny and Inner Diversity}
Let
$E = (C,V)$ be an election and let $R = \{r_1, \ldots, r_k\}$ be a set
of preference orders from $\calL(C)$. By the Kemeny score of $R$ with
respect to election $E$, we mean:
\[ \textstyle
  \kemscore_E(R) = \sum_{v \in V} \swap(R,v).
\]
In other words, it is the sum of the swap distances of the election's
votes to their closest rankings from $R$. The $k$-Kemeny score of an
election $E$, denoted $\kkemscore(E)$, is the smallest Kemeny score of
a size-up-to-$k$ set of rankings for this election.  By Kemeny score
we mean the $1$-Kemeny score.
Computing the Kemeny score of a given election is well-known to be
hard~\citep{bar-tov-tri:j:who-won,hem-spa-vog:j:kemeny}, even for the
case of four
voters~\citep{dwo-kum-nao-siv:c:rank-aggregation,bie-bra-den:j:kemeny-hardness}.
The notion of the Kemeny score was the original idea of
Kemeny~\citep{kem:j:no-numbers}, whereas the extension to collections
of rankings was put forward by
\citet{fal-kac-sor-szu-was:c:div-agr-pol-map}, in the context of
election diversity. Specifically, they claimed that the appropriately
normalized weighted sum of an election's $k$-Kemeny scores (for
varying $k$) captures its diversity. Indeed, the larger an 
election's $k$-Kemeny score, the more difficult it is to cluster its
votes into $k$ groups, meaning that its votes are quite different from
 one another. Consequently, these votes are diverse. The same view was taken by
\citet{fal-mer-nun-szu-was:c:map-dap-top-truncated} and was recently
applied to measure the diversity of preference
domains by~\citet{fal-sor-szu-was:c:structured-kemeny}.
Specifically, given domain $D$ over size-$m$ candidate set, they defined its 
Kemeny vector to be:
\[
\kemscore(D) = (1\hbox{-}\kemscore(D)/|D|, 2\hbox{-}\kemscore(D)/|D|, \ldots, m\hbox{-}\kemscore(D)/|D|)
\]
and they said that a given domain $D_1$ is more diverse than another domain $D_2$ (both over
equal-sized candidate sets) if $\kem(D_1)$ dominates $\kem(D_2)$ or is close to dominating it; they did not formalize this notion as they considered only a few domains. 

We broadly refer to measures of diversity based on the difficulty of
clustering as capturing \emph{inner diversity}.

\section{Outer Diversity}\label{sec:definition}

Let $C$ be a set of candidates and let $D \subseteq \calL(C)$ be a
domain over~$C$.
By the \emph{average normalized swap distance} of $D$, denoted
$\ansd(D)$, we mean the expected swap distance between a vote
chosen from $\mathcal{L}(C)$ uniformly at random and the closest vote
in $D$, divided by the maximal possible distance between two votes
in $\calL(C)$.  Formally:
\[ \textstyle
    \ansd(D) = \frac{1}{m!}\sum_{u \in \mathcal{L}(C)} \swap(D,u)/ \textstyle\binom{m}{2}.
\]

The largest possible value of $\ansd(D)$ is $0.5$, obtained when $D$
consists of a single vote, and the smallest one is $0$, obtained
for the general domain.
To ensure that \emph{outer diversity} of a domain~$D$ is between $0$
and $1$ (where $0$ means complete lack of diversity
and $1$ means full diversity), we define it as the following linear
transformation of $\ansd(D)$.

\begin{definition}
  For a domain $D \subseteq \mathcal{L}(C)$, its
  \emph{outer diversity} is defined as:
  $$ \out(D) = 1 - 2\cdot\ansd(D).  $$
\end{definition}

While outer- and inner diversity notions are based on different
principles, they are interrelated in several ways. For example, inner
diversity, as defined by
\citet{fal-kac-sor-szu-was:c:div-agr-pol-map,fal-mer-nun-szu-was:c:map-dap-top-truncated,fal-sor-szu-was:c:structured-kemeny},
relies on analyzing $k$-Kemeny scores of given elections or domains,
whereas $\ansd(D)$ is simply the normalized $k$-Kemeny score of the
input domain $D$, with respect to the $\UN$ election.  Considered from
a different perspective,
$\ansd(D)$ is equal to the smallest possible isomorphic swap distance
between $\UN$ and a $D$-election.

% \begin{restatable}{proposition}{ProOuterSwapDistance}\label{pro:outer-swap-distance}
\begin{proposition}\label{pro:outer-swap-distance}
    For every domain $D \subseteq \mathcal{L}(C)$,
    it holds that:
    $$
        \ansd(D) = \min_{\text{$E$ is a $D$-election}}d_{\swap}(\UN,E) / \textstyle\binom{m}{2}.
    $$
\end{proposition}
% \end{restatable} 
\begin{proof}
    Fix election $E = (C,V)$ yielding the minimal distance.
    Without loss of generality,
    we can assume that the number of votes in $V$
    is a multiple of $m!$, i.e.,
    $V = \{v_1,\dots, v_{k\cdot m!}\}$
    for some $k \in \mathbb{N}$,
    because creating $k$ additional copies of all votes
    does not affect the isomorphic swap distance.
    Let $u_1,\dots,u_{k \cdot m!}$ be
    copies of voters in $\UN$
    as denoted in \Cref{def:isomorphic-swap-dist}
    and $\pi : [k \cdot m!] \rightarrow [k \cdot m!]$
    be a matching of voters yielding the minimum distance.
    We can assume that the matching of candidates, $\sigma$, is the identity
    since for the distance to $\UN$
    every matching of candidates gives the same sum of distances.

    Observe that $\swap(u_i,v_{\pi(i)}) = \swap(D, u_i)$, for every $i \in [k \cdot m!]$
    as otherwise $d_\swap(\UN,E)$ could be decreased by
    exchanging $v_{\pi(i)}$ for $v$ yielding the minimum
    and keeping all other voters as is.
    This also implies that
    for each $i \in [m!]$ and $\ell \in [k-1]$,
    we have
    \(
    \swap(u_i,v_{\pi(i)}) = \swap(u_{i + \ell \cdot m!},v_{\pi(i + \ell \cdot m!)}).
    \)
    Then, we get
    \begin{align*}
\textstyle        \sum_{r \in \mathcal{L}(C)} \swap(D, r)
        &=\textstyle \sum_{i \in [m!]} \swap(D, u_{i}) \\
        &=\textstyle \sum_{i \in [m!]}\swap(u_{i}, v_{\pi(i)}) \\
        &=\textstyle \frac{1}{k} \sum_{i \in [k \cdot m!]}\swap(u_{i}, v_{\pi(i)}) \\
        &=\textstyle m! \cdot d_\swap(E,\UN),
    \end{align*}
  which yields the thesis.
\end{proof}

Since \citet{fal-kac-sor-szu-was:c:div-agr-pol-map} have shown that
proximity to $\UN$ is highly correlated with their form of inner
diversity, we conclude that both approaches are capturing the same
high-level idea.

\begin{table}[t]
    \centering
    \begin{tabular}{ccc}
    \toprule
                &       & Complexity of Finding \\
         Domain &  Size & Closest Ranking from $D$ \\
      \midrule
        $\GS(\calT)$ & $\leq 2^{m-1}$ & $O(m^2)$ \\
        $\GScat$ & $2^{m-1}$ & $O(m\log m)$ \\
        $\GSbal$ & $2^{m-1}$ & $O(m\log m)$ \\
        \midrule
        $\SP$    & $2^{m-1}$ & $O(m^2)$ \\
        $\SPDF$  & $2^{m+1} - 16$  & $O(m^4)$ \\
        $\SPOC$  & $m2^{m-2}$ & $O(m^2)$ \\
      $\SP(T)$ & --- & $\substack{\textstyle O(km^k) \\ \scriptscriptstyle \text{$k =$ number of $T$'s leaves}}$ \\
        $\SP(G)$ & --- & $\np$-com. \\
        \midrule
        $\SC$    & $1+\nicefrac{m(m-1)}{2}$ & $O(m^2)^*$ \\
        \midrule
        $\Int$   & $1+\nicefrac{m(m-1)}{2}$ & $O(m^2)^*$ \\
        $\Square$& $O(m^4)$                 & $O(m^4)^*$ \\
        $\Cube$  & $O(m^6)$                 & $O(m^6)^*$ \\
    \bottomrule
    \end{tabular}
    \caption{For each domain we give its size and the complexity of
      finding its closest member (in terms of swap distance) to a
      given input ranking. Running times marked with $^*$ do not
      include the time needed for preprocessing.}
    \label{tab:size-vote-complexity}
\end{table}

\section{Computing Outer Diversity}\label{sec:computation}

For domains over sufficiently small candidate sets, it is possible to
compute outer diversity exactly. In the most basic approach, given a
domain $D$ over candidate set $C$, we could simply compute the swap
distance between every vote in $D$ and every vote in
$\calL(C)$. Naturally, this is very inefficient and computing outer
diversity of, say, $\SP$ with $m$ candidates would require time
$O(m! \cdot 2^{m-1} \cdot m\sqrt{\log m})$;
the general domain has $m!$ rankings, $\SP$ has $2^{m-1}$ of them,
and it takes $O(m\sqrt{\log m})$ time to compute the swap
distance~\citep{cha-pat:c:fast-swap-distance}.
Fortunately, there is a faster approach that given a domain $D$,
for each $i$ forms a set $D_i$ of rankings at swap distance $i$
from $D$
(the algorithm and the proof can be found in \Cref{app:computation}).

\begin{restatable}{proposition}{ProBfs}\label{pro:bfs}
% \begin{proposition}\label{pro:bfs}
  There is an algorithm that
  given domain $D$ over $m$ candidates (represented by listing its members),  
  computes $\out(D)$ in time $O(m^2\cdot m!)$.
% \end{proposition}
\end{restatable}
% \begin{proof}
%   We use \Cref{alg:growing-balls}, whose correctness
%   follows directly from the definitions of $\ansd(D)$ and $\out(D)$.
%   In line \ref{gb:core} of the algorithm, for
%   each vote $v \in \calL(C)$ we consider all $m-1$ votes $u$
%   obtained from $v$ by a single swap of adjacent candidates,
%   resulting in $O(m \cdot m!)$ memberships checks.
%   Rankings from $\bigcup_{j=0}^{i} D_j$ are stored in a trie (prefix tree),
%   which allows $O(m)$-time membership checks and insertions.
%   For the current iteration $i$, we store only sets $D_i$ and $D_{i+1}$
%   (each taking $O(m \cdot m!)$ space) while retaining the values $|D_j|$ for $j<i$.
%   Hence, the computational bottleneck is line \ref{gb:core} executed $O(m \cdot m!)$ times, each taking $O(m)$ time, leading to a total running time of $O(m^2 \cdot m!)$.
% \end{proof}

To compute outer diversity for larger candidate sets, we resort to
sampling.  Namely, given a domain $D$ over a size-$m$ candidate set
$C$, we fix a number $N$, sample $N$ rankings from $\calL(C)$, for
each sampled ranking $v$ we compute $\swap(D,v)$ and output the
average of these values, divided by $\binom{m}{2}$. This gives an estimate for
$\ansd(D)$, based on which we obtain $\out(D)$.  However, to implement
this idea efficiently, we need fast algorithms for
the following problem:
Given a ranking $v$ and a domain $D$, compute
$\swap(D,v)$.
We dedicate the rest of this section to seeking  algorithms
for this problem for various domains, and to establishing
its complexity.

On the outset, the problem can be even $\np$-hard.
For example, for each set of $4m$ candidates
$C = \{c_{i,j} : i \in [4],j \in [m]\}$, let the \emph{4-alignment}
domain contain each vote of the form
\( \{c_{1,1},\dots,c_{1,m}\} \! \succ\! \{c_{2,1},\dots,c_{2,m}\} \!
\succ\!  \{c_{3,1},\dots,c_{3,m}\} \! \succ\!
\{c_{4,1},\dots,c_{4,m}\}, \) in which the order of the candidates,
based on their second indices, is identical in each block. Then we
have the following hardness result (in essence, for this domain the
problem of finding a closest vote in the domain becomes the problem of
computing Kemeny score for $4$ voters, known to be
$\np$-hard~\citep{dwo-kum-nao-siv:c:rank-aggregation,bie-bra-den:j:kemeny-hardness}).

% \begin{restatable}{theorem}{ThmAlignment}
\begin{theorem}
    Let $D$ be the 4-alignment domain.
    Given vote $v$ and integer $d \in \mathbb{N}$
    it is NP-complete to decide whether $\swap(D,v)\le d$.
\end{theorem}
% \end{restatable}
\begin{proof}
    The verification is straightforward.
    Given the vote in $D$ that yields the closest distance to $v$,
    we can check whether this distance is larger than $d$
    in polynomial time.

    To show hardness,
    we give a reduction from \textsc{KemenyOn4Votes}.
    In this problem we are given
    a candidate set $C = \{c_1,\dots,c_m\}$,
    four votes $v_1,v_2,v_3,v_4 \in \mathcal{L}(C)$,
    and an integer $d \in \mathbb{N}$,
    and we ask whether there exists a ranking $u \in \mathcal{L}(C)$
    for which it holds that $\sum_{i \in [4]}\swap(v_i,u)\le d$.
    This is known to be NP-complete~\cite{dwo-kum-nao-siv:c:rank-aggregation,bie-bra-den:j:kemeny-hardness}.

    Now, for each instance of \textsc{KemenyOn4Votes},
    let us construct an instance of our problem.
    To this end, let $C' = \{c_{i,j} : i \in [4],j \in [m]\}$
    and let $v \in \mathcal{L}(C')$ be a concatenation of votes $v_1,v_2,v_3$, and $v_4$, i.e.,
    $c_{i,j} \succ_v c_{i',j'}$, if and only if,
    $i < i'$ or $i = i'$ and $c_j \succ_{v_i} c_{j'}$.
    Also, for every ranking $u \in \mathcal{L}(C)$ let
    $f(u)$ denote a ranking in $\mathcal{L}(C')$ that is a
    concatenation of 4 copies of $u$, i.e.,
    $c_{i,j} \succ_{f(u)} c_{i',j'}$, if and only if,
    $i < i'$ or $i = i'$ and $c_j \succ_{u} c_{j'}$.
    Then, 4-agreement domain can be alternatively written as
    $D = \{f(u) : u \in \mathcal{L}(C)\}$.
    Moreover, $\swap(v,f(u)) = \sum_{i \in [4]}\swap(v_i,u)$,
    for each $u \in \mathcal{L}(C)$.
    Therefore, indeed, there exists $u \in \mathcal{L}(C)$ such that
    $\sum_{i \in [4]}\swap(v_i,u) \le d$, if and only if,
    $\swap(D,v) \le d$.
\end{proof}

Despite this negative result, for most of our domains
we find efficient algorithms for computing the distance to a given
vote (see \Cref{tab:size-vote-complexity}).  In the following, we
always use $C = \{c_1, \ldots, c_m\}$ to denote the set of $m$
candidates in the domain under consideration.

\subsection{Single-Peaked Domains}

Let us first consider the family of single-peaked domains. 
We note that \citet[Theorem 4.5.]{fal-hem-hem:j:nearly-sp} already
gave a polynomial-time algorithm for computing the distance between
$\SP$ and a given ranking, but their approach---based on dynamic
programming---required $O(m^3)$ time.
We improve this algorithm to run in $O(m^2)$ time.  The main idea is
to use dynamic programming to iteratively compute the distance between
a given ranking $v$ and votes that rank more and more bottom
candidates as required by $\SP$.

Assume that we are given a vote $v$ and a societal axis
$c_1 \rhd c_2 \rhd \dots \rhd c_m$.  For each
$\ell,r \in \{0,1,2,\dots,m\}$ such that $\ell + r \le m$, let
$C_{\ell,r}$ denote the set of the first $\ell$
and the last $r$  candidates according to $\rhd$.
Formally, we have
$C_{\ell,r} = \{c_1, \dots, c_\ell\} \cup \{c_{m+1-r}, \dots, c_m\}$;
by convention, for $\ell=0$ we have
$\{c_1, \dots, c_\ell\}=\varnothing$, and for $r=0$ we have
$\{c_{m+1-r}, \dots, c_m\}=\varnothing$.  Then, by $U_{\ell,r}$ we
denote the set of all votes $u \in \mathcal{L}(C)$ in which
(a)~candidates from $C_{\ell,r}$ are in the bottom $\ell+r$ positions,
and (b)~for each $t \in \{m, m-1, \ldots, m-\ell-r+1\}$,
the top $t$ candidates of $u$ form an interval within $\rhd$.  Observe
that $U_{0,0} = \calL(C)$, whereas if $\ell + r = m$, then
$U_{\ell, r} = \SP$.
We write $A_{\ell,r}$ to denote the minimal swap distance between
$v$ and $u \in
U_{\ell,r}$.  As we will show, all values of
$A_{\ell,r}$ can be computed efficiently in \Cref{alg:single-peaked},
using a recursive formula.

\begin{algorithm}[t] 
    \caption{Distance between a ranking and $\SP$}
    \label{alg:single-peaked}
    \begin{algorithmic}[1]{
        \REQUIRE Ranking $v \in \mathcal{L}(C)$, societal axis $c_1 \rhd \dots \rhd c_m$ \\ 

        \textsc{Phase 1, Precomputation:}
        \FOR{$i \in [m]$}
            \STATE $L_{i,i} \leftarrow 0, R_{i,i} \leftarrow 0$
            \STATE \textbf{for} $j \in \{i+1,\dots,m\}$ \textbf{do}
                $L_{i,j} \leftarrow L_{i,j-1} +
                    [c_{i} \succ_v c_j] $
            \STATE \textbf{for} $j \in \{i-1,\dots,1\}$ \textbf{do}
                $R_{j,i} \leftarrow R_{j+1,i} +
                [c_{i} \succ_v c_j]$
        \ENDFOR
        
        \textsc{Phase 2, Main Computation:}
        \STATE $A_{0,0} \leftarrow 0$
        \STATE \textbf{for} $\ell \in [m-1]$ \textbf{do}
            $A_{\ell,0} \leftarrow A_{\ell-1,0} + L_{\ell,m}$ \label{alg:sp:recursion1}
        \FOR{$r \in [m-1]$}
            \STATE $A_{0,r} \leftarrow A_{0,r-1} + R_{1,m+1-r}$ \label{alg:sp:recursion2}
            \FOR{$\ell \in [m-r-1]$}
                \STATE $A_{\ell,r} \leftarrow \min(
                    A_{\ell-1,r} + L_{\ell,m-r}, \
                    A_{\ell,r-1} + R_{\ell+1,m+1-r})$ \label{alg:sp:recursion}
            \ENDFOR
        \ENDFOR
        \RETURN \label{alg:sp:return} $\min_{\ell \in [m]}A_{\ell-1,m-\ell}$}
        \end{algorithmic}
\end{algorithm}

\begin{theorem}
    \label{thm:vote-domain:single-peaked}
    \Cref{alg:single-peaked} computes the distance between a given vote and a single-peaked domain in time $O(m^2)$.
\end{theorem}
  
\begin{proof}
  For the running time, observe that each of our loops is over at most
  $m$ elements, and we have at most two levels of nested loops.  Each
  individual iteration can be completed in time $O(1)$.  The final
  minimum in line \ref{alg:sp:return} requires $O(m)$ time.

  Let us now analyze the correctness of the algorithm.  For each
  $i, j \in [m]$, with $i \le j$, we let $L_{i,j}$ be the number of
  candidates in $\{c_i, c_{i+1}, \ldots, c_j\}$ that $v$ ranks below
  $c_i$. Consequently, we have that $L_{i,i} = 0$ and, if $i < j$,
  then either $L_{i,j} = L_{i,j-1} +1$ (if $v$ ranks $c_i$ ahead of
  $c_j$) or $L_{i,j} = L_{i,j-1}$ (otherwise).  Similarly, for $j \leq i$, $R_{i,j}$
  is the number of candidates in $\{c_j, c_{j+1}, \ldots, c_i\}$ that
  $v$ ranks below $c_i$ ($R_{j,i}$ satisfies analogous relations as
  $L_{i,j}$). The algorithm computes the values of $L_{i,j}$ and
  $R_{j,i}$ in \textsc{Phase 1}.

  Then, in \textsc{Phase 2}, the algorithm computes all the values
  $A_{\ell,r}$ for $\ell, r \in [m]$ such that $\ell+r \leq m-1$.  Let
  us fix such $\ell$ and $r$. We note that every ranking in
  $U_{\ell,r}$ either ranks $c_\ell$ or $c_{m+1-r}$ on position
  $m+1-\ell-r$ (i.e., on the $\ell+r$'th position from the
  bottom). Indeed, for all rankings in $U_{\ell,r}$ we have that the first
  $m+1-\ell-r$ candidates form an interval within $\rhd$. However, by
  definition, all of these candidates, except for the one ranked on
  position $m-\ell-r+1$, belong to $C \setminus
  C_{\ell,r}$. Consequently, to form the interval, the candidate on
  position $m-\ell-r+1$ must be either $c_\ell$ or $c_{m+1-r}$.  Let
  $U_{\underline{\ell},r}$ be a subset of votes from $U_{\ell-1,r}$
  that additionally have $c_\ell$ in the position $m+1-\ell-r$.
  Similarly, let $U_{\ell,\underline{r}}$, be a subset of votes from
  $U_{\ell,r-1}$ with $c_{m+1-r}$ in the position $m+1-\ell-r$. By the
  preceding argument, we have that
  $U_{\ell,r} = U_{\underline{\ell},r} \cup U_{\ell,\underline{r}}$
  (if $\ell =0$, we assume $U_{\underline{\ell},r} = \varnothing$, if
  $r=0$, $U_{\ell,\underline{r}} = \varnothing$).
  Thus,
  \( A_{\ell,r} = \min(\swap(U_{\underline{\ell},r},v),
  \swap(U_{\ell,\underline{r}},v)).  \)

  Let $u$ be a vote in $U_{\underline{\ell},r}$ that minimizes
  $\swap(u,v)$.  Observe that $m-\ell-r$ first candidates in $u$
  appear in the same relative order as they appear in $v$ (otherwise, ordering
  them as in $v$ would decrease the distance).  Let $u'$ be a vote
  obtained from $u$ by ensuring that it ranks its first $m-\ell-r+1$
  candidates in the same relative order as in $v$ (in other words,
  $u'$ is the same as $u$, except that it might rank $c_\ell$ some
  positions earlier).
  It must be that $u' \in U_{\ell-1,r}$.  Moreover, we can show that
  $u'$ minimizes swap distance to $v$ among rankings in
  $U_{\ell-1,r}$, i.e., $\swap(u',v) = A_{\ell-1,r}$.  Indeed, the
  first $m-\ell-r+1$ candidates are in the optimal order (the same as
  in $v$), and if rearranging the last $\ell+r-1$ candidates could
  decrease the distance, we could also rearrange them in the same way
  in $u$.  Now, when we look at the inversions counted in
  $\swap(v,u)$, we see that we count all inversions that we count in
  $\swap(v,u')$ and additionally those from having $c_\ell$ after all
  of the first $m-\ell-r$ candidates.  But those are exactly the
  inversions we store in $L_{\ell,m-r}$.  Thus, we get that
  \(\swap(U_{\underline{\ell},r},v) = A_{\ell-1,r} +
  L_{\ell,m-r}.  \) Analogously, we can prove that
  \( \swap(U_{\ell,\underline{r}},v) = A_{\ell,r-1} +
  R_{\ell+1,m+1-r}.  \) This way, we obtain the recursive equation
  used in line \ref{alg:sp:recursion}, as well as the equations from
  lines \ref{alg:sp:recursion1} and \ref{alg:sp:recursion2} (in their
  cases either $r=0$ or $\ell=0$ so respective parts of the equation
  disappear).

    Finally, for $\ell \in [m]$,
    we observe that $A_{\ell-1,m-\ell}$ is the minimal distance
    from $v$ to a single-peaked ranking $u$ in which $c_{\ell}$ is the top candidate.
    Thus, to get the overall smallest distance,
    we take the minimum from all these values.
\end{proof}

Every vote in $\SPOC$ is single-peaked along the axis
obtained by ``cutting'' the cycle between some two adjacent
candidates~\cite{pet-lac:j:spoc}.  There are $m$ such axes, hence we
can run \Cref{alg:single-peaked} for each of them and choose the
minimum distance.  This gives as an algorithm running in time
$O(m^3)$. We can improve that and get an $O(m^2)$ algorithm by a
similar dynamic programming algorithm as for $\SP$
(the algorithm and the proof can be found in \Cref{app:computation:sp}).

\begin{theorem}
  There is an algorithm that
  computes the swap distance between a given vote and $\SPOC$ in time
  $O(m^2)$.
\end{theorem}

We can also extend \Cref{alg:single-peaked} to work for the case of
$\SP(T)$, where $T$ is an $\SP$-tree. If $T$ has $k$ leaves (i.e., $k$
nodes of degree~$1$), then the algorithm requires $O(km^k)$ time.  The
main idea is to implement dynamic programming over sets of connected
vertices in $T$, of which there are $O(m^k)$
(the algorithm and the proof can be found in \Cref{app:computation:sp}).

\begin{theorem}
    \label{thm:vote-domain:single-peaked-on-a-tree}
    There is an algorithm that given an $\SP$-tree that has $k$
    leaves, computes the swap distance between a given vote and $\SP(T)$ in
    time $O(km^k)$.
\end{theorem}

Given the algorithms for SP, SPOC, and single-peaked-on-a-tree
domains, one could ask for a general polynomial-time algorithm that
works for all single-peaked-on-a-graph domains.  We prove that
in this general case the problem is $\np$-complete
(the proof can be found in \Cref{app:spog:hardness}).

\begin{restatable}{theorem}{ThmSpogDistHardness}\label{thm:spog:dist:hardness}
    Given a graph $G$,
    a vote $v \in \mathcal{L}(V(G))$,
    and an integer $d \in \mathbb{N}$,
    deciding if \break $\swap(SP(G),v)\le d$ is $\np$-complete.
\end{restatable}

\subsection{Group-Separable Domains}

For a group-separable domain with an arbitrary tree, we %can
show an algorithm that computes the distance to a given vote in time
$O(m^2)$.

Assume we are given a vote $v$
and a group separable domain $D = \GS(\calT)$.  Then, observe that
finding a vote $u \in D$ that minimizes $\swap(u,v)$ is equivalent to
reversing the order of some of the children of each internal node of
$\calT$ so that the frontier~$u$ of~$\calT$
minimizes $\swap(u,v)$.  Moreover, the change in distance we get by
reversing the order of the children of one particular node is
independent of the configuration of the other nodes.  Hence, we can
consider internal nodes of tree $\calT$ one by one, and for each
decide in which of the two ways its children should be ordered.
Fix such an arbitrary node with $k$ children, and let
$C_1,C_2,\dots,C_k$ denote the sets of candidates associated with
leaves that are descendants of each of the children, when looking from
left to right.  This configuration would incur the distance of:
\[ \textstyle
    \sum_{1\le i < j\le k} |\{(a,b)\in C_i\times C_j : b \succ_v a\}|,
\]
while reversing the order gives the distance of:
\[ \textstyle
    \sum_{1\le i < j\le k} |\{(a,b)\in C_i\times C_j : a \succ_v b\}|.
\]
Thus, we compute the values of both sums and choose the configuration that
leads to the lower one (or make an arbitrary choice in case of a tie).
When considering all internal nodes of $\calT$ in this way, we check
each pair of candidates exactly once.  Hence, the running time of this
algorithm is $O(m^2)$.

\begin{theorem}
  There is an algorithm that given a $\GS$-tree $\calT$ and a vote $v$, computes 
  $\swap(\GS(\calT),v)$ in time $O(m^2)$.
\end{theorem}

For
$\GSbal$ and $\GScat$, we give algorithms running in time
$O(m\log m)$. Both algorithms follow the general approach outlined
above, but for $\GSbal$ we speed up computing inversions using an
approach similar to that from the classic Merge Sort algorithm, and
for $\GScat$ we use a special data structure
(the details can be found in \Cref{app:computation:gs}).

\begin{theorem}
  There are algorithms that compute the swap distance between a given vote
  and $\GScat$ and $\GSbal$ (represented via $\GS$-trees) in time
  $O(m \log m)$.
\end{theorem}

\subsection{Single-Crossing and Euclidean Domains}

Both single-crossing and Euclidean domains contain polynomially many
votes, so a brute-force algorithm that given a ranking $v$ computes
its swap distance to all the rankings in the domain runs in polynomial
time. For example, for $\SC$, which contains $O(m^2)$ rankings, it
would run in time
$O(m^3 \sqrt{\log m})$~\cite{cha-pat:c:fast-swap-distance}. However,
as we typically want to compute the distance from many votes to our
domains, we get better running times via appropriate preprocessing.
Briefly put, for each domain $D \in \{\SC, \Int, \Square, \Cube\}$ we
can arrange the rankings from these domains on a tree $T(D)$---or even
on a path, in case of $\Int$ and $\SC$---so that two neighboring
rankings are at swap distance one. Then, to compute a distance from a
given ranking $v$ to each member of the domain, we compute the
distance between $v$ and an arbitrary ranking in the domain, and then
traverse the tree, updating the distance on the fly, so for each
member of the domain we get its swap distance to $v$.
Building $T(D)$ adds, at most, factor $O(m^2)$ to the complexity of
computing the rankings from the domain
(the details can be found in \Cref{app:computation:sc-euc}).

\begin{theorem}
  For each $D$ that is either $\SC$ or a Euclidean domain, there is an
  algorithm that given a ranking $v$ and tree $T(D)$ computes
  $\swap(D,v)$ in time $O(|D|)$.
\end{theorem}

\section{Analysis of the Domains}\label{sec:analysis}

Let us now analyze the outer diversity of our domains.  We first
consider the case of $8$ candidates, and then we analyze how the outer
diversities of our domains change as the number of candidates grows.
The case of $8$ candidates is interesting for the following, somewhat
interrelated, reasons: (1)~\citet{fal-sor-szu-was:c:structured-kemeny}
largely focused on this case, and we want our results to be comparable
to theirs; (2)~The case of $8$ candidates is among the most popular
ones in experiments within computational social
choice~\citep{boe-fal-jan-kac-lis-pie-rey-sto-szu-was:c:guide};
(3)~Considering only $8$ candidates allows us to perform exact
computations.

\subsection{Outer-Diversity for Eight Candidates}

\begin{table}
\centering
\begin{tabular}{cccccc}
\toprule
Domain $D$ & $|D|$ & $\ansd(D)$ &$\out(D)$ & dist-$1$ & dist-$1$/$|D|$ \\
\midrule
  Vote+Its Rev. & 2 & 0.384 & 0.232 & 14 & 7 \\
\midrule
  GS/cat & 128 & 0.194 & 0.613 & 704 & 5.5 \\
  GS/bal & 128 & 0.257 & 0.486 & 384 & 3 \\
\midrule
  SP & 128 & 0.284 & 0.432 & 384 & 3 \\
  SP/DF & 496 & 0.239 & 0.522 & 968 & 1.952 \\
  SPOC & 512 & 0.196 & 0.608 & 1280 & 2.5 \\
\midrule
  SC & 29 & 0.316 & 0.368 & 130.3 & 4.493 \\
\midrule
  1D-Int. & 29 & 0.311 & 0.378 & 134.8 & 4.648 \\
  2D-Square & 351 & 0.217 & 0.566 & 988.0 & 2.815 \\
  3D-Cube & 2311 & 0.138 & 0.724 & 3878.2 & 1.678 \\
\midrule
  Largest Cond. & 224 & 0.282 & 0.435 & 544 & 2.429 \\
\bottomrule
\end{tabular}
  \caption{\label{tab:diversity} Size, average normalized swap
    distance, outer diversity, and size of direct neighborhood (also
    normalized) of various domains, for the case of $\boldsymbol 8$
    candidates. The standard deviation of outer diversity for domains that we need to
    sample ($\boldsymbol\SC$, $\boldsymbol\Int$, $\boldsymbol\Square$,
    $\boldsymbol\Cube$) is no larger than $\boldsymbol{0.005}$ (for
    ten samples).}
\end{table}

For each of our domains, in \Cref{tab:diversity} we provide its size,
average normalized swap distance, outer diversity value, the number of
votes in $\calL(C)$ that are exactly at swap distance $1$ from this
domain (we refer to this as the \emph{size of the direct
  neighborhood}),
and the latter number normalized by the size of the
domain (we analyze these values later on).
Additionally, the table also includes $\LC$ domain, i.e., the largest
Condorcet domain over $8$ candidates, recently discovered by
\citet{lee-mar-rii:j:largest-con-domain8}.
Sorting our domains with respect to their outer diversity values
gives the following ranking:
\newcommand{\dmn}[2]{\substack{\textstyle #1\phantom/\!\! \\ \scriptstyle #2}}
{\small
\begin{multline*}
  \dmn{\Cube}{0.719} \pref \{\dmn{\GScat}{0.613}\ \dmn{,}{} \dmn{\SPOC}{0.608}\} \pref  \dmn{\Square}{0.565} \pref \dmn{\SPDF}{0.522}
                      \pref \dmn{\GSbal}{0.486}  \pref \{ \dmn{\LC}{0.435}\ \dmn{,}{} \dmn{\SP}{0.432} \}  \pref \{ \dmn{\Int}{0.386}\ \dmn{,}{} \dmn{\SC}{0.37}\}.
\end{multline*}}
It is quite interesting that even though $\LC$ is the largest
Condorcet domain over $8$ candidates, its outer diversity is very
similar to that of $\SP$, which contains nearly half of the votes, and
it is notably lower than outer diversities of $\GScat$ and $\GSbal$
(both of the same cardinality as $\SP$). However, a closer analysis of
this domain confirms that it is not as diverse as one might expect
given its size. For example, there are only $4$ candidates that are
ever ranked first in its votes, and $4$ different candidate that are
ever ranked last (indeed, the domain has further restrictions along
these lines, which we omit due to limited space).  Next, we note that
our ranking is very similar to an analogous one obtained by
\citet{fal-sor-szu-was:c:structured-kemeny} based on inner diversity
(also for the case of 8 candidates; note in their case there are no
specific values measuring diversity and the ranking was obtained by
comparing Kemeny vectors of the domains):
\begin{multline*}
  \text{GS/cat} \pref \text{3D-Cube} \pref \{ \text{2D-Square, SPOC}\}
                 \pref \{\text{SP/DF, GS/bal}\} \pref \text{SP} \pref \{\text{SC, 1D-Int.}\}.
\end{multline*}
Both rankings put $\Cube$ and $\GScat$ as the most diverse domains,
and they both put $\Int$ and $\SC$ as the least diverse ones.
Further, they both rank domains from the same families identically:
$\SPOC$ is more diverse than $\SPDF$, which is more diverse than
$\SP$, and $\GScat$ is more diverse than $\GSbal$ (not to mention the
ranking of the Euclidean domains). 
The fact that $\Cube$ has higher outer diversity than $\GScat$, as
well as the tie between $\GScat$ and $\SPOC$, are artifacts of
considering only 8 candidates and for larger numbers of candidates these
relations change (see \Cref{sec:out-div-larger}).

Below, we analyze two features of our domains that are not directly
related to capturing diversity, but which manifest themselves during
outer diversity computations and which shed some light on
how our domains are arranged within the general domain.

\begin{figure*}[t]
\begin{flushleft}
  \includegraphics[width=\textwidth]{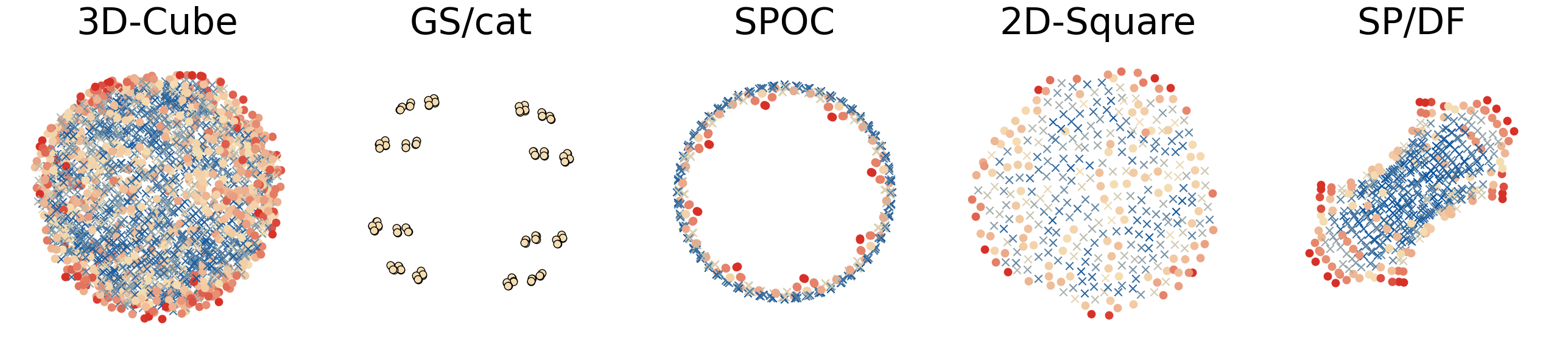}
  \vspace{0.2cm}\\
\begin{minipage}{0.79\textwidth}
  \includegraphics[width=\textwidth]{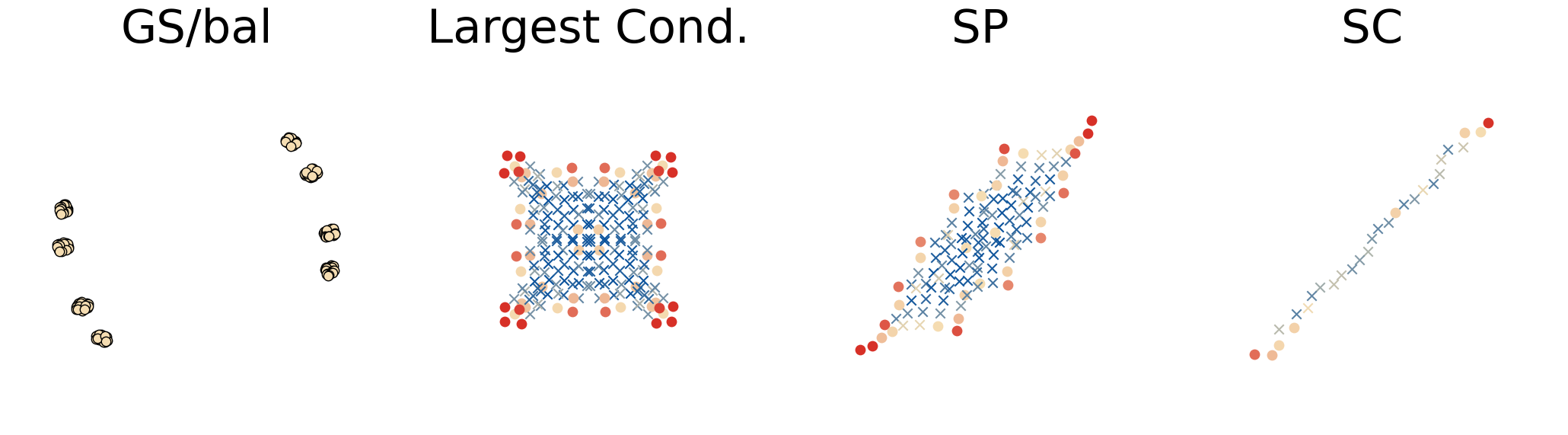}
\end{minipage}%
\hspace{0.7cm}
\begin{minipage}{0.1\textwidth}
  \includegraphics[width=\textwidth]{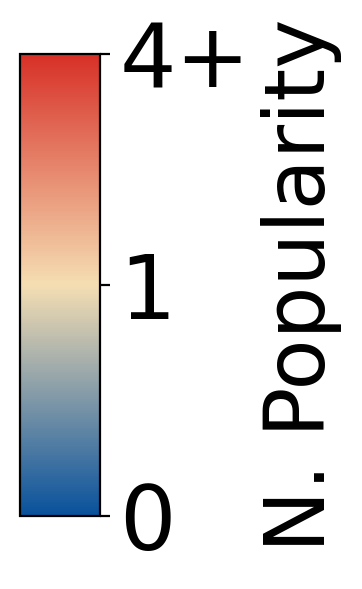}
\end{minipage}
\end{flushleft}
\caption{\label{fig:microscope_m8_without_ic}Microscope plots of our
  domains, where each dot/cross represents a ranking from the domain,
  colored according to its normalized popularity (see
  \Cref{rem:microscope}). Rankings with normalized popularity below
  $\boldsymbol 1$ are marked with crosses, and the remaining ones with
  dots.  Dots marking rankings with normalized popularity equal to
  exactly $\boldsymbol 1$ have a black border.}
\end{figure*}

\subsubsection{Direct Neighborhoods}

The size of the direct neighborhood of a domain, normalized by the
sizes of this domains, is interesting as it gives some intuition on how
the domain is ``spread'' over $\calL(C)$. For example, the domain that
consists of a single ranking and its reverse is ``maximally spread:''
Its two members are as far apart as possible and, as we consider $8$
candidates, there are exactly 7 rankings next to each of the domain
members, neither of which belongs to the domain.
Among our structured domains, $\GScat$ is the most spread one,
with the value of $5.5$, and $\Cube$ is the least spread, with the value of
$1.678$.
Hence, members of $\Cube$
are packed quite closely within $\calL(C)$. While one could think that
this is a consequence of $\Cube$'s large
size, $\calL(C)$ contains more than $16$ rankings for every ranking in
$\Cube$.
It is interesting that for some domains the normalized sizes of their
direct neighborhoods are appealing, round numbers (such as $3$ for
$\GSbal$ or $5.5$ for $\GScat$). For $\GSbal$ and $\GScat$, we show
that this is not a mere coincidence; for the other domains we leave
this issue open.

% \begin{restatable}{proposition}{ProGSBalDist}\label{pro:gs-bal-dist1}
\begin{proposition}\label{pro:gs-dist1}
  Let $D$ be the $\GSbal$ domain for $m = 2^k$ candidates. For
  every ranking $v \in D$ there are exactly $2^{k-1}-1$ unique ones
  from $\calL(C) \setminus D$ at swap distance $1$ from $v$.
\end{proposition}
% \end{restatable} 
\begin{proof}  
  Let $m = 2^k$ be the number of candidates, let the candidate set be
  $C = \{c_1, \ldots, c_m\}$, and let $D$ be our $\GSbal$ domain for
  $C$. Further, let $v$ be ranking in $\calL(C)$.  W.l.o.g., we can
  assume that $v$ ranks the candidates as
  $c_1 \pref c_2 \pref \cdots \pref c_m$. For each $i \in [m-1]$, let
  $v(i)$ be the ranking obtained from $v$ by swapping candidates $c_i$
  and $c_{i+1}$. These are all the rankings from $\calL(C)$ that are
  at swap distance $1$ from $v$. By definition of $D$, for every odd
  $i \in [m-1]$, $v(i)$ is in $D$, and for each even $i$ it is in
  $\calL(C) \setminus D$. Further, for each even~$i$, swap distance of
  $v(i)$ to every member of $D$ other than $v$ is larger than $1$:
  Indeed, in every vote from $D$, $c_i$ and $c_{i-1}$ must be ranked
  next to each other.
  To achieve this, by performing a single swap on $v(i)$,
  we need to swap $c_{i+1}$ with $c_{i-1}$ or $c_i$.
  The former, does not lead to a vote from $D$, the latter leads to $v$.
  This completes the proof.
\end{proof}

% \begin{restatable}{proposition}{ProGSCatDist}\label{pro:gs-cat-dist1}
\begin{proposition}
  Consider $\GScat$ over $m \geq 4$ candidates. For every ranking
  $v \in \GScat$ there are exactly $m-3$ unique ones from
  $\calL(C) \setminus D$ at swap distance $1$ from $v$, and one ranking
  from $\calL(C)$ that is at swap distance $1$ from $v$ and one other
  ranking in $\GScat$.
\end{proposition}
% \end{restatable}
\begin{proof}
  Let $\calT$ be a binary caterpillar tree
  with $c_1$ being the leaf closest to the root,
  followed by $c_2$, and so on up to $c_m$.
  Let $D$ be a $\GScat$ domain consistent with tree $\calT$.
  
  Observe that in each vote $v \in D$,
  when we read the candidates along $\succ_v$,
  the indices are increasing until candidate $c_m$,
  after which they are decreasing.%
  \footnote{This property makes $\GScat$ in some sense dual to $\SP$,
  which was observed e.g. in~\cite{boe-bre-elk-fal-szu:c:map-pref-learning}.}
  After swapping any pair of candidates in $v$,
  except for the pair $\{c_{m-1}, c_{m}\}$,
  we obtain $v'$ that does not have this property,
  hence $v' \not \in D$.
  Moreover, unless we swapped
  $c_{m-2}$ with either $c_m$ or $c_{m-1}$,
  the only way to restore this property by a single swap,
  is to go back to $v$.
  Thus, $\swap(v',u) > 1$, for each $u \in D \setminus \{v\}$.
  
  On the other hand, if $v'$ is obtained from $v$ by swapping
  $c_{m-2}$ with $c_m$ or $c_{m-1}$
  (whichever is adjacent to $c_{m-2}$ in $v$),
  then in $v'$, candidate $c_{m-2}$ is ranked exactly between
  $c_m$ and $c_{m-1}$.
  Swapping it with any of these two candidates yields a vote from $D$
  (and no single swap apart from these two results in that).
\end{proof}

\subsubsection{Popularity}
Given a domain $D \subseteq \calL(C)$ and a ranking $v \in D$, we
define its \emph{popularity}, denoted $\pop(v)$, as the number of
rankings from $\calL(C)$ for which $v$ is the closest member of $D$
(if for a given ranking $u \in \calL(C)$ there are $p$ members of $D$
that are closest to $u$, then $u$ contributes $\nicefrac{1}{p}$ to the
popularity of each of them).  The average popularity of a ranking in
$|D|$ is equal to $\nicefrac{|\calL(C)|}{|D|}$ and by \emph{normalized
  popularity} of a ranking $v$ we mean the ratio between its
popularity and this value. Namely, we have
$\npop(v) = \frac{\pop(v)}{|L(C)|/|D|}$.
Popularity
gives hints on both the internal symmetry of a domain, and the
arrangement of its rankings in $\calL(C)$. Indeed, the more uniform
are the popularity values of the rankings, the more likely it is that
they are symmetrically spread within $\calL(C)$. On the other hand, a
mixture of high and low popularity values suggests that the more
popular rankings are on the ``outskirts'' of the domain, and the less
popular ones belong to its ``interior.''
We show the normalized popularities of the rankings in our domains in
\Cref{fig:microscope_m8_without_ic}, on the microscope plots of
\citet{fal-kac-sor-szu-was:c:div-agr-pol-map}.

\begin{remark}\label{rem:microscope}
  Let $D$ be a domain. A microscope plot of $D$ presents each ranking
  from the domain as a dot, whose Euclidean distance from the other
  dots is as similar to the swap distance between the respective
  rankings as possible (exact correspondence between Euclidean
  distances and swap distances is, typically, impossible to achieve,
  but microscopes still give useful intuitions).
\end{remark}

The plots show some remarkable features of our domains. The first
observation is that for both $\GSbal$ and $\GScat$, all rankings have
equal popularity, equal to the expected one. Indeed, this is a general
feature of group separable domains.
% \begin{restatable}{proposition}{ProPopGS}
\begin{proposition}
  Let $D = \GS(\calT)$ be a group separable domain over candidate set
  $C$. Then, for each $v \in \GS(\calT)$, $\npop(v) = 1$.
\end{proposition}
% \end{restatable}
\begin{proof}
    For a contradiction, assume that the thesis does not hold.
    Then, there exist $u,v \in D$ such that $\npop(u) > \npop(v)$.
    Let $\pi : C \rightarrow C$ be a permutation
    such that $v = \pi(u)$,
    where $\pi(u)$ denotes a vote in which each candidate $c \in C$
    is replaced by $\pi(c)$.
    In this way, by a slight abuse of notation,
    $\pi$ is also a permutation of $\mathcal{L}(C)$.
    
    By the definition of $\GS$ domain,
    $\pi$ corresponds to rotating the children
    of certain internal nodes in $\calT$.
    Thus, for every $w \in \mathcal{L}(C)$ it holds that
    $w \in D$ if and only if $\pi(w) \in D$.
    Moreover, for each $w, w' \in \mathcal{L}(C)$
    we have that $\swap(w,w')=\swap(\pi(w),\pi(w'))$.
    Both facts imply that $\npop(u) = \npop(\pi(u))$,
    as $\npop(\cdot)$ is invariant under $\pi$
    (since $\pi$ does not affect the domain, nor the swap distance).
    However, this leads to a contradiction as
    $\npop(u) = \npop(\pi(u)) = \npop(v) < \npop(u).$
\end{proof}

The other domains show a high variance in popularity among their
members.  For example, the most popular rankings in $\SP$ are the
societal axis and its reverse, whereas most rankings in between these
two have low popularity.
Overall,
group-separable domains are perfectly
symmetric and clearly stand out.

\subsection{Outer Diversity for Larger Candidate Sets}\label{sec:out-div-larger}

When considering more than eight candidates, we compute outer
diversity using the sampling approach, with sample size $N = 1000$
(see \Cref{sec:computation}).  For each domain, we repeat this
computation $10$ times, to also obtain standard deviation (it is so
small as to be nearly invisible on our plots, which justifies the use
of sampling).

In \Cref{fig:out-div-m}, we show how the
outer diversity of our domains evolves as a function of the number $m$
of candidates, for $m \in \{2, 3, \ldots, 20\}$. 
In particular, we note that the outer diversity of polynomially-sized
Euclidean domains drops much more rapidly than that of the other,
exponential-sized, ones. It is also notable how $\SPOC$ becomes less
diverse than $\GScat$ (for $9$ candidates or more) and how $\GScat$
becomes the most diverse among our domains (for 12 candidates or
more).  Further, $\GScat$ is consistently more diverse than $\GSbal$.
As these two domains are
extreme among the group-separable ones (one uses the tallest binary
$\GS$-tree and the other one the shortest), we ask if $\GScat$ is the
most diverse group-separable domain and $\GSbal$ is the least diverse
one.

It is interesting if outer diversity of our domains eventually
approaches zero, or if it stays bounded away from it.  As shown below,
the former happens, e.g., if the size of the domain is bounded by a
constant, whereas the latter happens, e.g., for $\GScat$.  Hence,
outer diversity of a domain may be bounded away from zero even if its
size grows notably more slowly than that of the general domain (as a
function of the number of candidates).

% \begin{restatable}{proposition}{ProLimConstSizeDom}
\begin{proposition}
    Let us fix value $k$ and 
    let $D_2$, $D_3$, $\ldots$ be a sequence of domains, where each $D_m$ contains at most $k$ rankings over $m$ candidates. Then $\lim_{m \rightarrow \infty}\out(D_m) = 0$.
\end{proposition}
% \end{restatable}
\begin{proof}
    Let $\UN_m$ denote the $\UN$ election with $m$ candidates.
    As already noted,
    the average normalized swap distance of a domain $D$
    is equal to the normalized Kemeny score of $D$
    with respect to the $\UN_m$ election, i.e.,
    $m!\binom{m}{2}\cdot \ansd(D) = \kemscore_{\UN_m} (D)$.
    This, in turn, is not smaller than
    the $k$-Kemeny score of the $\UN_m$ election,
    where $k = |D|$, which gives us
    $m!\binom{m}{2}\cdot \ansd(D) \ge k\hbox{-}\kemscore(\UN_m).$
    This yields the following bound on the outer diversity:
    \[
        \out(D)= 1 - 2\cdot \ansd(D)
        \le 1 - 2 \cdot \frac{k\hbox{-}\kemscore(\UN_m)}{m!\binom{m}{2}}.
    \]
    \citet[Proposition 3.6]{fal-mer-nun-szu-was:c:map-dap-top-truncated}
    showed that for every $k \in \mathbb{N}$,
    it holds that
    \[
    \lim_{m \to \infty} \frac{k\hbox{-}\kemscore(\UN_m)}{\frac{1}{2}\cdot m!\binom{m}{2}} = 1.
    \]
    Thus,
    \[
        \lim_{m \to \infty} \out(D_m) \le 1 - \lim_{m \to \infty} \frac{k\hbox{-}\kemscore(\UN_m)}{\frac{1}{2}\cdot m!\binom{m}{2}}
        = 0.
    \]
\end{proof}  

\begin{proposition}\label{pro:gs-cat-lower-bound}
  If the number of candidates is even, then
  $\out(\GScat) > \nicefrac{1}{2}$.
\end{proposition}
\begin{proof}
  Take a $\GScat$ domain for candidate set
  $C = \{c_1, \ldots, c_m\}$, where $m$ is even, defined via binary
  caterpillar tree where the leaf closest to the root is $c_1$, the
  next one is $c_2$, and so on.

  Let $v$ be some arbitrary ranking from $\calL(C)$. To transform it
  into a member of $\GScat$ we can, for example, sort its top half in
  the increasing order of the candidate indices, and sort the bottom
  half in the decreasing order of candidate indices.  As shown by
  \citet{boe-bre-elk-fal-szu:c:map-pref-learning}, ensuring that
  candidate indices first increase and then decrease is a necessary
  and sufficient condition for a ranking to belong to $\GScat$.
  The
  number of swaps needed to implement such sorting in the top half of
  the ranking is equal to the number of inversions there.  Since the
  expected number of inversions in a random permutation is
  $\frac{1}{4}n(n-1)$, when considering all votes from $\calL(C)$, on
  average we need to perform
  $\frac{1}{4}(m/2)(m/2-1) = \frac{1}{16}m(m-2)$ swaps in their top
  halves, and the same number of swaps in their bottom
  halves. Altogether, we need to perform $\frac{1}{8}m(m-2)$ swaps per
  ranking in $\calL(C)$, so we have
  $
    \textstyle
    \ansd(\GScat) \leq \frac{\nicefrac{m(m-2)}{8}}{\nicefrac{m(m-1)}{2}} = \frac{1}{4}\cdot\frac{m-2}{m-1}.
  $
  This means that we have
  $\out(\GScat)
  \geq 1 -
  \frac{1}{2}\cdot\frac{m-2}{m-1} > \frac{1}{2}$.
\end{proof}

Experimentally, we checked that for 1000 candidates,
outer diversity of $\SP$ and $\SPOC$
is equal to around $0.039$ and $0.055$, respectively,
raising a question of convergence to zero for these domains as well.

\section{Most Diverse Domains}\label{sec:most-diverse}
Given a number $k$, we ask for a domain of $k$ rankings with the
highest outer diversity value. As per our observation in
\Cref{sec:definition}, we can compute such a domain by solving the
$k$-Kemeny problem for the $\UN$ election using, e.g., integer linear
programming (ILP).\footnote{Finding $k$ rankings that achieve the
  optimal $k$-Kemeny score for $\UN$ can be formulated as the
  \textsc{$k$-Median} clustering applied on the metric space of all possible
  rankings under the $\swap$ distance. We use the standard ILP
  formulation for this problem.} Unfortunately, solving this ILP is
challenging, as its size for $m$ candidates is
$\Theta((m!)^2)$. Hence, for $m \geq 6$ we use the following heuristics
(to compute the outer diversity of the domains produced
by them, we use the sampling approach, with $N = 1000$ samples):
\begin{enumerate}
\item We sample $k$ rankings uniformly at random from $\calL(C)$ (this
  is known as sampling from impartial culture, IC).
\item We sample $k$ rankings from IC and perform simulated annealing 
  (technical details available in the~\Cref{apdx:sec:computing-most-diverse}).
\end{enumerate}
We also use a heuristic that does not allow us to control the size of
the domain, but selects rankings that are spread out over
$\calL(C)$:
\begin{enumerate}
  \setcounter{enumi}{2}
\item We choose a threshold $t \in \{5, 6, \ldots, 25\}$ and keep on
  sampling rankings from IC (altogether $10^4$ of them), keeping only
  those whose swap distance from the closest already-kept one is
  greater or equal to $t$.
\end{enumerate}
Instead of using this heuristic, we would rather keep on 
selecting rankings that are at the largest possible swap distance from
those previously selected, but finding such rankings is $\np$-complete.

\begin{figure}[t]
\centering
  \begin{tabular}{cc}
  
    \begin{subfigure}{0.4\columnwidth}
      \includegraphics[width=1\linewidth]{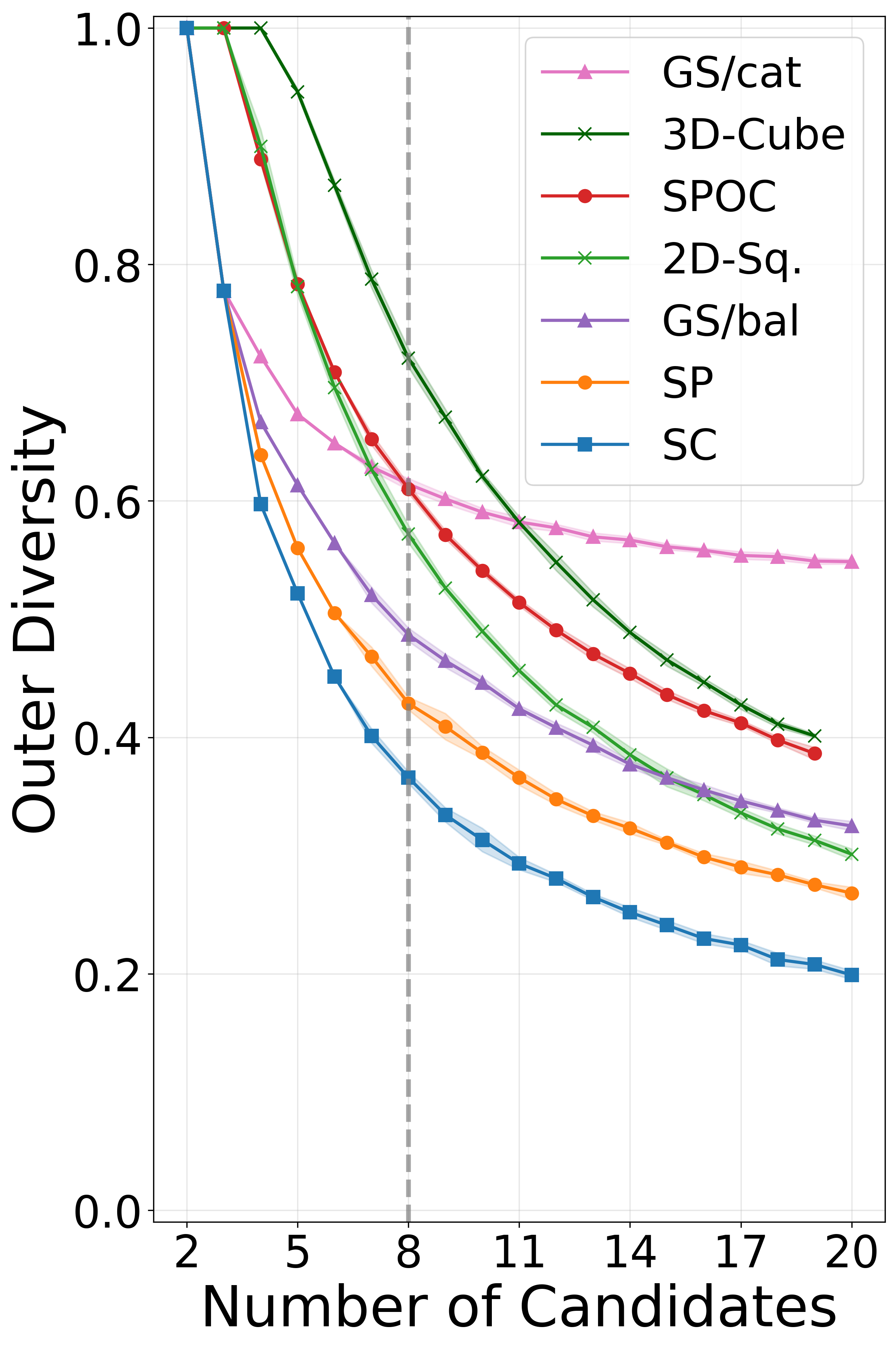}
    \end{subfigure}
    \hspace{0.75cm}
    \begin{subfigure}{0.4\columnwidth}
      \includegraphics[width=1\linewidth]{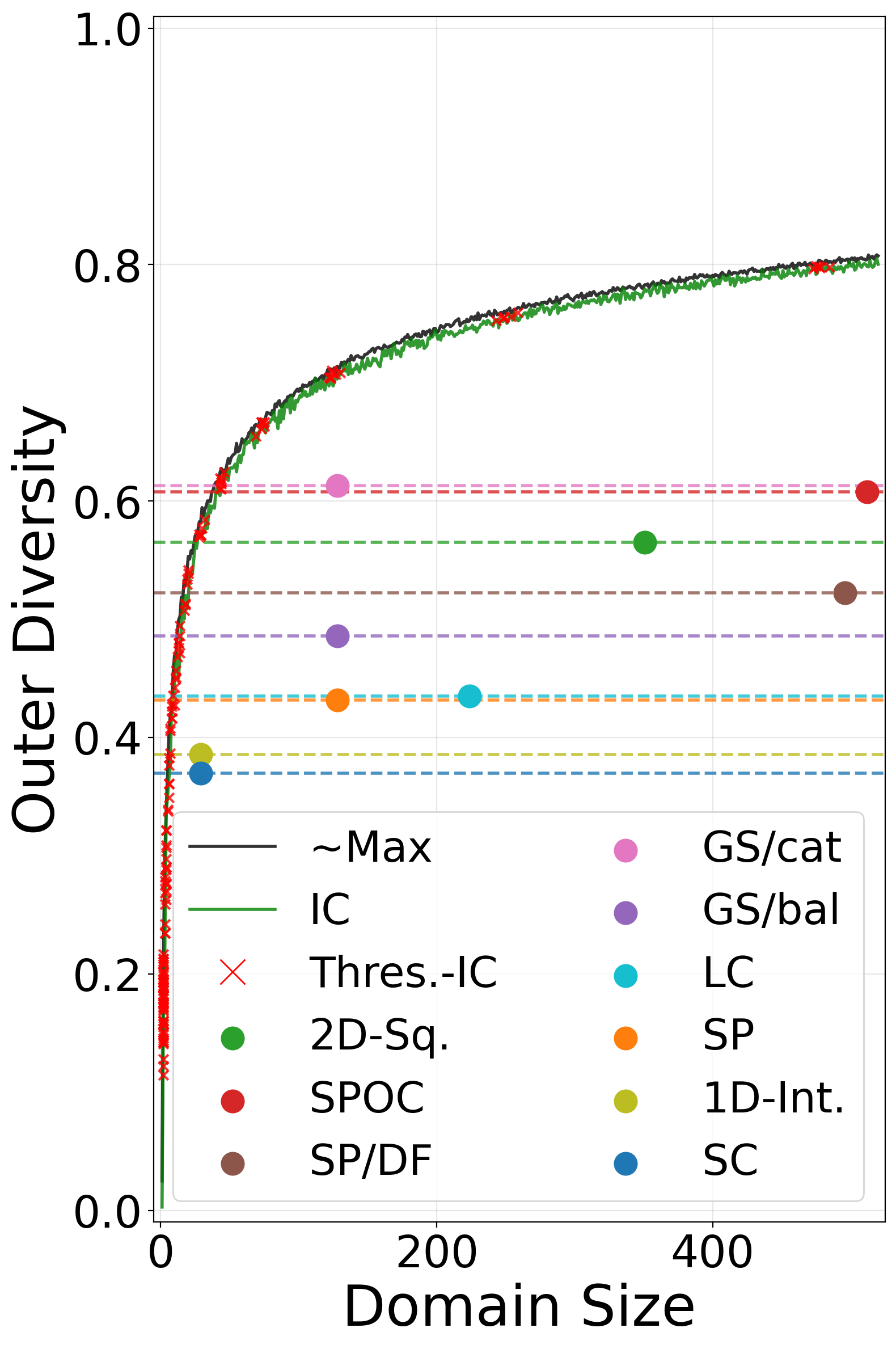}
    \end{subfigure} 
    
  \end{tabular}

  \caption{\label{fig:out-div} Outer diversity of several structured
    domains as a function of the number of candidates (on the left),
    or as a function of their size (on the right; including
    approximations of most diverse domains). For $\boldsymbol\SPOC$ and $\boldsymbol\Cube$, we omit
    outer diversity for $\boldsymbol{20}$ candidates, due to computation time.}
\end{figure}

\begin{theorem}\label{thm:largest-hole-np-hard}
  Given a positive integer $t$ and a domain $D \subseteq \calL(C)$,
  represented by explicitly listing its rankings, deciding if there is
  a ranking $v$ such that $\min_{u \in D}\swap(u,v) \geq t$ is
  $\np$-complete.
\end{theorem}

\begin{figure}[t]
\centering
  \begin{tabular}{ccc}
  
    \begin{subfigure}{0.363\columnwidth}
      \includegraphics[width=0.97\linewidth]{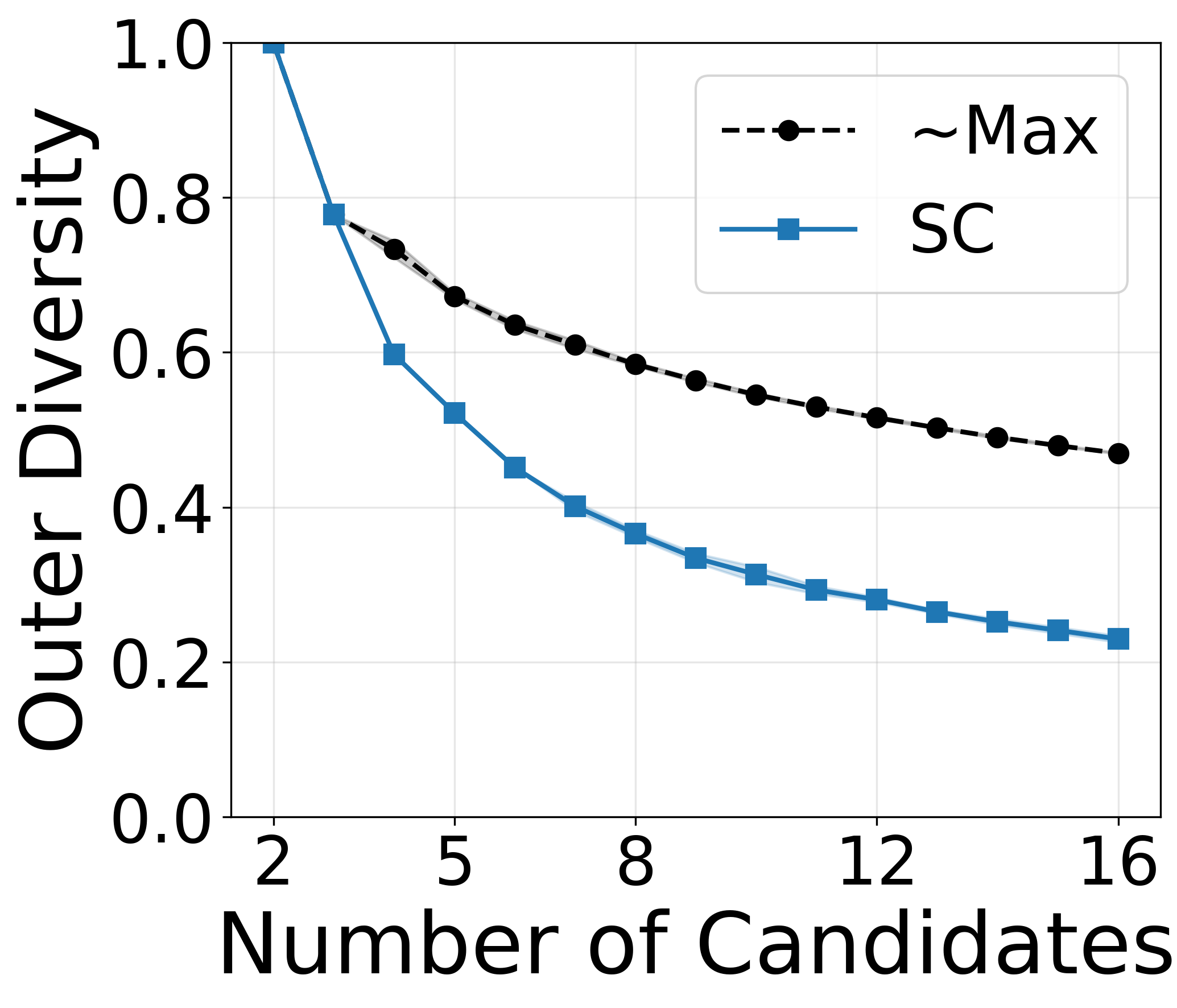}
    \end{subfigure} 
    \begin{subfigure}{0.30\columnwidth}
      \includegraphics[width=0.97\linewidth]{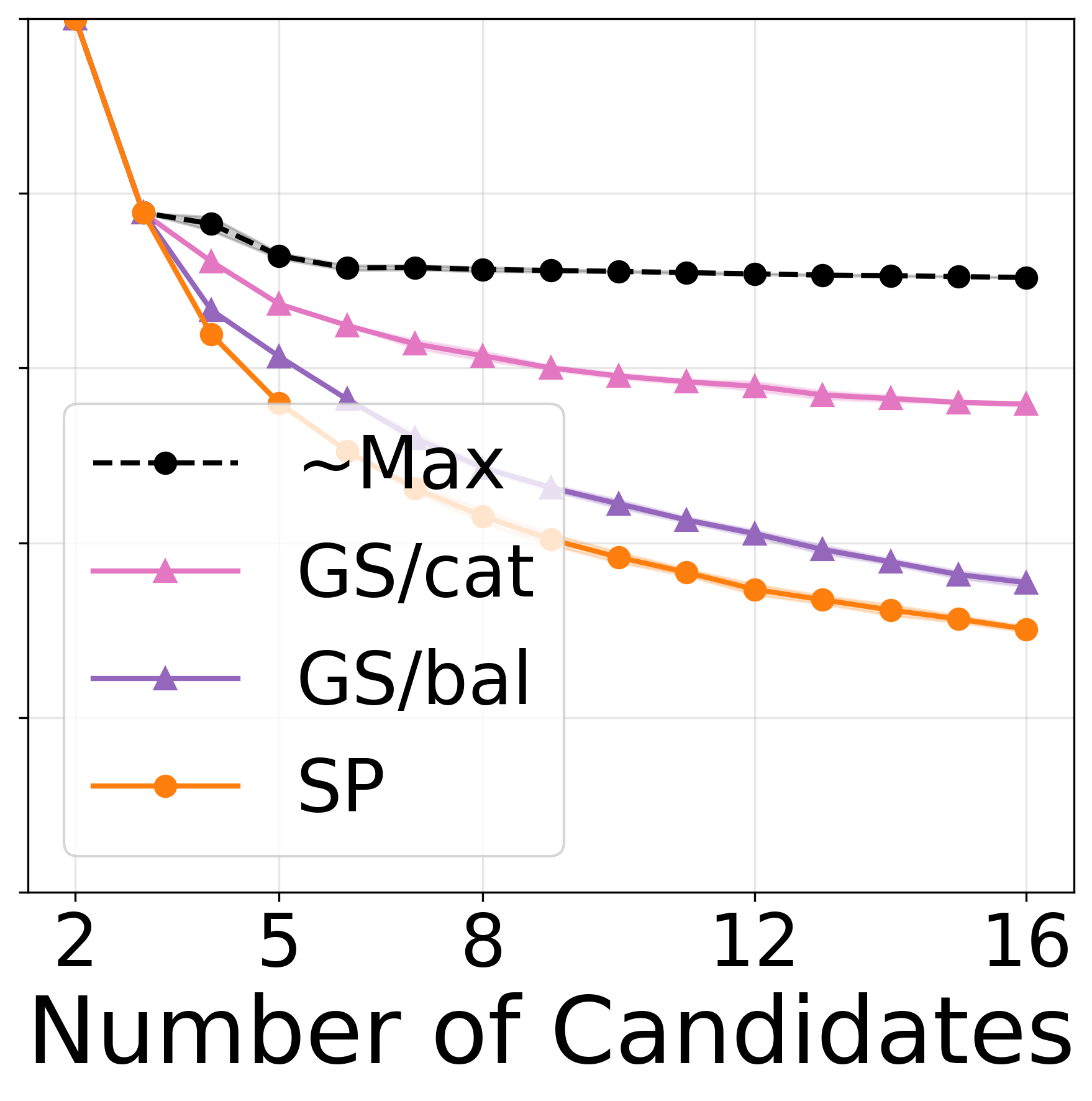}
    \end{subfigure}
    \begin{subfigure}{0.30\columnwidth}
      \includegraphics[width=0.97\linewidth]{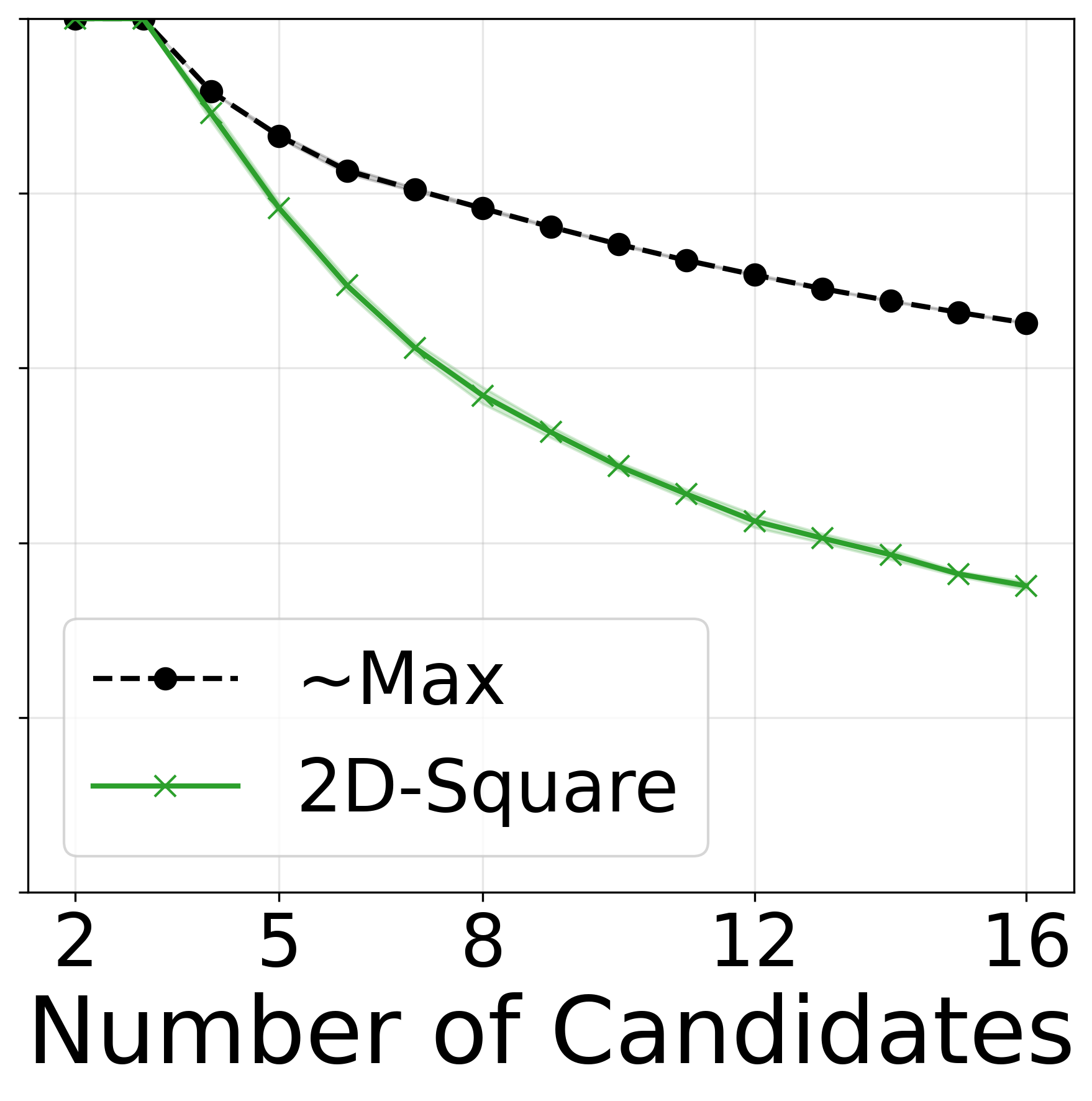}
    \end{subfigure}

  \end{tabular}
    
  \caption{\label{fig:out-div-m}Outer diversity of several structured
    domains as a function of the number of candidates, compared to the
    outer diversity of (an approximation of) the most diverse domain of the
    same size.}
\end{figure}

On the plots, we denote domains computed using the first heuristic as
IC, those computed using simulated annealing as \textasciitilde Max,
and those using the threshold approach as Thres.-IC. In
\Cref{fig:out-div} (right) we show how the outer diversities of these
domains for the case of $m = 8$ candidates, as we increase $k$ (for
the first two heuristics) or decrease $t$ (for the third one). We see
that for each given size of the domain, all three heuristics produce
very similar results. We interpret this as suggesting that, indeed, we
get close to the highest possible diversities. For the case of $6$
candidates we also compared our heuristically computed domains to the
optimal ones, obtained using ILP, and the results were nearly
identical (see \Cref{apdx:sec:computing-most-diverse}). \Cref{fig:out-div} (right) also includes
points corresponding to our structured domains, illustrating how far
off they are from the most diverse domains of their size.

In \Cref{fig:out-div-m}, for each domain $D \in \{\SC$, $\GScat$,
$\GSbal$, $\SP$, $\Square\}$, we plot the outer diversity of this
domain and the outer diversity of the most diverse domain of size
$|D|$ (as computed using our second heuristic) as a function of the
number of candidates (for up to $16$ of them, as beyond this number
computations proved too intensive). In particular, we see that for
polynomial-sized domains ($\SC$ and $\Square$), the diversity of the
most diverse domains seems to be dropping up to $16$ candidates. In contrast,
 for $\SP$, $\GSbal$, and $\GScat$, which are all of size
$2^{m-1}$, the outer diversity of the most diverse domain seems to
stabilize around the value $0.7$ (indeed, by
\Cref{pro:gs-cat-lower-bound}, we know that it cannot go below $0.5$;
proving a stronger bound would be interesting).

\section{Conclusions}

Our main conclusion is that outer diversity is a useful, practical
measure of domain diversity. Using it, we have found that $\GScat$
sharply stands out from many other structured domains in various
respects and, so, we recommend its use in experiments.  Throughout
the paper, we have made a number of observations, and we have explained some
of them theoretically. We propose seeking such explanations for the
remaining observations as future work.
Another interesting direction for future work
is to establish a formal relation between
outer and inner diversity notions.

\section*{Acknowledgments}

Tomasz Wąs was supported by UK Engineering and Physical Sciences Research Council (EPSRC) under grant EP/X038548/1.
This project
has received funding from the European Research Council (ERC) under
the European Union’s Horizon 2020 research and innovation programme
(grant agreement No 101002854).
\begin{center}
  \includegraphics[width=3.5cm]{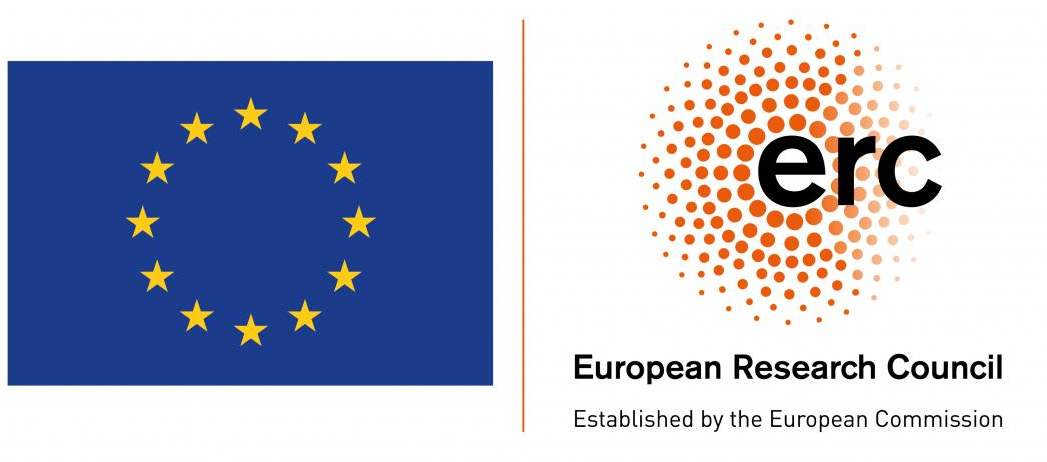}
\end{center}

\bibliography{bib-all}

\clearpage

\appendix

\section{Computing Outer Diversity}
\label{app:computation}

In this appendix,
we provide further details on algorithmic techniques
for establishing outer diversity of given domains.

\begin{algorithm}[t] 
    \caption{Computing outer diversity by BFS}
    \label{alg:growing-balls}
    \begin{algorithmic}[1]{
        \REQUIRE Domain $D$ over candidate set $C = \{c_1, \ldots, c_m\}$ \\ 
        \STATE $D_0 \leftarrow D$, $i \leftarrow 0$
        \WHILE{$\bigcup_{j=0}^i D_i \neq \calL(C)$}
           \STATE $D_{i+1} = \emptyset$
           \FOR{$v \in D_i$}
             \FOR{$u$ such that $\swap(u,v)=1$}
               \STATE \label{gb:core} \textbf{if} $u \notin \bigcup_{j=0}^{i}D_j$: $D_{i+1} \leftarrow D_i \cup \{u\}$
             \ENDFOR
           \ENDFOR
           \STATE $i \leftarrow i+1$
        \ENDWHILE        
        \RETURN $1 - 2(\frac{1}{m!}\sum_{j=0}^{i} j\cdot|D_j|)$}
        \end{algorithmic}
\end{algorithm}

Our exact algorithm, given as \Cref{alg:growing-balls} proceeds as
follows: Given domain $D$, we form a sequence of sets, $D_0$, $D_1$,
$\ldots$, such that for each $i$, $D_i$ contains rankings that are at
swap distance $i$ from $D$. For each $i$, we compute $D_{i+1}$ by
considering all the votes that can be obtained from those in $D_{i}$
by a single swap of adjacent candidates, and include in $D_{i+1}$
those that do not belong to $\bigcup_{j=0}^i D_i$. Given $D_0$, $D_1$,
$\ldots$, we compute $\ansd(D)$ as the weighted sum of their sizes,
and $\out(D)$ as $1-2\ansd(D)$.  Fast implementation requires storing
each $D_i$ individually, as well as the growing union of these sets,
for increasing values of $i$.

\ProBfs*
\begin{proof}
  We use \Cref{alg:growing-balls}, whose correctness
  follows directly from the definitions of $\ansd(D)$ and $\out(D)$.
  In line \ref{gb:core} of the algorithm, for
  each vote $v \in \calL(C)$ we consider all $m-1$ votes $u$
  obtained from $v$ by a single swap of adjacent candidates,
  resulting in $O(m \cdot m!)$ memberships checks.
  Rankings from $\bigcup_{j=0}^{i} D_j$ are stored in a trie (prefix tree),
  which allows $O(m)$-time membership checks and insertions.
  For the current iteration $i$, we store only sets $D_i$ and $D_{i+1}$
  (each taking $O(m \cdot m!)$ space) while retaining the values $|D_j|$ for $j<i$.
  Hence, the computational bottleneck is line \ref{gb:core} executed $O(m \cdot m!)$ times, each taking $O(m)$ time, leading to a total running time of $O(m^2 \cdot m!)$.
\end{proof}

Next, following the sampling approach,
we give detailed descriptions for computing
a distance between a given vote $v$ and domain $D$,
where $D$ is either
SPOC, SP$(T)$, GS/cat, GS/bal, SC, or Euclidean domain.

\subsection{Algorithms for Single-Peaked Domains}
\label{app:computation:sp}

Assume we are given a vote $v$
and single-peaked-on-a-cycle domain with cycle
$(c_1,\dots,c_m)$.
For convenience, we will sometimes allow candidate indices to go over $m$
and treat them as if they cycle over, i.e.,
$c_{i+m} = c_i$ for each $i \in [m]$.

For each $i \in [m]$ and $j \in \{0,1,\dots,m-1\}$,
let $C_{i,i+j}$ denote the set of candidates
$\{c_i,c_{i+1},\dots,c_{i+j}\}$
that form an interval in the cycle.
Then, let $U_{i,i+j}$ be a subset
containing all votes $u$ that
rank candidates from $C_{i,i+j}$ as
top $j+1$ candidates
and for each $t \in [j+1]$
the first $t$ candidates in $u$
form an interval in the cycle.
Also, let $A_{i,i+j}$ denote the minimum swap distance from $v$
to a vote in $U_{i,i+j}$.
These can be efficiently computed using
\Cref{alg:single-peaked-on-a-circle}.

\begin{theorem}
    \label{thm:vote-domain:single-peaked-on-a-circle}
    \Cref{alg:single-peaked-on-a-circle}
    computes the distance between a given vote
    and a single-peaked-on-a-circle domain in time $O(m^2)$.
\end{theorem}
\begin{proof}
For the running time,
we observe that the loops in \Cref{alg:single-peaked-on-a-circle}
are at most $2$-nested,
over at most $m$ elements, and
each individual iteration runs in time $O(1)$.
The final minimum in line 14 runs in time $O(m)$,
but it is not part of any loop.

For the correctness,
similarly as in the proof of \Cref{thm:vote-domain:single-peaked},
we first note that for each $i \in [m], r \in \{0,\dots,m-1\}$
in $L_{i,i+j}$ we store the number of candidates in
$\{c_i,c_{i+1},\dots,c_{i+j}\}$ 
that are preferred over $c_i$ in $v$
(in \Cref{alg:single-peaked} we counted the number of candidates
that are less preferred than $c_i$, here it is reversed).
Analogously, in $R_{i+m-j,i}$,
we store the number of candidates preferred over $c_i$ from
$\{c_{i+m-j},c_{i+m-j+1}, \dots, c_{i+m}\}$
(note that it is also an interval).
We can efficiently compute both sets of numbers recurrently.

For every $i \in [m]$,
$U_{i,i}$ is just the set of all votes that
have $c_i$ as the top candidate.
Thus, $A_{i,i}$ is just a number of candidates
that are preferred over $c_i$ in $v$,
which is what is stored in $L_{i,i+m-1}$
(see line 7 of \Cref{alg:single-peaked-on-a-circle}).

For $i \in [m]$ and $r \in [m-1]$,
we compute values of $A_{i,i+r}$
using a recursive formula in line 10,
in a similar way to how it was done in \Cref{alg:single-peaked}.
Let $U_{\underline{i},i+r}$ be a subset of votes in $U_{i+1,i+r}$
that additionally have $c_{i}$ at the position $r+1$.
Similarly, let $U_{i,\underline{i+r}}$
be a subset of votes in $U_{i,i+r-1}$
with $c_{i+r}$ at the position $r+1$.
Since every vote in $U_{i,i+r}$
has candidates $\{c_i,c_{i+1},\dots,c_{i+r}\}$
at the first $r+1$ positions
and the first $r$ candidates form an interval in the cycle,
it must be $c_i$ or $c_{i+r}$ in the position $r+1$.
Thus,
\(
    U_{i,i+r} = U_{\underline{i},i+r} \cup U_{i,\underline{i+r}}.
\)
Hence,
\(
    A_{i,i+r} = \min(\min_{u \in U_{\underline{i},i+r}} \swap(v,u),
        \min_{u \in U_{i,\underline{i+r}}} \swap(v,u)).
\)

Fix a vote $u \in U_{\underline{i},i+r}$ minimizing $\swap(v,u)$.
Observe that the last $m-r-1$ candidates in $u$
have to appear in exactly the same order as they appear in $v$.
Let $u'$ be obtained from $u$
by arranging $m-r$ last candidates in this way
(so we additionally relocate $c_i$).
Observe that $u'$ minimizes the swap distance to $v$
among votes in $U_{i+1,i+r}$, i.e.,
$A_{i+1,i+r} = \swap(u',v)$.
Moreover, the difference between
$\swap(u,v)$ and $\swap(u',v)$
is the number of candidates outside of $\{c_{i+1},\dots,c_{i+r}\}$
that are preferred over $c_i$.
Observe that
$C \setminus \{c_{i+1},\dots,c_{i+r}\} =
\{c_{i+r+1},\dots,c_{i+m}\}$.
Therefore, we get that
$\swap(u,v) = A_{i+1,i+r} + R_{i+r+1,i}$.
Analogously, we can prove that
$\min_{u \in U_{i,\underline{i+r}}} \swap(v,u) =
A_{i,i+r-1} + L_{i+r,i+m-1}$.
This yields the recursive equation from line 10.

Finally, observe that for each $i \in [m]$,
the set $U_{i,i+m-2}$ contains all single-peaked-on-a-circle votes
that have candidate $c_{i+m-1}$ at the bottom of the ranking.
Thus, taking the minimum of distances to each such vote
we get the minimum distance to any single-peaked-on-a-circle vote.
\end{proof}

\begin{algorithm}[t] 
    \caption{Distance between a ranking and $\SPOC$}
    \label{alg:single-peaked-on-a-circle}
    \begin{algorithmic}[1]{
        \REQUIRE Vote $v \in \mathcal{L}(C)$, societal axis $c_1 \rhd \dots \rhd c_m$ \\ 

        \textsc{Phase 1, Precomputation:}
        \STATE \textbf{for} $i \in [m]$ \textbf{do}
            $c_{i+m} \leftarrow c_i$
        \FOR{$i \in [m]$}
            \STATE $L_{i,i} \leftarrow 0, R_{i+m,i} \leftarrow 0$
            \STATE \textbf{for} $j \in [m-1]$ \textbf{do}
                $L_{i,i+j} \leftarrow L_{i,i+j-1} +
                [c_{i+j} \succ_v c_i]$
            \STATE \textbf{for} $j \in [m-i-1]$ \textbf{do}
                $L_{i+m,i+m+j} \leftarrow L_{i,i+j}$
            \STATE \textbf{for} $j \in [m-1]$ \textbf{do}
                $R_{i+m-j,i} \leftarrow R_{i+m-j+1,i} +
                    [c_{i+m-j} \succ_v c_i]$
        \ENDFOR
        
        \textsc{Phase 2, Main Computation:}
        \STATE \textbf{for} $i \in [m]$ \textbf{do}
        $A_{i,i} \leftarrow L_{i,i+m-1}, \quad
        A_{m+i,m+i} \leftarrow A_{i,i}$
        \FOR{$r \in [m-2]$}
            \FOR{$i \in [2m-r]$}
            \STATE $A_{i,i+r} \leftarrow \min(
                    A_{i,i+r-1} + L_{i+r,i+m-1}, \
                    A_{i+1,i+r} + R_{i+r+1,i})$
            \ENDFOR
        \ENDFOR
        \RETURN $\min_{i = [m]}A_{i,i+m-2}$}
        \end{algorithmic}
\end{algorithm}

We can also extend \Cref{alg:single-peaked}
to work on arbitrary tree with $k$ leaves in time $O(km^k)$.
The pseudocode is summarized in \Cref{alg:single-peaked-on-a-tree}.

\begin{theorem}
    \label{thm:vote-domain:single-peaked-on-a-tree-app}
    \Cref{alg:single-peaked-on-a-tree}
    computes the distance between a given vote
    and a single-peaked-on-a-tree domain in time $O(km^k)$.
\end{theorem}
\begin{proof}
Let us fix such tree $G$
on a set of candidates $C$
and a given vote $v$.
Let $\mathcal{S}=(S_1,S_2,\dots,S_\ell)$
be a sequence of subsets of $C$
such that each $S \in \mathcal{S}$,
if and only if,
$S$ is nonempty and
connected in $G$,
and for each $i,j \in [\ell]$,
we have that $S_i \supseteq S_j$
only if $i < j$.
In particular, this means that $S_1 = C$.
Observe that the length of the sequence $\mathcal{S}$
is bounded by $O(m^k)$,
as each connected subset of $C$
can be uniquely identified by how far away from each leaf
is the closest node from $S$
(and the maximal value of such distance is bounded by $m$).
We can also compute such sequence $\mathcal{S}$
in time $O(m^k)$ by checking each possible
$k$-tuple of such distances
starting from the smallest ones.

Then, for each $S \in \mathcal{S}$
by $X_S \subseteq S$ let us denote the set of nodes in $S$
that are leafs in the subgraph induced by $S$
(note that there are at most $k$ of them).
Furthermore, for each such $x \in X_S$
we denote the number of candidates in $S$
over which $x$ is preferred in $v$ by:
\[
    I_{x,S} = |\{ c \in S : c \succ_v x\}|.
\]
This corresponds to values $L_{i,j}$ and $R_{i,j}$
used in \Cref{alg:single-peaked}.
We can compute all of them in time
$O(km^k)$ in the reversed order to that in sequence $\mathcal{S}$.
This is since for each other leaf $y \in X_S \setminus \{x\}$,
it holds that
\(
    I_{x,S} = I_{x,S \setminus \{y\}} +
    [x \succ_v y].
\)

Next, for each $S \in \mathcal{S}$
we define $U_S$ as a set of all votes $u \in \mathcal{L}(C)$
in which candidates $C \setminus S$ are in the last positions
and for each $t \in [m] \setminus [|S|]$,
the first $t$ candidates in $u$ form a connected subset in $G$.
Also, we denote $A_S = \min_{u \in U_S} \swap(u,v)$.

Clearly, $A_{S_1} = 0$ as $S_1 = C$, thus $U_C = \mathcal{L}(C)$.
For each $S \in (S_2, \dots, S_\ell)$,
we compute $A_{S}$ recursively,
similarly
to how we computed $A_{l,r}$ in \Cref{alg:single-peaked}.
Let $Y_S \subseteq S$ be a subset of nodes in $C \setminus S$
that are connected to some node in $S$
(again, there are at most $k$ of them).
Then, for each $y \in Y_S$,
we can denote $U_{S,y}$ as a subset of votes in $U_{S \cup \{y\}}$
that have $y$ in the position $|S| + 1$.
Since every vote in $U_S$ has to have one of the nodes in $Y_S$
in the position $|S| + 1$, we get that
\(
    U_S = \bigcup_{y \in Y_S} U_{S,y}.
\)
Thus,
$A_S = \min_{y \in Y_S} \min_{u \in U_{S,y}}\swap(u,v).$

Then, as in the proof of \Cref{thm:vote-domain:single-peaked},
we can show that:
\[
    \min_{u \in U_{S,y}}\swap(u,v) =
    A_{S \cup \{y\}} + I_{y,S\cup \{y\}}.
\]
To this end, take $u \in U_{S,y}$ minimizing $\swap(u,v)$
and observe that in $u$
the first $|S|$ candidates are in the same order
in which they appear in $v$.
Let $u'$ be a vote obtained from $u$
by having the first $|S|+1$ candidates
ordered according to $v$
(i.e., candidate $y$ is relocated).
Then, $u'$ actually minimizes $\swap(u,v)$ in $U_{S \cup \{y\}}$
(the first $|S|+1$ candidates are in the optimal order,
and if reordering the last $m - |S|-1$ candidates was possible,
it would also be possible to reorder them in that way in $u$
decreasing the distance).
Finally, $\swap(u,v) - \swap(u',v)$ is the number of candidates from $S$
which are less preferred by $v$ than $y$,
which is what we store in $I_{y,S \cup \{y\}}$.

Observe that in this way,
we have computed $A_S$ for each singleton set
$S = \{c\}$ with $c \in C$.
In $U_{\{c\}}$ we have all votes in the domain that start with $c$.
Thus, taking the minimum over $A_{\{c\}}$
for all $c \in C$
we get the minimum distance in question.
\end{proof}

\begin{algorithm}[t] 
    \caption{Distance between a ranking and $\SP(G)$, where $G$ is a tree}
    \label{alg:single-peaked-on-a-tree}
    \begin{algorithmic}[1]{
        \REQUIRE Vote $v \in \mathcal{L}(C)$, tree $G$ with $C$ as nodes \\ 

        \textsc{Phase 1, Precomputation:}
        \STATE $\mathcal{S} = (S_1,\cdots,S_\ell) \leftarrow$ a sequence of subsets of $C$, such that:

        \hspace{2.5cm} $S \in \mathcal{S} \Leftrightarrow S \neq \varnothing$ and $ S$ connected in $G$ 
        
        \hspace{2.5cm} $S_i \supseteq  S_j \Rightarrow i <j$ 
        \FOR{$S \in (S_\ell, S_{\ell - 1}, \dots, S_1)$}
            \STATE $X_S \leftarrow$ leaves in graph induced by $S$
            \FOR{$x \in X_S$}
                \IF{$S = \{x\}$}
                    \STATE $I_{x,S} \leftarrow 0$
                \ELSE
                    \STATE $y \leftarrow$ arbitrary node from $X_S \setminus \{x\}$
                    \STATE $I_{x,S} = I_{x,S \setminus \{y\}} +
                        [x \succ_v y]$
                \ENDIF
            \ENDFOR
        \ENDFOR
        
        \textsc{Phase 2, Main Computation:}
        \STATE $A_{S_1} \leftarrow 0$
        \FOR{$S \in (S_2,\dots,S_\ell)$}
            \STATE $Y_S \leftarrow$ nodes in
            $C \setminus S$ connected to $S$
            \STATE $A_{S} \leftarrow \min_{y \in Y_S}(A_{S \cup \{y\}} + I_{y,S \cup \{y\}})$
        \ENDFOR
        \RETURN $\min_{c \in C}A_{\{c\}}$}
        \end{algorithmic}
\end{algorithm}

\subsection{Algorithms for Group-Separable Domains}
\label{app:computation:gs}
      
Now, let us look at the specific cases of GS/bal and GS/cat.
Let as consider $\GScat$ first, and let $v$ be the ranking whose swap
distance from $\GScat$ we want to compute.  We use an algorithm very
similar to the general one, but processing the internal nodes in the
decreasing order of their distance from the root, and using additional
data structures. Namely, when we consider an internal node whose
children are a leaf associated with some candidate $c$ and a subtree
whose leaves hold candidates from the set
$C' = \{c'_1, \ldots, c'_t\}$, then we assume that we also have a data
structure that for each $c'_i \in C$ holds the position that $c'_i$
has in $v$. We require that it is possible to insert positions into
this data structure in time $O(\log m)$ and that this data structure
can also answer in $O(\log m)$ time how many of the positions that it
stores are earlier in $v$ than a given one (so, this data structure
can be, e.g., a classic red-black tree, annotated with sizes of its
subtrees~\citep{cor-lei-riv-ste:b:algorithms-second-edition}).  Now,
we can simply query the data structure for the number $\mathit{inv}$
of candidates in $D$ that are ranked ahead of $d$ (i.e., whose
position is smaller than $\pos_v(d)$). This is the number of
inversions imposed by the current node in case we order its children,
so that in the frontier we have $\{d\} \pref C'$.  $t-\mathit{inv}$ is
the number of inversions imposed in the reversed configuration.  We
implement the configuration that leads to fewer inversions (or we
choose one arbitrarily in case of a tie), we insert $\pos_v(d)$ into
the data structure, and we proceed to the parent node of the current
one (or terminate, in case the current node was a root).
The correctness follows from the correctness of the general algorithm.
The running time follows from the fact that the tree has $O(m)$
internal nodes, and for each of them we need time $O(\log m)$.

\begin{theorem}
  There is an algorithm that computes the distance between a given
  vote and $\GScat$ (represented via a $\GS$-tree) in time
  $O(m \log m)$.
\end{theorem}

In case
of $\GSbal$, we proceed similarly as in the classic Merge Sort
algorithm. Let $\calT$ be a balanced $\GS$-tree and let $v$ be the
ranking under consideration.  As above, our algorithm manipulates the
ordering of the children of each node, to obtain a tree whose frontier
$u$ minimizes $\swap(u,v)$.
We use a recursive procedure that given an internal node $z$ with two
children, $z_\ell$ on the left and $z_r$ on the right, such that
$A = \{a_1, \ldots, a_x\}$ is the set of candidates associated with
the leaves of the tree rooted at $z_\ell$ and
$B = \{b_1, \ldots, b_y\}$ is the set of candidates associated with
the leaves of the tree rooted at $z_r$, proceeds as follows:
\begin{enumerate}
\item It calls itself recursively on $z_\ell$ and $z_r$ (unless a
  given subtree is a leaf). These calls order the children within the
  respective subtrees to minimize the number of inversions between the
  candidates in $A$ and $v$ and between the candidates in $B$ and
  $v$. Additionally, they return rankings $v_A$ and $v_B$ that are
  equal to $v$ restricted to $A$ and $B$, respectively.  Without loss
  of generality, we assume that $v_A$ orders the candidates in $A$
  according to their indices, and so does $v_B$ for the candidates in
  $B$
  
\item We perform the ``merge'' step, to decide whether to reverse the
  order of children of $z$ and to obtain $v_{A \cup B}$ (i.e., $v$
  restricted to the candidates in $A \cup B$).  We first consider the
  case where we do not reverse the order of $z$'s children.
  Initially, we set the number of inversions between to be $0$ and,
  then, we fill-in $v_{A \cup B}$ from the top position to the bottom
  one, by considering the prefixes of $v_A$ and $v_B$. Suppose that we
  have already filled-in the top $k-1$ positions in $v_{A \cup B}$
  with $i-1$ candidates from $A$ and $j-1$ candidates from $B$. The
  candidate on the $k$-th position in $v_{A \cup B}$ will either be
  the $i$-th candidate from $v_A$ or the $j$-th candidate from $v_B$,
  i.e., either $a_i$ or $b_j$. If $a_i \pref_v b_j$ then we choose
  $a_i$, and otherwise we choose $b_j$ and increase the number of
  inversions by $x-(i-1)$ because, in this configuration, in the
  frontier of our tree $b_j$ is ranked below
  $a_i, a_{i+1}, \ldots, a_{x}$. After we use up all the candidates of
  $A$ or $B$, then we fill-in $v_{A \cup B}$ with those from the other
  set, in the order in which they appear in $v_A$ or $v_B$,
  respectively.
  Let $\mathit{inv}$ be the computed
  number of inversions. If we reversed the order of children of $z$,
  then then number of inversions would be $|A||B|-\mathit{inv}$; if
  this value is smaller than $\mathit{inv}$ then we reverse the children.
  Finally, we output $v_{A \cup B}$. 
\end{enumerate}
Our algorithm executes this procedure on the root of the tree. The
correctness is immediate, whereas the running time of $O(m \log m)$
follows from the fact that $\GSbal$ trees have $O(\log m)$ levels, and
on each level, the merge steps require $O(m)$ steps.

\begin{theorem}
  There is an algorithm that computes the distance between a given
  vote and $\GSbal$ (represented via a $\GS$-tree) in time
  $O(m \log m)$.
\end{theorem}

\subsection{Algorithms for Single-Crossing and Euclidean Domains}
\label{app:computation:sc-euc}

\subsubsection{Single-Crossing}
\label{app:computation:sc}

Single-crossing domain contains $O(m^2)$ votes,
thus computing swap distance to each of them
and taking the minimum would give $O(m^3\sqrt{\log m})$ algorithm~\cite{cha-pat:c:fast-swap-distance}.

However, in practice,
we often want to compute the distance 
from multiple given votes to a single fixed domain.
For that case,
we present an algorithm that needs
a preprocessing step that also runs in time $O(m^3\sqrt{\log m})$ (this time due to the bottleneck in recognizing a single-crossing ordering of voters),
but then, for each input vote,
allows for computation of the distance in $O(m^2)$.

The preprocessing step involves sorting the votes in the domain
in a sequence that witnesses the single-crossingness
$(u_0,u_1,\dots,u_{M})$, where $M= \binom{m}{2}$.
This can be done in time $O(Mm\sqrt{\log M}) = O(m^3\sqrt{\log m})$~\cite{elk-lac-pet:t:restricted-domains-survey,bah-con-wat:c:fastest-sc-recognition}.
Next, for each $i \in [M]$,
we establish the unique pair of candidates $(a_i, b_i) \in C \times C$
such that $a_i \succ_{u_i} b_i$
but $b_i \succ_{u_{i-1}} a_i$.
We can establish all of them in time $O(m^2\log m)$
by looking at each pair of candidates 
and finding the place where its ordering switches using binary search.

Now, for each input vote $v \in \mathcal{L}(C)$
we first find the vector $\pos_v$
in which we keep the position
of every candidate in $C$
according to $v$.
This can be computed in time $O(m\log m)$
by sorting the arguments of the list in which we store vote $v$.
Next, we compute $\swap(u_0,v)$,
again in time $O(m\sqrt{\log m})$~\cite{cha-pat:c:fast-swap-distance}.
Further, for each $i \in [M]$,
we check whether $\pos_v(a_i) < \pos_v(b_i)$.
If it holds, then it means that in
$u_i$ candidates $a_i$ and $b_i$
are ordered in the same way as in $v$,
which is the opposite ordering to that in $u_{i-1}$.
Since all other pairs are ordered in the same way in $u_i$ and $u_{i-1}$, we get that
$\swap(u_i,v) = \swap(u_{i-1},v) - 1$.
If $\pos_v(a_i) > \pos_v(b_i)$ holds,
then analogously $\swap(u_i,v) = \swap(u_{i-1},v) + 1$.
In this way, we can compute swap distance from $v$
to each of $u_0,u_1,\dots,u_M$ in time $O(m^2)$.
Finally, we output the minimum of these values.

\subsubsection{Euclidean}
\label{app:computation:euc}

For Euclidean elections
we proceed largely analogous to
how we treated single-crossing elections.
We know that there are $O(m^{2d})$ votes in the domain,
where $d$ is the dimension of the Euclidean space.
Hence, the brute-force algorithm of computing the distances directly
and taking the minimum would give us running time $O(m^{2d+1}\sqrt{\log m})$.
However, we can find an alternative algorithm
with $O(m^{2d+2})$ preprocessing step
and $O(m^{2d})$ running time for each input vote.

For the preprocessing step we construct a graph
in which the votes in the domain are vertices
and the edge appears when
the swap distance between two votes is equal to $1$.
Then, let $(u_0,u_1,\dots,u_M)$
be a sequence of votes we get when we run a DFS on this graph.
Also, for each $i \in [M]$ let $p_i$ denote the parent of $u_i$ in the spanning tree that we get as a result.
Then, let $(a_i, b_i) \in C \times C$ be a unique pair of candidates
such that $a_i \succ_{u_i} b_i$
but $b_i \succ_{p_i} a_i$.

We construct the graph and identify the associated pair of candidates on each edge
in overall time $O(m^{2d+2})$.
To do so, we first organize all votes in the domain using a trie (prefix tree) in time $O(m^{2d+1})$, which enables lexicographic ordering and allows membership checks in $O(m)$ time.
Next, for each vote
and for each pair of consecutive candidates in a node,
we check whether the vote obtained by swapping this pair belongs to the domain.
DFS consider at most $O(m^{2d+1})$ many edges.
Since each membership check takes $O(m)$, the total running time is $O(m^{2d+2})$.

Now, as in \Cref{app:computation:sc}
for each input vote $v \in \mathcal{L}(C)$
we first find the vector $\pos_v$
with position of every candidate in $C$
according to $v$,
which we compute in time $O(m\log m)$.
Then, in time $O(m\sqrt{\log m})$~\cite{cha-pat:c:fast-swap-distance}
we compute $\swap(u_0,v)$.
Next, iteratively, for each $i \in [M]$
we check whether $\pos_v(a_i) < \pos_v(b_i)$.
If yes, $\swap(u_i,v) = \swap(p_i,v) - 1$,
otherwise, $\swap(u_i,v) = \swap(p_i,v) + 1$.
In this way, we compute the swap distances between $v$
and all the votes in the domain in time $O(m^{2d})$.
Finally, we output the minimum of these values.

\subsection{Hardness for Single-Peaked-on-a-Graph Domains}
\label{app:spog:hardness}
In this section,
we provide a complete proof of \Cref{thm:spog:dist:hardness}
that finding a distance to an arbitrary single-peaked-on-a-graph domain
is $\np$-complete.

\ThmSpogDistHardness*
\begin{proof}
    If we are given a ranking $u \in SP(G)$
    that is the closest to the given vote, $v$,
    then checking if $\swap(u,v) \le d$
    can be done in polynomial time.
    Thus, the problem belongs to $\np$.
    Hence, in the remainder of the proof,
    we focus on showing the hardness.

    To this end, we will provide a reduction from \textsc{SetCover}.
    In this problem, we are given a universe of elements
    $\mathcal{U} = \{u_1,\dots,u_n\}$,
    a family of $\mathcal{U}$'s subsets
    $\mathcal{S} = \{S_1,\dots,S_m\}$, and
    an integer $k \in \mathbb{N}$.
    The question is whether there exists
    a subset $K \subseteq \mathcal{S}$ of size $|K|=k$,
    known as a \emph{set cover},
    that contains all elements from the universe, i.e.,
    $\bigcup_{S_j \in K} S_j = \mathcal{U}$.
    Answering this question
    is known to be $\np$-complete~\cite{gar-joh:b:int}.
    Without loss of generality, we assume that $n>2$ and $m > k$.

    For each instance of \textsc{SetCover}
    we construct an instance of our problem as follows
    (see \Cref{fig:thm:spog:dist:hardness} for an illustration).
    We let the set of candidates $C = V(G)$ contain three groups of candidates:
    (1) $n \cdot m$ \emph{element candidates}
    $(c_{i,j})_{i \in [n], j \in [m]}$,
    among which $n$ candidates, $(c_{i,1})_{i \in [n]}$,
    are called \emph{first element candidates};
    (2)
    $m$ \emph{subset candidates}
    $s_1,\dots,s_m$;
    and
    (3) $n^2 \cdot m^2 + 1$ \emph{path candidates}
    $p_0,p_1,\dots,p_{n^2 \cdot m^2}$.
    As for the edges in graph $G$,
    for each $i \in [n]$, we connect all
    element candidates $c_{i,1},\dots,c_{i,m}$,
    to form a path,
    and the same we do with all of
    the path candidates $p_0,p_1,\dots, p_{n^2 \cdot m^2}$.
    Additionally, we connect $p_0$ 
    to all set candidates $s_1,\dots,s_m$.
    Finally, for each $j \in [m]$, we connect set candidate $s_j$
    to the first element candidates
    with indices corresponding 
    to the indices of the elements of subset $S_j$.
    Formally,
    \begin{align*}
        E(G) = \ &\{\{c_{i,j},c_{i,j+1}\} : i \in [n],j\in[m-1]\} \\
            \cup \ &\{\{p_{i-1},p_i\} : i \in [n^2 \cdot m^2]\} \\
            \cup \ &\{\{p_0, s_j \} : j \in [m]\} \\
            \cup \ &\{\{s_j, c_{i,1} \} : j\in [m], u_i \in S_j\}.
    \end{align*}
    In the given input vote $v$,
    the top candidate is $p_0$,
    followed by all element candidates,
    then remaining path candidates,
    and lastly the subset candidates at the bottom of the ranking, i.e.,
    \begin{align*}
        &p_0 \succ_v 
        c_{1,1} \succ_v c_{1,2} \succ_v \dots \succ_v  c_{1,m}  \succ_v\\ 
        &c_{2,1}  \succ_v  \dots  \succ_v  c_{n-1,m}  \succ_v
        c_{n,1}  \succ_v c_{n,2}  \succ_v  \dots  \succ_v c_{n,m}  \succ_v\\
        &p_1  \succ_v  p_2  \succ_v \dots  \succ_v p_{n^2 \cdot m^2}  \succ_v 
        s_1  \succ_v  s_2  \succ_v  \dots  \succ_v  s_m.
    \end{align*}
    Finally, we set $d = (k+1) n^2 \cdot m^2 - 1$.

    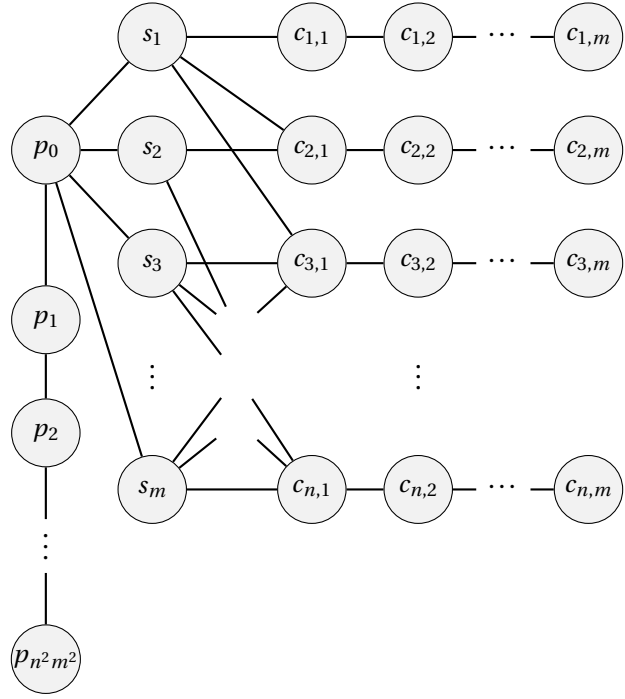
\begin{figure}
    \centering
    \begin{tikzpicture}
        \def\xs{1.4cm} %size of a graph
        \def\ys{1.5cm} %size of a graph
        \def\x{0cm} %placement of a graph
        \def\y{0cm} %placement of a graph
        \def\ls{\footnotesize} %letter sizes
        
        \tikzset{
            node/.style={circle, draw, minimum size=0.9cm, inner sep=0, fill = black!05, font=\small},
            white/.style={circle, draw=white, minimum size=0.7cm, inner sep=0},
            empty/.style={},
            edge/.style={draw, thick, sloped, -, above, font=\footnotesize},
        }
        \node[node] (p0) at (\x + -0.5*\xs, 0*\ys + \y) {$p_0$};
        \node[node] (p1) at (\x + -0.5*\xs, -1.5*\ys + \y) {$p_1$};
        \node[node] (p2) at (\x + -0.5*\xs, -2.5*\ys + \y) {$p_2$};
        \node[white] (p_) at (\x + -0.5*\xs, -3.5*\ys + \y) {};
        \node[empty] () at (\x + -0.5*\xs, -3.5*\ys + 0.1cm + \y) {$\vdots$};
        \node[node] (pnm) at (\x + -0.5*\xs, -4.5*\ys + \y) {$p_{n^2m^2}$};
        
        \node[node] (s1) at (\x + 0.5*\xs, 1*\ys + \y) {$s_1$};
        \node[node] (s2) at (\x + 0.5*\xs, 0*\ys + \y) {$s_2$};
        \node[node] (s3) at (\x + 0.5*\xs, -1*\ys + \y) {$s_3$};
        \node[white] (s_) at (\x + 0.5*\xs, -2*\ys + \y) {};
        \node[white] (s_) at (\x + 0.5*\xs, -2*\ys + 0.1cm + \y) {$\vdots$};
        \node[node] (sm) at (\x + 0.5*\xs, -3*\ys + \y) {$s_m$};

        \node[white] (A) at (\x + 1.3*\xs, -2*\ys + 0.6cm + \y) {};
        \node[white] (B) at (\x + 1.3*\xs, -2*\ys + \y) {};
        \node[white] (C) at (\x + 1.3*\xs, -2*\ys - 0.6cm + \y) {};

        \node[node] (c11) at (\x + 2*\xs, 1*\ys + \y) {$c_{1,1}$};
        \node[node] (c21) at (\x + 2*\xs, 0*\ys + \y) {$c_{2,1}$};
        \node[node] (c31) at (\x + 2*\xs, -1*\ys + \y) {$c_{3,1}$};
        \node[white] (c_) at (\x + 3*\xs, -2*\ys + \y) {};
        \node[white] (c_) at (\x + 3*\xs, -2*\ys + 0.1cm + \y) {$\vdots$};
        \node[node] (cn1) at (\x + 2*\xs, -3*\ys + \y) {$c_{n,1}$};

        \node[node] (c12) at (\x + 3*\xs, 1*\ys + \y) {$c_{1,2}$};
        \node[white] (c1_) at (\x + 3.8*\xs, 1*\ys + \y) {$\cdots$};
        \node[node] (c1m) at (\x + 4.6*\xs, 1*\ys + \y) {$c_{1,m}$};
        \node[node] (c22) at (\x + 3*\xs, 0*\ys + \y) {$c_{2,2}$};
        \node[white] (c2_) at (\x + 3.8*\xs, 0*\ys + \y) {$\cdots$};
        \node[node] (c2m) at (\x + 4.6*\xs, 0*\ys + \y) {$c_{2,m}$};
        \node[node] (c32) at (\x + 3*\xs,-1*\ys + \y) {$c_{3,2}$};
        \node[white] (c3_) at (\x + 3.8*\xs, -1*\ys + \y) {$\cdots$};
        \node[node] (c3m) at (\x + 4.6*\xs, -1*\ys + \y) {$c_{3,m}$};
        \node[node] (cn2) at (\x + 3*\xs, -3*\ys + \y) {$c_{n,2}$};
        \node[white] (cn_) at (\x + 3.8*\xs, -3*\ys + \y) {$\cdots$};
        \node[node] (cnm) at (\x + 4.6*\xs, -3*\ys + \y) {$c_{n,m}$};
        
        \path[edge]
        (p0) edge (p1)
        (p1) edge (p2)
        (p2) edge (p_)
        (p_) edge (pnm)
        (p0) edge (s1)
        (p0) edge (s2)
        (p0) edge (s3)
        (p0) edge (sm)
        (s1) edge (c11)
        (s1) edge (c21)
        (s1) edge (c31)
        (s2) edge (c21)
        (s2) edge (A)
        (s3) edge (c31)
        (s3) edge (A)
        (s3) edge (B)
        (sm) edge (B)
        (sm) edge (C)
        (sm) edge (cn1)
        (c31) edge (A)
        (cn1) edge (B)
        (cn1) edge (C)
        (c11) edge (c12)
        (c12) edge (c1_)
        (c1_) edge (c1m)
        (c21) edge (c22)
        (c22) edge (c2_)
        (c2_) edge (c2m)
        (c31) edge (c32)
        (c32) edge (c3_)
        (c3_) edge (c3m)
        (cn1) edge (cn2)
        (cn2) edge (cn_)
        (cn_) edge (cnm)
        ;
      
    \end{tikzpicture}
    \caption{An illustration of the construction from the proof of \Cref{thm:spog:dist:hardness}.}
    \label{fig:thm:spog:dist:hardness}
\end{figure}

    First, let us show that
    if there exists a set cover $K$ in the original instance,
    then $\swap(\SP(G),v)\le d$.
    Let $K'$ contain all the subset candidates
    corresponding to subsets in $K$, i.e.,
    $K' = \{s_j : S_j \in K\}$.
    Then, let $u$ be a vote obtained from $v$
    by moving all subset candidates in $K'$
    upwards in the ranking so that
    all of them are between candidates $p_0$ and $c_{1,1}$
    (the ordering of the remaining candidates is the same).
    For each $s_j \in K'$,
    there are exactly $n\cdot m$ element candidates,
    $n ^2 \cdot m^2$ path candidates, and at most
    $m$ subset candidates for which the ordering in $u$ and $v$ is different.
    Hence,
    \begin{align*}
        \swap(u,v)
        &\le k n^2m^2 + knm + km \\
        &= kn^2m^2 + km(n+1) \\
        &\le (k+1)n^2m^2 - 1 \\
        &= d,
    \end{align*}
    where the last inequality comes from
    our assumption that $n \ge 3$ and $m > k$.
    Moreover, we can observe that $u$ belongs to the $\SP(G)$ domain.
    Indeed, for each $t \in [k+1]$,
    the first $t$ candidates in $u$ form a connected subgraph in $G$,
    as each subset candidate is connected to $p_0$.
    Then, for each $t \in \{k+2, \dots, k+1+n\cdot m\}$,
    each element candidate $c_{i,j}$ for $i \in [n], j \in [m]$
    is connected through the path of element candidates
    (that all appear before it in $u$)
    to the first element candidate $c_{i,1}$,
    which in turn is connected to some $s_j \in K'$
    (as there is $S_j \in K$ such that $u_i \in S_j$).
    Finally, for $t > k+1+n\cdot m$,
    every path candidate $p_i$ for $i \in [n^2m^2]$
    is connected to $p_0$ through a path of path candidates
    (that all appear before it in $u$),
    and every set candidate $s_j \not \in K'$
    is connected directly to $p_0$.
    Therefore, indeed,
    $\swap(\SP(G),v) \le \swap(u,v) \le d$.

    In the remainder of the proof,
    let us assume that there is no set cover
    in the original \textsc{SetCover} instance
    and let us show that this implies that $\swap(\SP(G),v) > d$.
    Take an arbitrary vote $u \in \SP(G)$.
    Let $i^* \in [n]$ be such that
    $c_{i^*,1}$ is the least preferred
    among the first element candidates in $u$, i.e.,
    $c_{i,1} \succ_u c_{i^*,1}$, for each $i\in [n] \setminus \{i^*\}$.
    
    Observe that it must hold that
    $c_{i^*,1} \succ_u c_{i^*,2} \succ_u \dots \succ_u c_{i^*,m}$
    (otherwise, if there is $j < j' \le m$ such that
    $c_{i^*, j'} \succ_u c_{i^*, j}$,
    then the subset of the first $t$ candidates up to $c_{i^*, j'}$
    is not connected in $G$, as $c_{i^*, j}$
    is on the only path from $c_{i^*, j'}$ to $c_{i, 1}$
    for any $i \in [n] \setminus \{i^*\}$).
    Moreover, by the definition of $\SP(G)$ domain,
    the set of candidates that are weakly preferred over $c_{i^*,1}$ in $u$,
    i.e., $C' = \{c \in C : c \succ_u c_{i^*,1} \} \cup \{c_{i^*,1}\}$,
    must form a connected subgraph in $G$.
    This means that $C'$ must contain
    at least one subset candidate connected
    to each first element candidate.
    Let $K' = \{s_1,\dots,s_m\} \cap C'$ be a subset of
    all subset candidates in $C'$
    and let $K = \{S_j : s_j \in K'\}$
    be a set of corresponding subsets in $\mathcal{S}$.
    We know that $K$ covers all the elements in $\mathcal{U}$,
    but since there is no set cover of size $k$,
    in the \textsc{SetCover} instance,
    we get that $|K'| = |K| \ge k+1$.
    
    Now, let $P$ denote the set of path candidates,
    excluding $p_0$ that are ranked above $c_{i^*,1}$ in $u$,
    i.e., $P = C' \cap \{p_1,\dots,p_{n^2\cdot m^2}\}$.
    Then, there are at least $m \cdot |P|$ pairs
    of an element candidate from the path $(c_{i^*,j})_{j \in [m]}$
    and a path candidate in $P$
    that are ordered differently in $u$ and $v$.
    Moreover, there are at least $(k+1) \cdot (n^2 \cdot m^2 - |P|)$ pairs
    of a subset candidate in $K'$
    and a path candidate in $\{p_1,\dots,p_{n^2\cdot m^2}\} \setminus P$
    that are ordered differently in $u$ and $v$.
    Hence,
    \begin{align*}
        \swap(u,v)
        &\ge m \cdot |P| + (n^2 \cdot m^2 - |P|) \cdot (k+1) \\
        &\ge (k+1) \cdot |P| + (n^2 \cdot m^2 - |P|) \cdot (k+1) \\
        &= n^2 \cdot m^2 \cdot (k+1) \\
        &> d,
    \end{align*}
    where the second inequality comes from our assumption that $m \ge k+1$.
    This concludes the proof.
\end{proof}

\section{Most Diverse Domains}\label{apdx:sec:most-diverse}
Below, we provide a formal definition for the \textsc{Most Diverse Domain}.

\begin{definition}
  For a set of candidates $C$ and an integer $k \leq |C|!$ the \textsc{Most Diverse Domain} problem asks for a set $D \subseteq \mathcal{L}(C)$ of size $k$ that maximizes $\out(D)$.    
\end{definition}

We observe that an optimal solution to \textsc{Most Diverse Domain} is a set of $k$ rankings that achieves the optimal $k$-Kemeny score for the election $(C,\mathcal{L}(C))$.
Moreover, finding $k$ rankings that realize the optimal $k$-Kemeny score of $(C,\mathcal{L}(C))$ can be formulated as the classic clustering problem \textsc{$k$-Median} of the metric space of all possible rankings together with the $\swap$ distance, i.e., $(\mathcal{L}(C), \swap)$. %add cite dap+cos z kmedian

To compute optimal solutions for \textsc{Most Diverse Domain}, we used a standard Integer Linear Program (ILP) for \textsc{$k$-Median}.
Unfortunately, this approach is computationally expensive because the ILP has size $\Theta((m!)^2)$.
A faster, heuristic alternative is simulated annealing:
We initialize a random set of $k$ rankings and iteratively attempt to improve the solution by replacing a single ranking to reduce the total swap distance.
This heuristic appears surprisingly effective, likely because randomly sampling $k$ rankings from the impartial culture model already yields near-optimal solutions, especially for large $k$.

For completeness, we provide an ILP formulation for \textsc{Most Diverse Domain} below that is equivalent to an ILP for \textsc{$k$-Median} in a specific metric space $(\mathcal{L}(C), \swap)$ and with a specific set of points to cluster $\mathcal{L}(C)$. 

  Let $\mathcal{L}(C) = \{ u_1, \dots, u_{m!}\}$ be a set of rankings over a set $C$ of $m$ candidates, and let $k$ denote the size of domain.
  For readability, we define $n = m!$.
  For each ranking $u_i \in \mathcal{L}(C)$, we define a binary variable $y_i$ with the intention that value $1$ indicates that ranking $u_i$ is selected to a solution.
  For each pair of rankings $u_i, u_j \in \mathcal{L}(C)$, we define a binary variable $x_{ij}$ with the intention that value $1$ means that a ranking $u_i$ has $u_j$ as the closest ranking in a solution ($u_j$ is a representative, or cluster center, of $u_i$).
  Let $d_{ij}$ denote the swap distance between rankings $u_i$ and $u_j$, i.e., $d_{ij} = \swap(u_i, u_j)$.
  We introduce the following constraints:
  \begin{align}
    x_{ij}, y_i &\in \{0,1\}, \quad\quad\:\; \forall i, j \in [n] \nonumber\\
    \label{ilp:assign-vote}
    \textstyle\sum_{i \in [n]} x_{ij} &= 1, \quad\quad\quad\quad\quad \forall j \in [n] \\
    \label{ilp:select-vote}
    x_{ij} &\leq y_i, \quad\quad\quad\:\:\; \forall i, j \in [n] \\
    \label{ilp:count-vote}
    \textstyle\sum_{i \in [n]} y_i &= k.
  \end{align}
  Constraint~\eqref{ilp:assign-vote} ensures that each ranking is assigned to exactly one selected ranking.
  Constraint~\eqref{ilp:select-vote} ensures that a vote can only be assigned to another vote if that vote is selected.
  Constraint~\eqref{ilp:count-vote} ensures that exactly $k$ rankings are selected.
  The objective function defined in~\eqref{ilp:obj-vote} minimizes the total cost, i.e., the total swap distance:
  \begin{align}
    \label{ilp:obj-vote}
    \min \sum_{i \in [n]} \sum_{j \in [n]} d_{ij} \cdot x_{ij}.
  \end{align}

\begin{figure}[t]
\centering

    \begin{subfigure}{0.35\columnwidth}
      \includegraphics[width=1\linewidth]{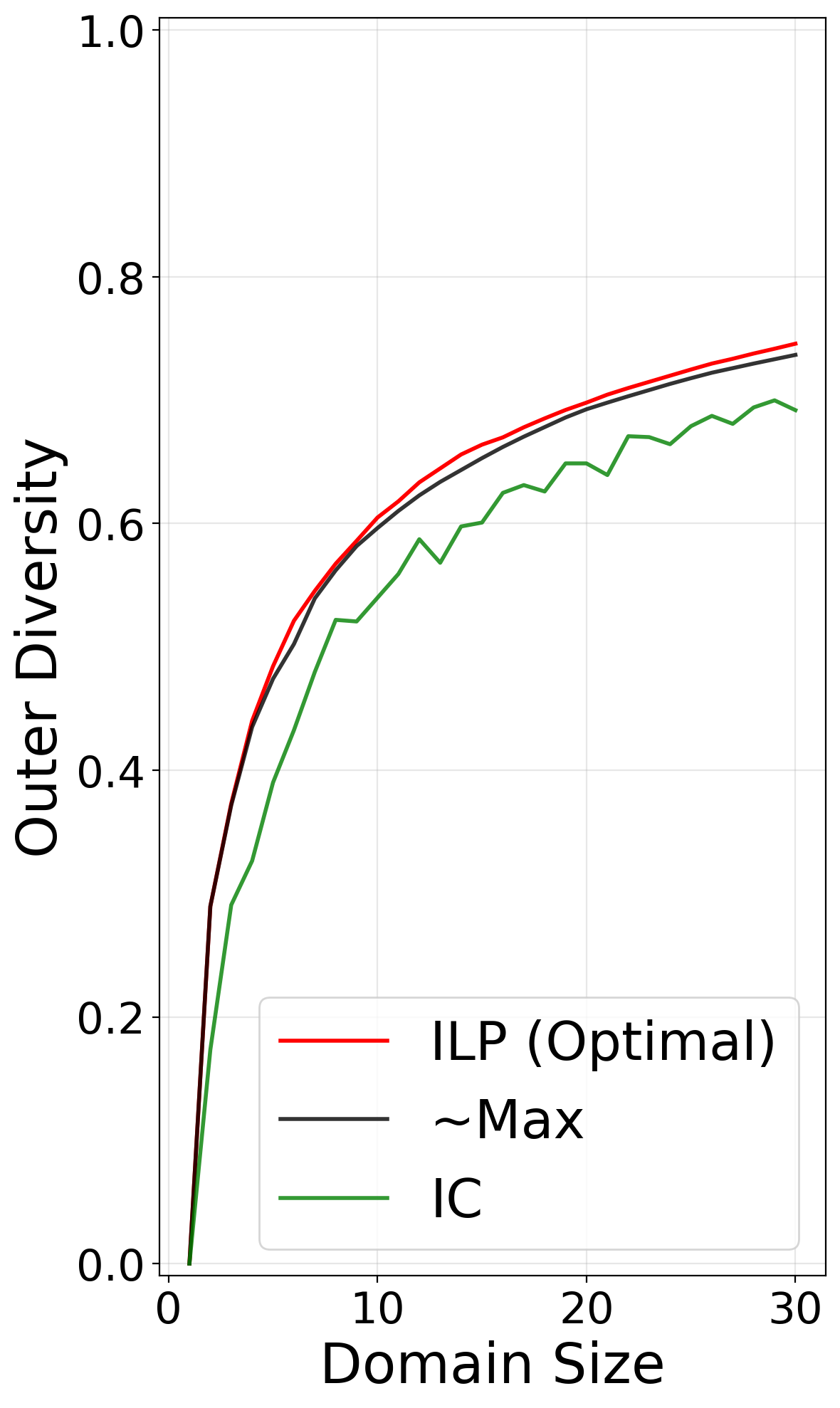}
    \end{subfigure}
    \caption{Comparison of the optimal diversity (red line) and the one achieved by simulated annealing (black line) for $6$ candidates.}
    \label{apdx:fig:out-div-ilp}
\end{figure}

\subsection{Computing Most Diverse Domain}\label{apdx:sec:computing-most-diverse}
Our simulated annealing algorithm operates as follows. We begin with a randomly generated set of rankings. At each iteration, we uniformly at random remove one of the rankings and add one ranking sampled from IC. If the new solution is better than the current one, it is always accepted. Otherwise, it is accepted with probability
\[
P = \exp\left(\frac{E_{\text{new}} - E_{\text{current}}}{T}\right),
\]
where $T$ denotes the current temperature. The initial temperature is set to $T_0 = 0.5$, and it decreases geometrically with a cooling rate of $0.95$ per iteration. Moreover, we perform at most $256$ iterations.

In~\Cref{apdx:fig:out-div-ilp}, we compare the performance of simulated annealing and ILP for the case of six candidates. As shown, the solution found by simulated annealing is nearly optimal. Moreover, note that simply sampling votes from the IC distribution serves as an effective heuristic.

\subsection{Largest Gap in a Domain}

A domain can be considered diverse if it is well distributed across a metric space of all possible rankings, i.e., $(\mathcal{L}(C), \swap)$.
This implies that there are no large gaps between rankings within the domain.
Consequently, it is natural to search for the largest such gap.
To formalize this, we define a decision problem of finding the center of a ball in $(\mathcal{L}, \swap)$ with a given radius that contains no rankings from the given domain $D$.

\begin{definition}
    In the \farthestpermutation problem we are given $D \subseteq \mathcal{L}(C)$ and $r_\far \in \{-1, 0\} \cup \mathbb{N}$.
    We ask if there exists a ranking $f \in \mathcal{L}(C)$ which swap distance to $D$ is at least $r_\far$, i.e.,
    $$r_\far < \min_{v \in D} \swap(v,f) = \swap(D,f).$$
\end{definition}

We emphasize that the definition uses a strict inequality because the goal is to identify a ball that excludes all rankings from $D$.
For example, if $D = \mathcal{L}(C)$ then the only value of $r_\far$ for which a response is YES, is $-1$.
For $D \subset \mathcal{L}(C)$, every $f \in \mathcal{L}(C) \setminus D$ is at distance $1$ from $D$, so $r_\far \geq 0$ in this case.
We can also define an optimization version of \farthestpermutation, searching for a maximum $r_\far$ for which a response is YES.

In the subsequent proofs, we will rely on results concerning the \kemenycenter problem, which may be regarded as a dual problem to \farthestpermutation in the sense that \farthestpermutation looks for a ball of radius $r_\far$ where none of domain rankings are included, but \kemenycenter looks for a ball of radius $r_\cent$ where all of domain rankings are included.

\begin{definition}
    In the \kemenycenter problem we are given $D \subseteq \mathcal{L}(C)$ and $r_\cent \in \{0\} \cup \mathbb{N}$.
    We ask if there exists a ranking $c \in \mathcal{L}(C)$ which swap distance to every element in $D$ is at most $r_\cent$, i.e.,
    $$\max_{v \in D} \swap(v,c) \leq r_\cent.$$
\end{definition}

The duality mentioned can be formalized in a quantitative way as done in \Cref{lem:duality-1center-farthest-perm} which essentially says that for every permutation $x \in \mathcal{L}(C)$, the sum of radii of two balls: 1) a ball with a \farthestpermutation objective and a center in $x$ and; 2) a ball with a \kemenycenter objective and a center in $\rev(x)$, where $\rev(x)$ is a reversed permutation of $x$; is always equal to $\binom{m}{2}-1$, i.e., a maximum distance between two permutations of $m$ elements decreased by $1$.

Formally, for $D \subseteq \mathcal{L}(C)$ and $x \in \mathcal{L}(C)$ we define the radii described above as:
$\fp(D,x) = \swap(D,x) -1$
and
$\kc(D,x) = \max_{v \in D} \swap(v,x)$.

\begin{lemma}\label{lem:duality-1center-farthest-perm}
    For every $D \subseteq \mathcal{L}(C)$ and every $x \in \mathcal{L}(C)$ we have
    $$\fp(D,x) + \kc(D,\rev(x)) = \textstyle{\binom{m}{2}}-1.$$
\end{lemma}
\begin{proof}
    Let us fix $D \subseteq \mathcal{L}(C)$ and $x \in \mathcal{L}(C)$.
    First we observe that, by the $\swap$ distance definition, we have
    $$\swap(v,x) + \swap(v,\rev(x)) = \textstyle{\binom{m}{2}}.$$
    Using it, we obtain a sequence of equalities:
    \begin{align*}
        \fp(D,x) &= -1 + \min_{v \in D} \swap(v,x)\\
                 &= -1 + \min_{v \in D} \Big(\textstyle{\binom{m}{2}} - \swap(v,\rev(x))\Big)\\
                 &= \textstyle{\binom{m}{2}} - 1 -\max_{v \in D} \swap(v,\rev(x))\\
                 &= \textstyle{\binom{m}{2}} - 1 - \kc(D,\rev(x)).
    \end{align*}
    This finishes the proof.
\end{proof}

The lemma implies that hardness of finding a solution to \kemenycenter implies hardness of finding a solution to \farthestpermutation
as well as an additive approximation algorithm with additive loss guarantee of at most $\beta$ for \kemenycenter is also an approximation algorithm for \farthestpermutation with the same additive loss guarantee.
The two results results are formally presented in the following two theorems.
We observe that \Cref{thm:farthest-perm-np-com} directly implies the result stated in \Cref{thm:largest-hole-np-hard}.

\begin{theorem}\label{thm:farthest-perm-np-com}
    \farthestpermutation is NP-complete, even when $|D|=4$.
\end{theorem}
\begin{proof}
    The inclusion in NP is straightforward as this is enough to compute all pairwise distances between a solution and elements of a domain and check if any of them is equal or smaller than $r$.

    In order to show NP-hardness for $|D| = 4$, we will construct a reduction from \kemenycenter which is NP-hard for $|D| = 4$, where all input orders are distinct (see \cite{dwo-kum-nao-siv:c:rank-aggregation} for the original proof and \cite[Theorem 5]{bie-bra-den:j:kemeny-hardness} for its correction).

    Let $D \subseteq \mathcal{L}(C), r_\cent \in \{0\} \cup \mathbb{N}$ be an input of \kemenycenter.\footnote{While the original definition of \kemenycenter allows inputs with non-distinct orders, every such instance can, without loss of generality, be reduced to an equivalent instance consisting solely of distinct orders.}
    We define an input of \farthestpermutation simply by providing the same domain $D$ and $r_\far = {\binom{m}{2}} -1 -r_\cent$.

    Correctness of the reduction directly follows from \Cref{lem:duality-1center-farthest-perm}.
    For completes we provide the two formal implications below.
    
    ($\Rightarrow$) If $(D,r_\cent)$ is a YES-instance of \kemenycenter then there exists $c \in \mathcal{L}(C)$ such that $\kc(D,c) \leq r_\cent$ and we obtain 
    \begin{align*}
      \swap(D,\rev(c)) &\stackrel{\phantom{\textnormal{\Cref{lem:duality-1center-farthest-perm}}}}{=} \fp(D,\rev(c))+1\\
      &\stackrel{\textnormal{\Cref{lem:duality-1center-farthest-perm}}}{=} \textstyle{\binom{m}{2}} -1 -\kc(D,c)+1\\
      &\stackrel{\phantom{\textnormal{\Cref{lem:duality-1center-farthest-perm}}}}{\geq} \textstyle{\binom{m}{2}} -r_\cent > r_\far.    
    \end{align*}
    Therefore, $\rev(c)$ is a solution to the \farthestpermutation instance $(D,r_\far)$.
    
    ($\Leftarrow$)  If $(D,r_\far)$ is a YES-instance of \farthestpermutation then, analogously,
    there exists $f \in \mathcal{L}(C)$ such that $\fp(D,f) \geq r_\far$ and we obtain
    \begin{align*}
      \max_{v \in D} \swap(v,\rev(f)) &\stackrel{\phantom{\textnormal{\Cref{lem:duality-1center-farthest-perm}}}}{=} \kc(D,\rev(f))\\
      &\stackrel{\textnormal{\Cref{lem:duality-1center-farthest-perm}}}{=} \textstyle{\binom{m}{2}} -1 -\fp(D,f)\\
      &\stackrel{\phantom{\textnormal{\Cref{lem:duality-1center-farthest-perm}}}}{\leq} \textstyle{\binom{m}{2}} -1 -r_\far = r_\cent.    
    \end{align*}
    Therefore, $\rev(f)$ is a solution to the \kemenycenter instance $(D,r_\cent)$.
    This finishes the proof.
\end{proof}

The duality between \farthestpermutation and \kemenycenter presented in\break \Cref{lem:duality-1center-farthest-perm} holds for centers of balls at $x$ and $\rev(x)$.
The duality can be also expressed, in \Cref{lem:additive-duality-1center-farthest-perm}, in terms of how far radii of balls with centers at $x$ and $\rev(x)$ are from optimum solutions.
For that we will need a few more definitions.
For a given $D \subseteq \mathcal{L}(C)$,
let $\fp(D)$ be a maximum $r_\far$ for which $(D,r_\far)$ is a YES-instance of \farthestpermutation.
Analogously, let $\kc(D)$ be a minimum $r_\cent$ for which $(D,r_\cent)$ is a YES-instance of \farthestpermutation.

\begin{lemma}\label{lem:additive-duality-1center-farthest-perm}
    For every $D \subseteq \mathcal{L}(C), x \in \mathcal{L}(C)$ and $\beta \in \mathbb{N}$ we have:
    $$\fp(D,x) \geq \fp(D) - \beta \Leftrightarrow \kc(D,\rev(x)) \leq \kc(D) + \beta.$$
\end{lemma}
\begin{proof}
    We fix $D \subseteq \mathcal{L}(C), x \in \mathcal{L}(C)$ and $\beta \in \mathbb{N}$.
    Let $x_\far$ be such that $ \fp(D,x_\far) = \fp(D)$.
    Then, by \Cref{lem:duality-1center-farthest-perm}, we have that
    $\kc(D,\rev(x_\far)) = \kc(D)$.
    We obtain a sequence of equivalent inequalities:
    \begin{align*}
        \fp(D,x) &\geq \fp(D) - \beta\\
        \fp(D,x) &\geq \fp(D,x_\far) - \beta\\
        \textstyle{\binom{m}{2}}-1-\kc(D,\rev(x)) &\geq \textstyle{\binom{m}{2}}-1-\kc(D,\rev(x_\far)) - \beta\\
        -\kc(D,\rev(x)) &\geq -\kc(D) - \beta \\
        \kc(D,\rev(x)) &\leq \kc(D) + \beta,
    \end{align*}
    where the second equivalence comes from \Cref{lem:duality-1center-farthest-perm}. This finishes the proof.
\end{proof}

An algorithm for \farthestpermutation is an (additive) $\beta$-approximation if for an input $D$ it outputs $x \in \mathcal{L}(C)$ such that $\fp(D,x) \geq \fp(D) - \beta$.
An algorithm for \kemenycenter is an (additive) $\beta$-approximation if for an input $D$ it outputs $x \in \mathcal{L}(C)$ such that $\kc(D,x) \leq \kc(D) + \beta$.
The following corollary is an implication of \Cref{lem:additive-duality-1center-farthest-perm}.

\begin{corollary}
    For a given $D \subseteq \mathcal{L}(C)$ it holds that:
\begin{enumerate}
    \item Let $x$ be an output of an additive $\beta$-approximation algorithm for \kemenycenter on $D$.
    Then, $\rev(x)$ is an additive $\beta$-approximate solution to \farthestpermutation on $D$.

    \item Let $x$ be an output of an additive $\beta$-approximation algorithm for \farthestpermutation on $D$.
    Then, $\rev(x)$ is an additive $\beta$-approximate solution to \kemenycenter on $D$.
\end{enumerate}
\end{corollary}

\section{Additional Plots}
Additional histograms of swap distances are shown in~\Cref{apdx:fig:domain_swap_histogram}.

\begin{figure}[t]
  \centering
  \begin{tabular}{ccc}
    \begin{subfigure}{0.28\columnwidth}
      \includegraphics[width=1.16\linewidth]{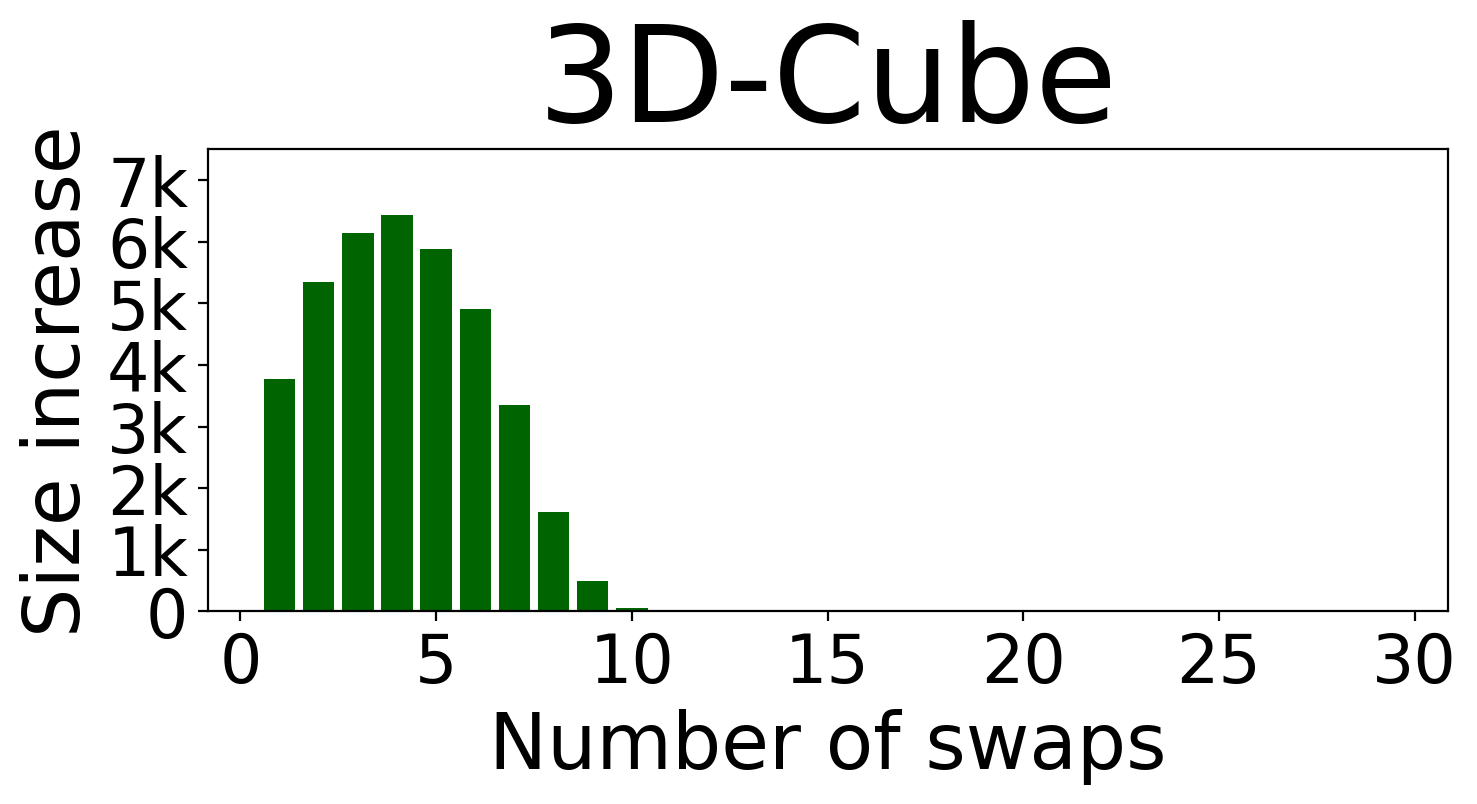}
    \end{subfigure} &
    \begin{subfigure}{0.28\columnwidth}
      \includegraphics[width=1.16\linewidth]{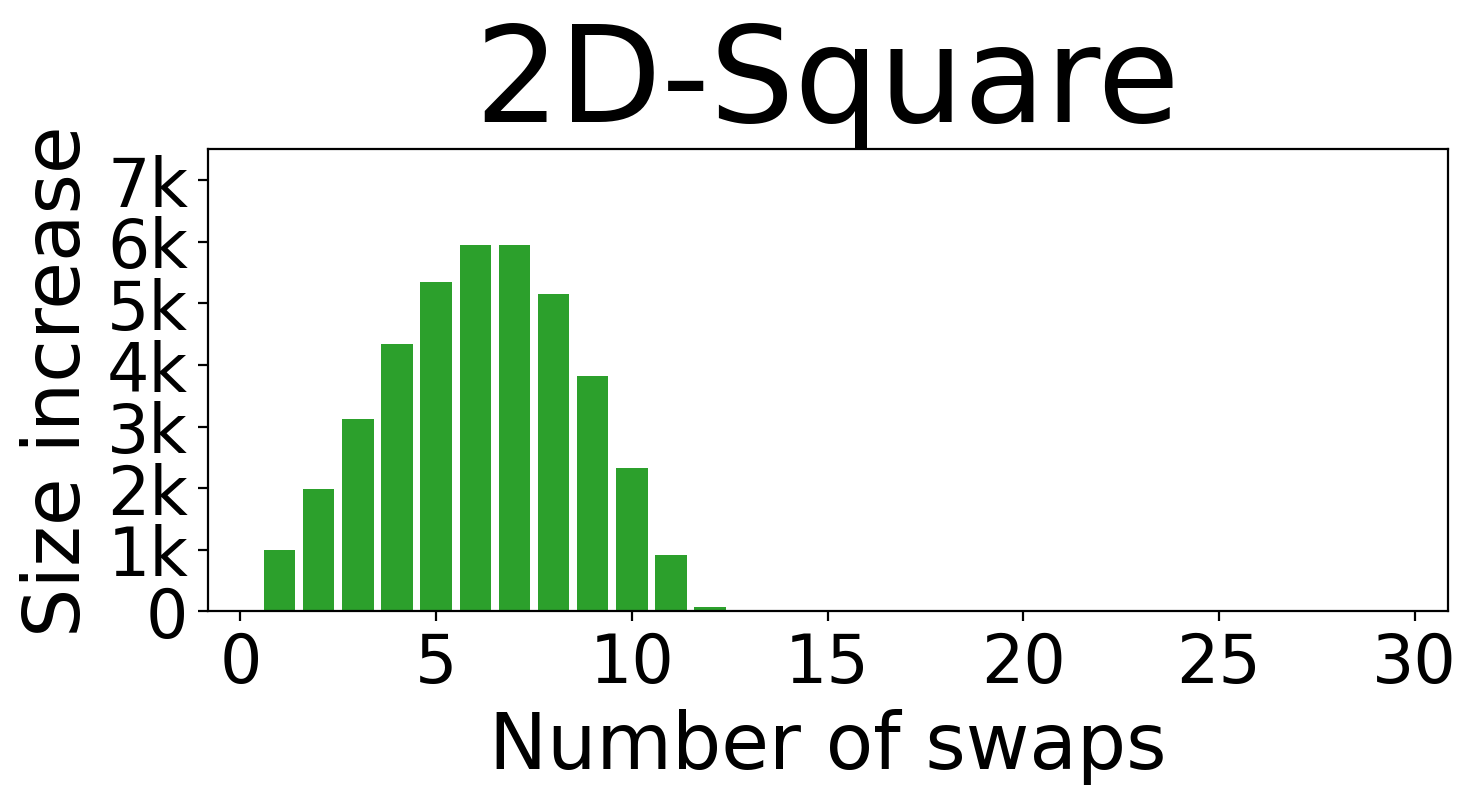}
    \end{subfigure} &
    \begin{subfigure}{0.28\columnwidth}
      \includegraphics[width=1.16\linewidth]{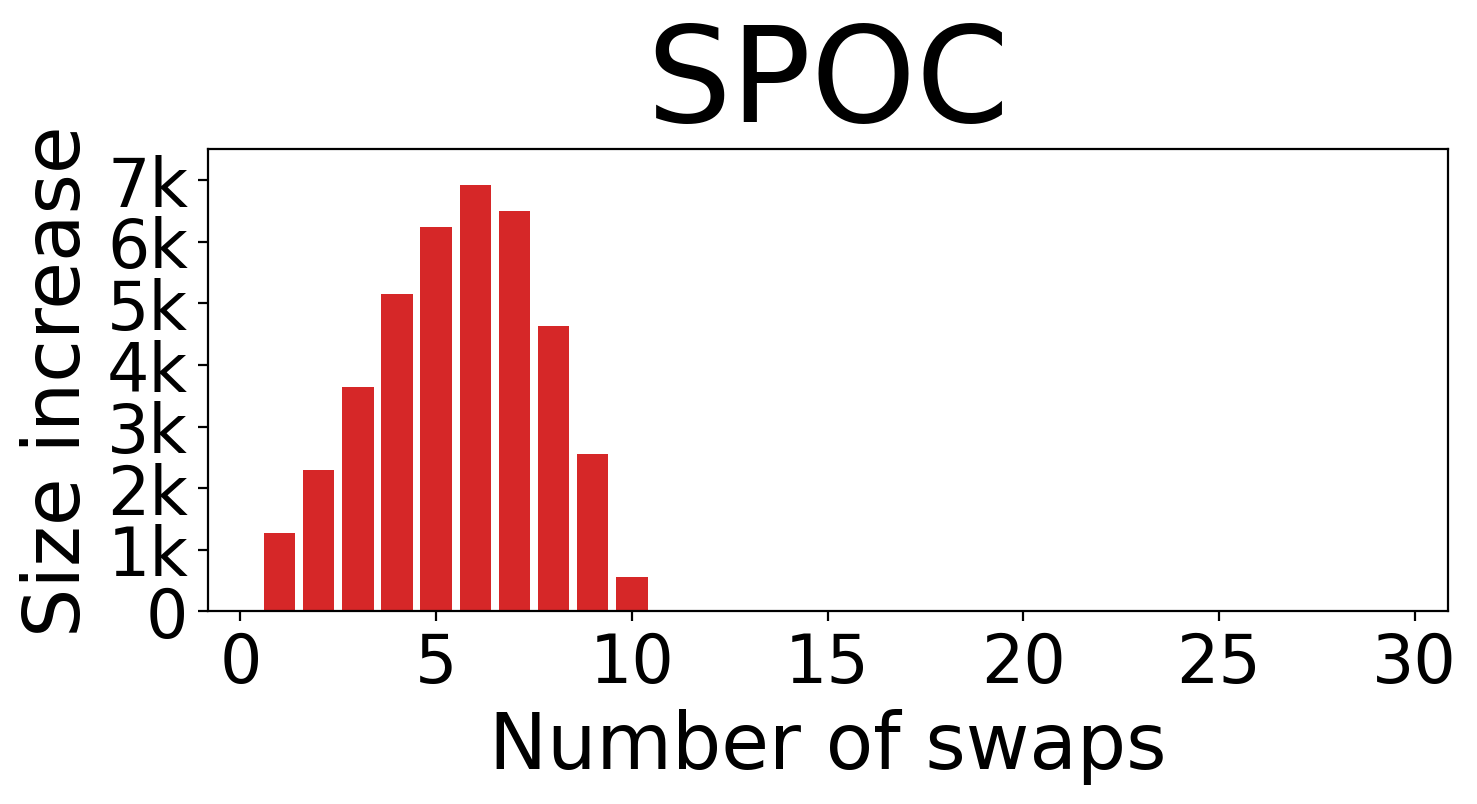}
    \end{subfigure} \\

    \begin{subfigure}{0.28\columnwidth}
      \includegraphics[width=1.16\linewidth]{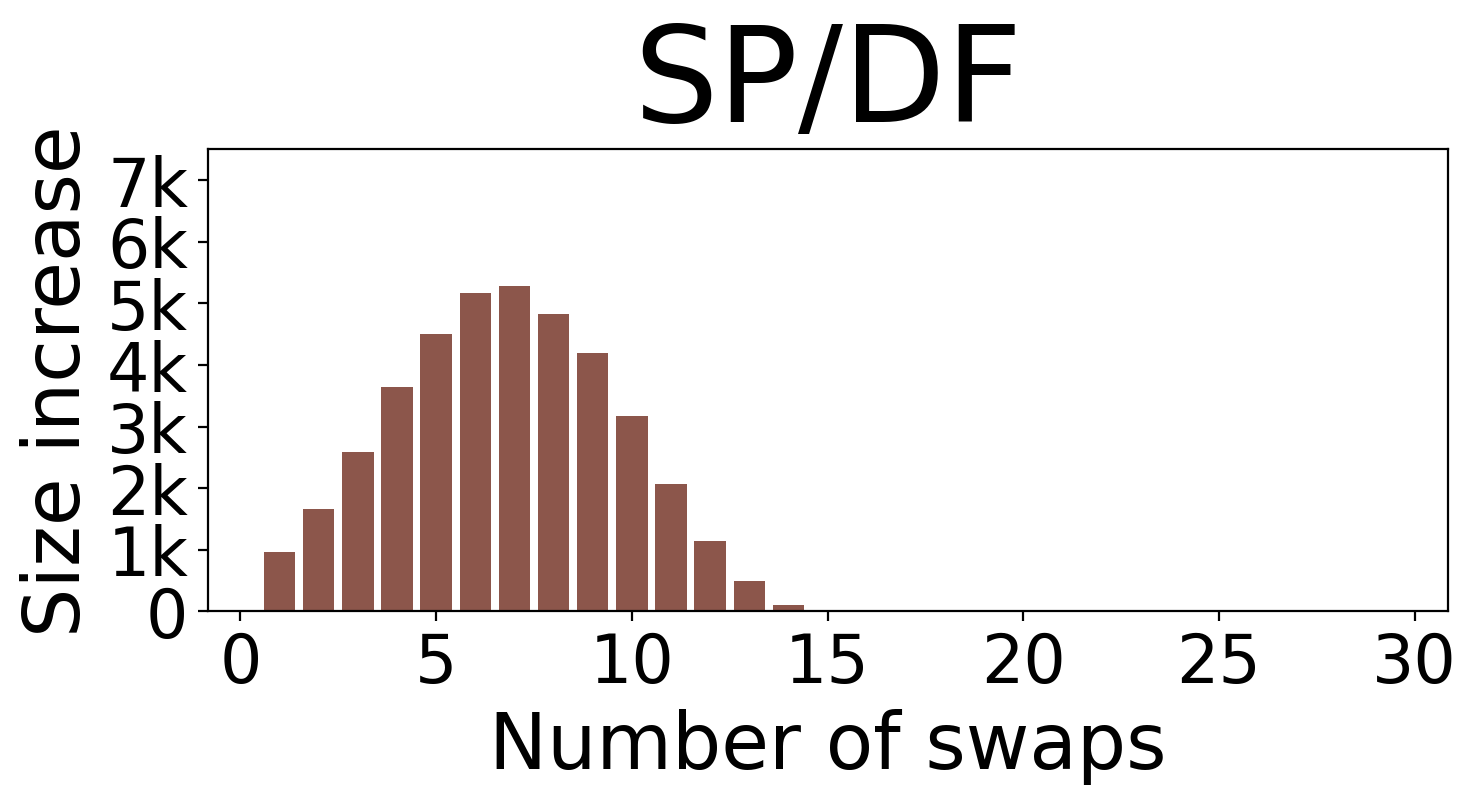}
    \end{subfigure} &
    \begin{subfigure}{0.28\columnwidth}
      \includegraphics[width=1.16\linewidth]{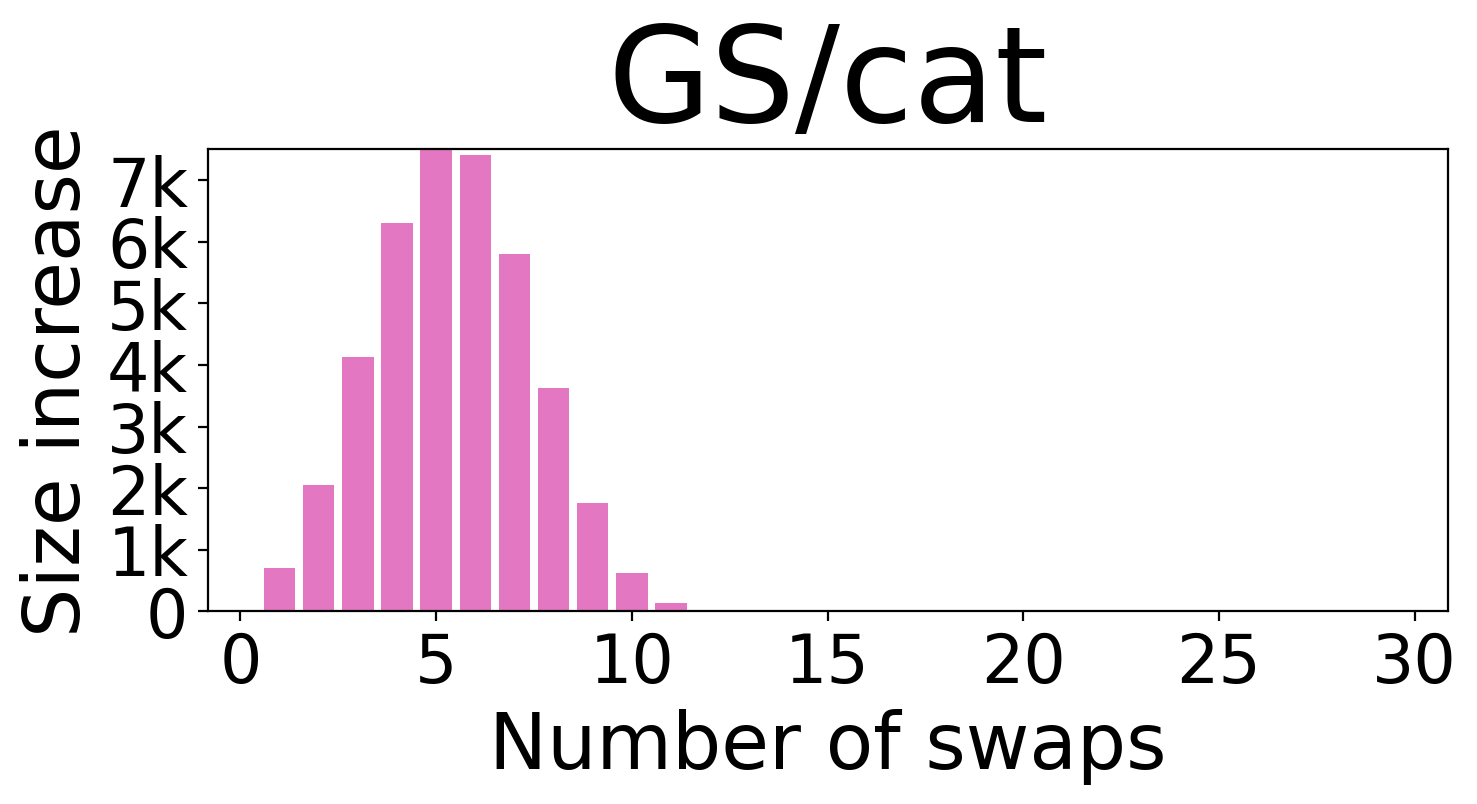}
    \end{subfigure} &
    \begin{subfigure}{0.28\columnwidth}
      \includegraphics[width=1.16\linewidth]{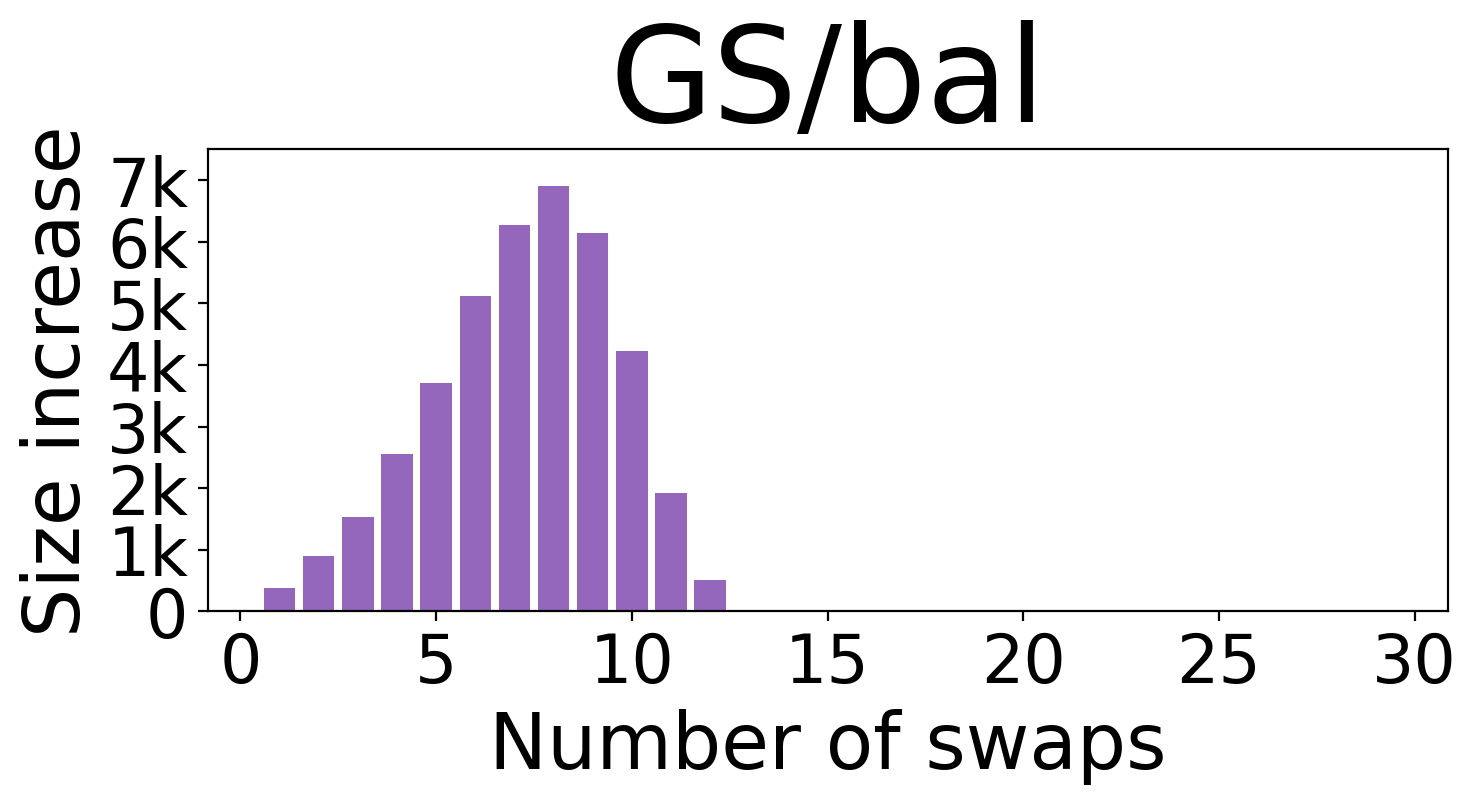}
    \end{subfigure} \\

    \begin{subfigure}{0.28\columnwidth}
      \includegraphics[width=1.16\linewidth]{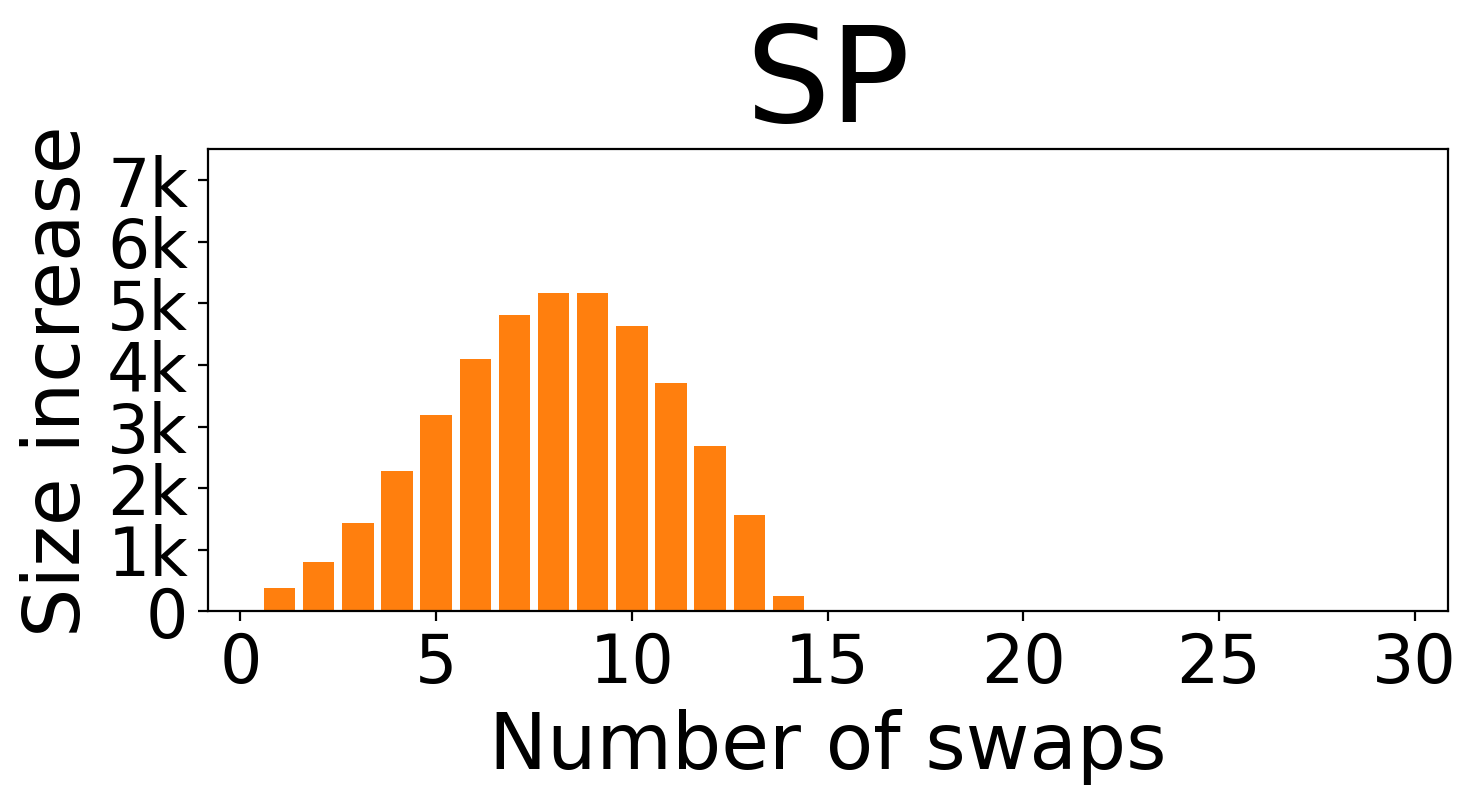}
    \end{subfigure} &
    \begin{subfigure}{0.28\columnwidth}
      \includegraphics[width=1.16\linewidth]{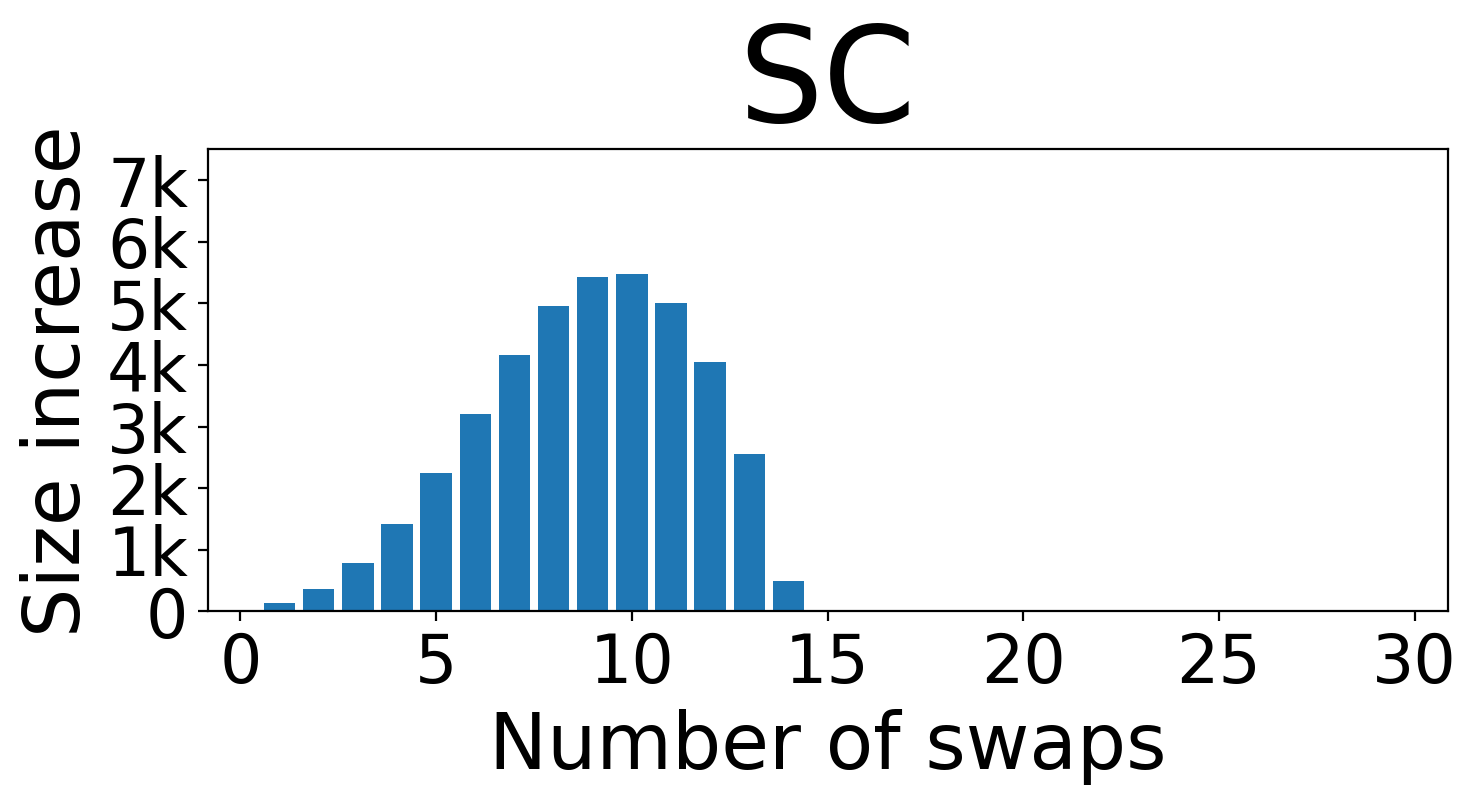}
    \end{subfigure} &
    \begin{subfigure}{0.28\columnwidth}
      \includegraphics[width=1.16\linewidth]{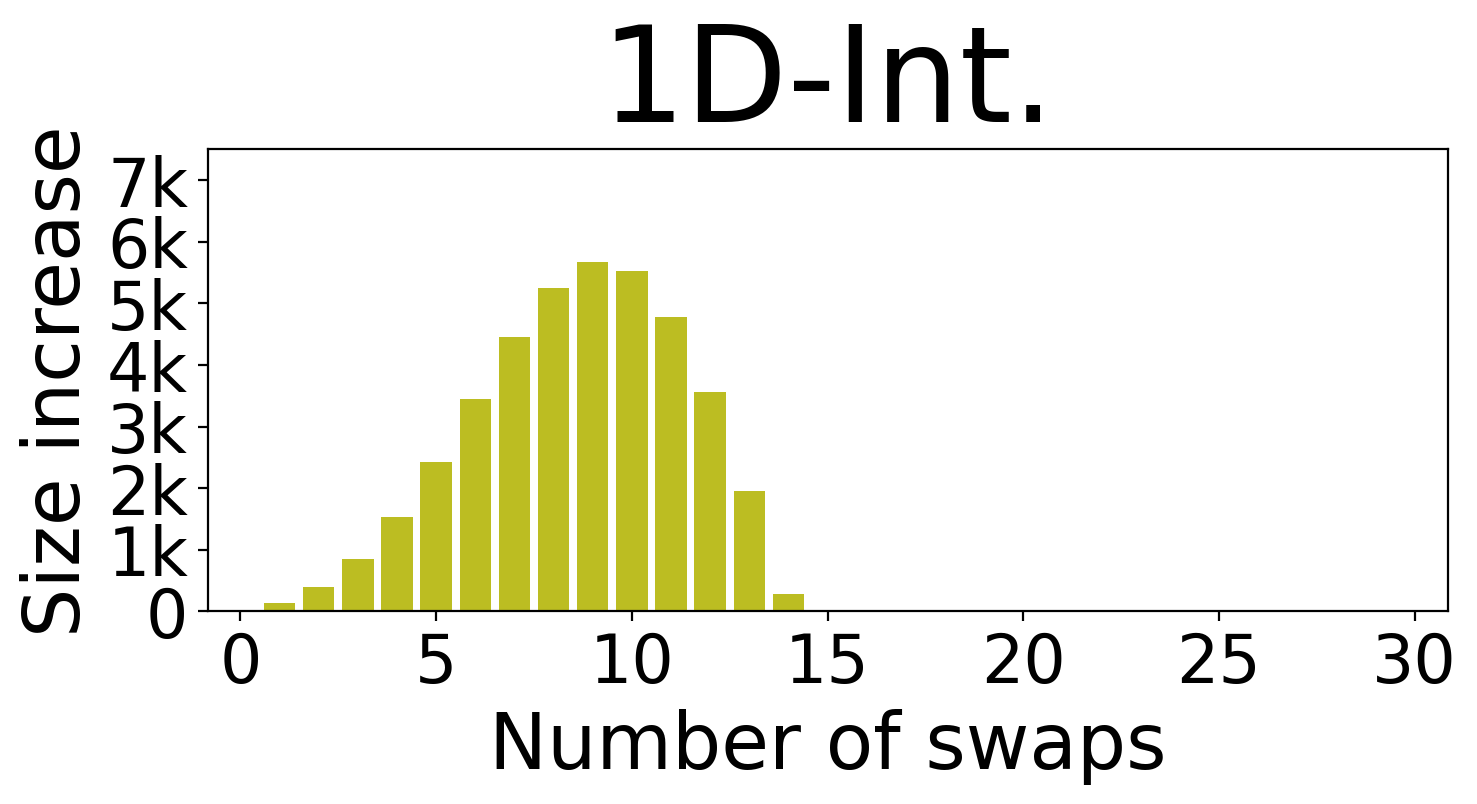}
    \end{subfigure} \\

    \begin{subfigure}{0.28\columnwidth}
      \includegraphics[width=1.16\linewidth]{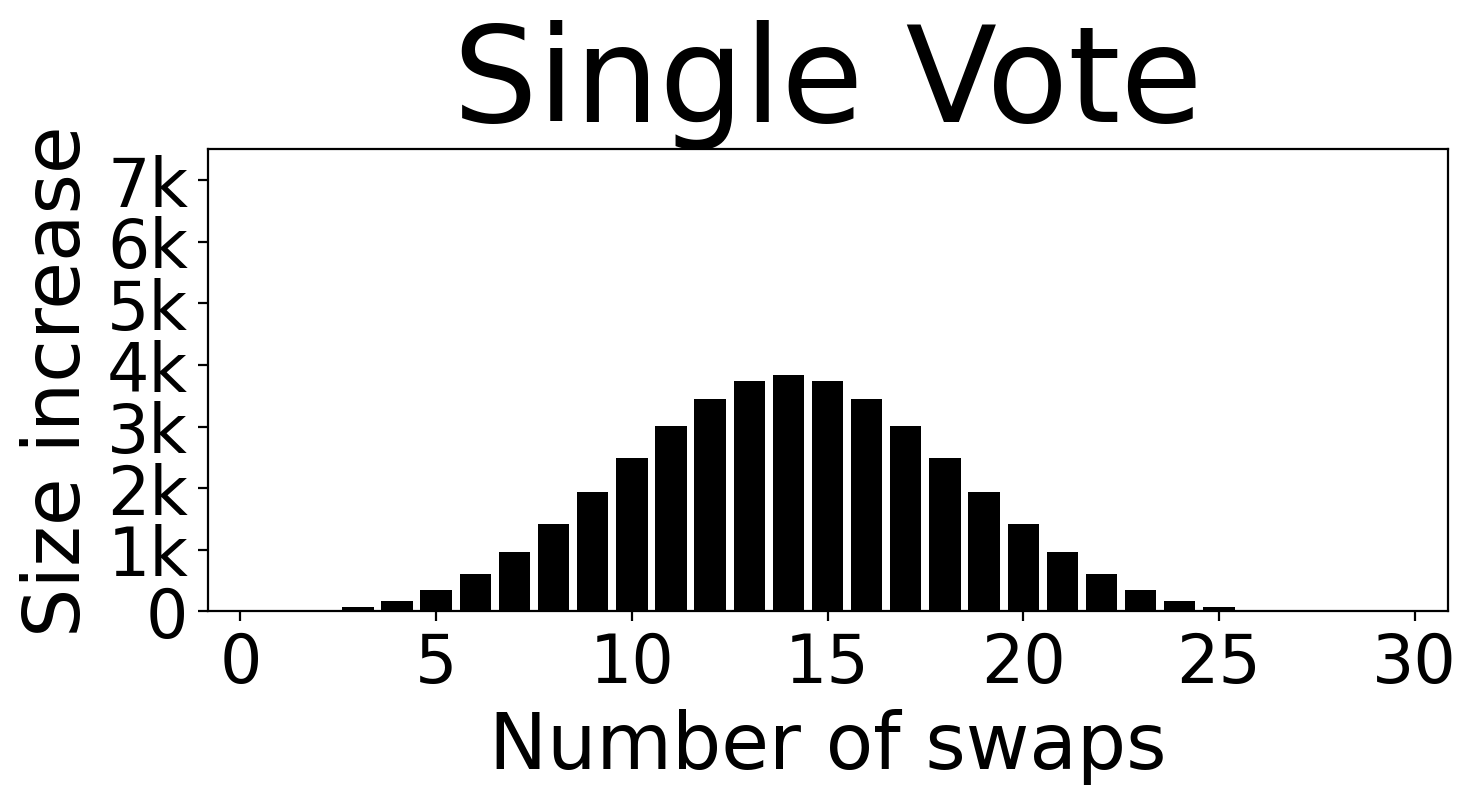}
    \end{subfigure} &
    \begin{subfigure}{0.28\columnwidth}
      \includegraphics[width=1.16\linewidth]{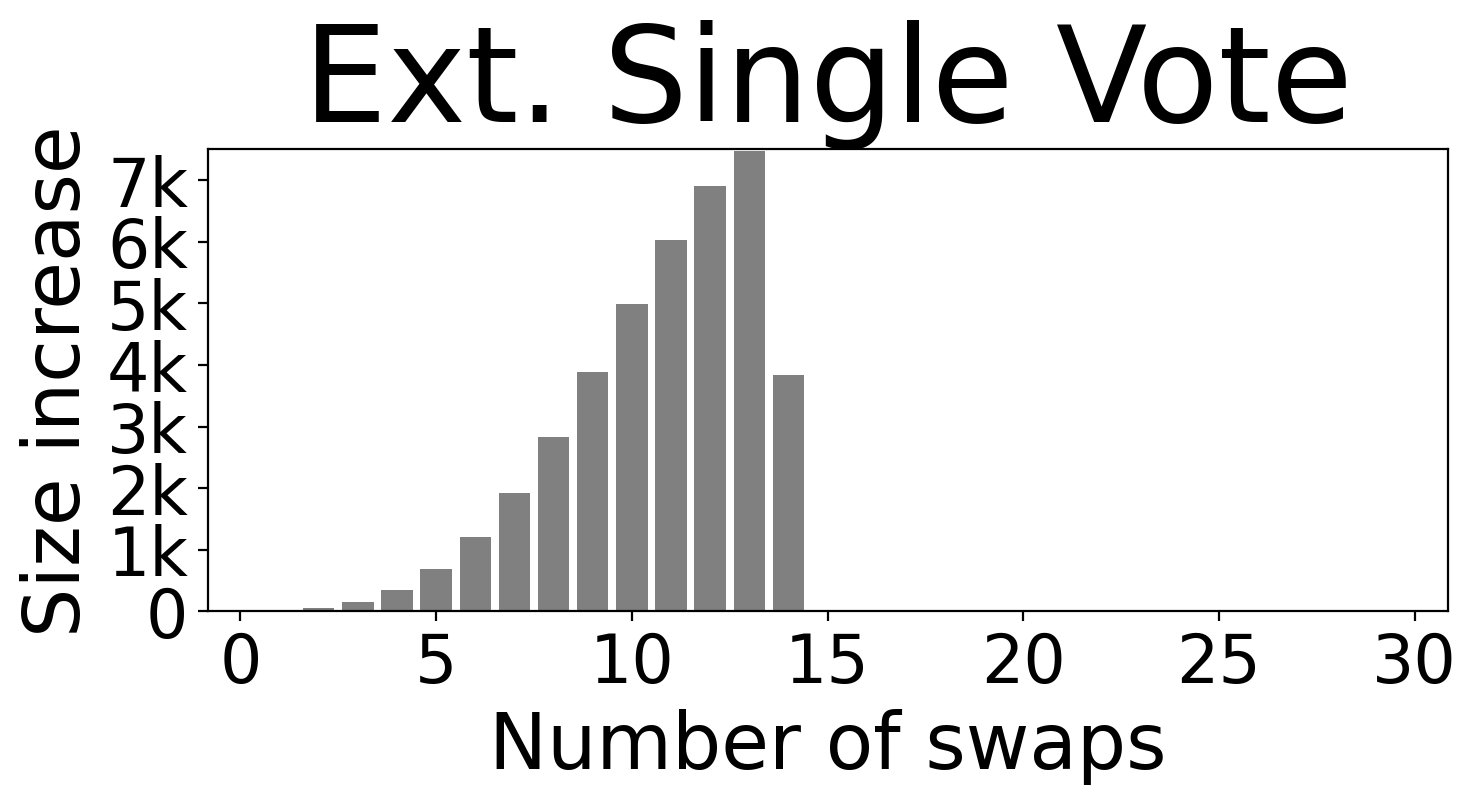}
    \end{subfigure}  &
    \begin{subfigure}{0.28\columnwidth}
      \includegraphics[width=1.16\linewidth]{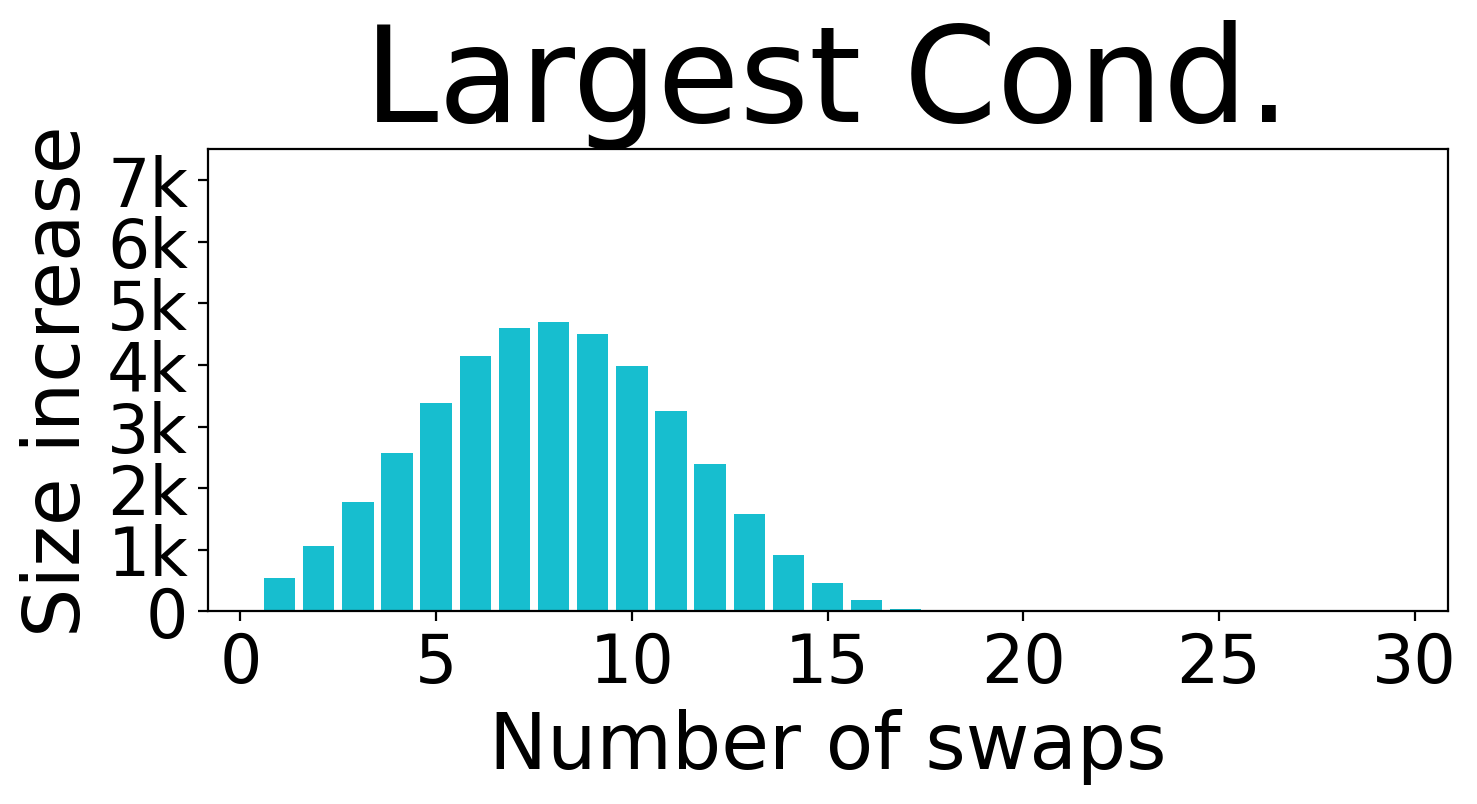}
    \end{subfigure} 
    
  \end{tabular}
  \caption{Histograms of votes at a given swap distance.}
  \label{apdx:fig:domain_swap_histogram}
\end{figure}

\end{document}